\documentclass{lmcs}
\pdfoutput=1

\usepackage{lastpage}
\lmcsdoi{18}{1}{35}
\lmcsheading{}{\pageref{LastPage}}{}{}%
{Mar.~11,~2020}{Mar.~01,~2022}{}

\usepackage[utf8]{inputenc}
\usepackage{microtype}

\newcommand{\dconj}{\mathbin{\dot\wedge}}
\newcommand{\ddisj}{\mathbin{\dot\vee}}
\newcommand{\dexists}{\mathop{\dot\exists}}
\newcommand{\dlift}[2]{{#1}\uparrow{#2}}
\renewcommand{\dlift}[2]{{#2}{\cdot}{#1}}

\newcommand{\dplus}{\mathbin{\dot+}}
\newcommand{\dmult}{\mathbin{\dot\times}}
\newcommand{\dequal}{\mathrel{\dot=}}

\newcommand{\genvar}{x}
\newcommand{\dvar}[1]{\genvar_{#1}}
\newcommand{\dsucc}[1]{1+{#1}}

\newcommand{\dcuto}[2]{\abst{\dvar i}{{#1}\,\dvar{{#2}+i}}}
\newcommand{\dcut}[1]{\dcuto{#1}{1}}

\newcommand{\plus}{+}
\newcommand{\mult}{\times}

\newcommand{\lconj}{\mathbin{\wedge}}
\newcommand{\ldisj}{\mathbin{\vee}}
\newcommand{\lneg}{\mathop{\neg}}

\newcommand{\gsubst}{\sigma}
\newcommand{\gren}{\rho}
\newcommand{\gval}{\nu}
\newcommand{\genv}{\varphi}
\newcommand{\grel}{\mathrel{\bowtie}}

\newcommand{\sem}[2]{\llbracket{#1}\rrbracket_{#2}}
\newcommand{\semb}[3]{\llbracket{#1}\rrbracket_{#2}^{#3}}

\newcommand{\abst}[2]{\lambda{#1}.{#2}}

\newcommand{\diotype}[1]{\mathbb D_{\mathrm{#1}}}

\newcommand{\dform}{\typefont{dio\_form}}

\newcommand{\diorel}{\cstfont{dio\_rel}}

\renewcommand{\dform}{\diotype{form}}
\newcommand{\dsingle}{\diotype{single}}
\newcommand{\dpoly}{\diotype{poly}}
\newcommand{\dpolyz}{\dpoly^\Z}

\newcommand{\delem}{\diotype{cstr}}

\renewcommand{\diorel}{\diotype{rel}}
\newcommand{\diofun}{\diotype{fun}}

\newcommand{\Var}{\typefont V}
\newcommand{\BVar}{\typefont U}

\newcommand{\dsize}[1]{|{#1}|}

\newcommand{\recalg}[1]{\mathcal{A}_{#1}}

\usepackage{csquotes}


\usepackage[numbers]{natbib}
\usepackage{mathpartir}
\usepackage{xspace}
\usepackage{amssymb,amsmath}
\usepackage{stmaryrd}

\definecolor[named]{ACMPurple}{cmyk}{0.55,1,0,0.15}
\definecolor[named]{ACMDarkBlue}{cmyk}{1,0.58,0,0.21}

\usepackage{hyperref}
\usepackage{cleveref}
\usepackage[countsame]{coqtheorem}
\setBaseUrl{{https://uds-psl.github.io/H10-LMCS-v1.1/website/Undecidability.}}

\newcommand{\coqmathlink}[2]{\text{\coqlink[#1]{$#2$}}}

\newcoqtheorem{definition}{Definition}
\newcoqtheorem{proposition}{Proposition}
\newcoqtheorem{lemma}{Lemma}
\newcoqtheorem{theorem}{Theorem}
\newcoqtheorem{corollary}{Corollary}

\newcommand{\defeq}{\mathrel{:=}}
\newcommand{\bnfdef}{\mathrel{::=}}
\let\cdef\defeq

\newcommand\kleenestar{*} %
\newcommand{\toot}{\mathrel\leftrightarrow}
\newcommand{\cfun}{\mathrel\rightarrow}

\newcommand{\coqshorttype}[1]{\mathbb{#1}}

\newcommand{\fin}[1]{\coqshorttype{F}_{#1}}
\newcommand{\Type}{\coqshorttype{T}}
\renewcommand{\Type}{\cstfont{Type}}
\newcommand{\Prop}{\coqshorttype{P}}
\newcommand{\bool}{\coqshorttype{B}}
\newcommand{\nat}{\coqshorttype{N}}
\newcommand{\Z}{\coqshorttype{Z}}
\newcommand{\List}{\coqshorttype{L}\,}
\newcommand{\option}{\coqshorttype{O}\,}

\newcommand{\some}[1]{\lfloor#1\rfloor}

\newcommand{\cnil}{\mathalpha{[\,]}}

\newcommand{\ccons}{\mathbin{::}}
\newcommand{\capp}{\mathbin{{+\mskip-10mu+}}}

\newcommand{\typefont}[1]{\ensuremath{\mathsf{#1}}}

\newcommand{\cstfont}[1]{\text{\tt{#1}}}
\newcommand{\filefont}[1]{\text{\sffamily{#1}}}

\newcommand{\ctactic}{\cstfont}

\newcommand{\True}{\cstfont{True}}
\newcommand{\False}{\cstfont{False}}

\newcommand{\ol}{\overline}

\newcommand{\app}{\mathbin{\ensuremath{+\!\!\!+}}} 
\renewcommand{\app}{\capp}

\newcommand{\Fin}[1]{\mathbb{F}_{#1}}

\providecommand{\M}{}
\renewcommand{\M}[1]{{\textsf{#1}}}

\newcommand{\reducesto}{\mathrel\preceq}

\newcommand{\semstep}{\succ}

\newcommand{\sssstepsG}[5]{{#2}\mathrel{#1}{#3}\mathrel{\semstep^{#4}}{#5}}
\newcommand{\sssstepG}[4]{\sssstepsG{#1}{#2}{#3}{}{#4}}

\newcommand{\ssscomputeG}[4]{\sssstepsG{#1}{#2}{#3}*{#4}}
\newcommand{\sssoutputG}[4]{{#2}\mathrel{#1}{#3}\mathrel\rightsquigarrow{#4}}

\newcommand{\sssterminatesG}[3]{{#2}\mathrel{#1}{#3}\,\mathalpha{\downarrow}}

\newcommand{\dslash}{\mbox{$/\mskip-5mu/$}}

\newcommand{\pslash}[1]{\dslash_{#1}}
\newcommand{\sssstepsX}[1]{\sssstepsG{\pslash{#1}}}
\newcommand{\sssstepX}[1]{\sssstepG{\pslash{#1}}}

\newcommand{\ssscomputeX}[1]{\ssscomputeG{\pslash{#1}}}
\newcommand{\sssoutputX}[1]{\sssoutputG{\pslash{#1}}}
\newcommand{\sssterminatesX}[1]{\sssterminatesG{\pslash{#1}}}

\newcommand{\clength}[1]{\mathalpha{|{#1}|}}

\newcommand{\divides}{\mathrel{|}}
\newcommand{\ndivides}{\mathrel{\nmid}}

\newcommand{\INC}[1]{\llinstr{INC}\ #1}
\newcommand{\DEC}[2]{\llinstr{DEC}\ #1\ #2}

\newcommand{\instr}{\typefont{I}}

\newcommand{\llinstr}[1]{\texttt{#1}}

\newcommand\modeq[3]{{#1}\equiv {#2}~\mathrm{mod}~{#3}}
\renewcommand\modeq[3]{{#1}\equiv {#2}\,[{#3}]}

\newcommand{\bwleq}{\mathrel{\preccurlyeq}}
\newcommand{\binomial}[2]{\mathcal{C}^{#2}_{#1}}

\newcommand{\Zp}[1]{\Z/{#1}\Z}

\newcommand{\FRACTRAN}{\M{FRACTRAN}\xspace}

\newcommand\Hten{\M{H10}\xspace}
\newcommand\Htenz{\ensuremath{\M{H10}_\Z}\xspace}

\begin{document}

\title[Hilbert's Tenth Problem in Coq (Extended Version)]{Hilbert's Tenth Problem in Coq\texorpdfstring{\\}{ }(Extended Version\rsuper*)}
\titlecomment{{\lsuper*}extended version of~\cite{fscdversion}} 

\author{Dominique Larchey-Wendling}
\address{Université de Lorraine, CNRS,  LORIA, Vand{\oe}uvre-l\`es-Nancy, France}
\email{dominique.larchey-wendling@loria.fr}
\author{Yannick Forster}
\address{Saarland University, Saarland Informatics Campus, Saarbrücken, Germany}
\email{forster@cs.uni-saarland.de}{}{}

\begin{abstract}
  We formalise the undecidability of solvability of Diophantine equations, i.e.\ polynomial equations over natural numbers, in Coq's constructive type theory.
  To do so, we give the first full mechanisation of the Davis-Putnam-Robinson-Matiyasevich theorem, stating that every recursively enumerable 
  problem~--~in our case by a Minsky machine~--~is Diophantine.
  We obtain an elegant and comprehensible proof by using a synthetic approach to computability and by introducing Conway's \FRACTRAN language as intermediate layer.
  Additionally, we prove the reverse direction and show that every Diophantine relation is recognisable by $\mu$-recursive functions and give a certified compiler 
  from $\mu$-recursive functions to Minsky machines.
\end{abstract}

\maketitle

\section{Introduction}

Hilbert's tenth problem (\Hten) was posed by David Hilbert in 1900 as part of his famous 23 problems~\cite{hilbert1902mathematical} and asked for the ``determination of the solvability of a Diophantine equation\rlap.''\enspace
A Diophantine equation\footnote{Named after the Greek mathematician Diophantus of Alexandria, who started the study of polynomial equations in the third century.}
 is a polynomial equation over natural numbers (or, equivalently, integers) with constant exponents, e.g.\ $x^2 + 3 z = yz + 2$.
When Hilbert asked for ``determination\rlap,'' he meant, in modern terms, a decision procedure, but computability theory was yet several decades short of being developed.

The first undecidable problems found by Church, Post and Turing were either native to mathematical logic or dependent on a fixed model of computation.
\Hten, to the contrary, can be stated to every mathematician and its formulation is independent from a model of computation.
Emil Post stated in 1944 that \Hten ``begs for an unsolvability proof''~\cite{post1944recursively}.
From a computational perspective, it is clear that \Hten is recursively enumerable (or \textit{recognisable}), meaning there is an algorithm 
that halts on a Diophantine equation iff it is solvable.

Post's student Martin Davis conjectured that even the converse is true, i.e.\ that every recognisable set is also Diophantine.
More precisely, he conjectured that if $A \subseteq \nat^k$ is recognisable then $(a_1, \dots, a_k) \in A \toot \exists x_1\dots x_n,P(a_1, \dots, a_k,x_1,\dots,x_n) = 0$ holds for some polynomial $P$ in $k + n$ variables.
He soon improved on a result by Gödel~\cite{godel1931formal} and gave a proof of his conjecture, however requiring up to one bounded universal quantification~\cite{davis1953arithmetical}: 
$(a_1, \dots, a_k) \in A \toot\allowbreak \exists z, \forall y < z, \exists x_1 \ldots x_n,P(a_1, \dots, a_k,x_1,\dots,x_n,y,z) = 0$.
Davis and Putnam~\cite{davis1959computational} further improved on this, and showed that, provided a certain number-theoretic assumption holds, every recognisable set is \textit{exponentially} Diophantine, meaning variables are also allowed to appear in exponents.
Julia Robinson then in 1961 modified the original proof to circumvent the need for the assumption, resulting in the DPR theorem~\cite{davis1961decision}, namely that every recognisable set is exponentially Diophantine.
Due to another result from Robinson~\cite{robinson1952}, the gap now only consisted of proving that there is a Diophantine equation exhibiting exponential growth.
In 1970, Yuri Matiyasevich showed that the Fibonacci sequence grows exponentially while being Diophantine, closing the gap and finishing the proof of the theorem nowadays called \textit{DPRM theorem}, ultimately establishing that exponentiation is Diophantine itself~\cite{matijasevic1970enumerable} (known as ``Matiyasevich's theorem'').

Even the most modern and simpler proofs of the DPRM theorem
still require many preliminaries and complicated number-theoretic ideas, for an overview see~\cite{matiyasevich2016martin}.
We formalise one such proof as part of our ongoing work on a 
\href{https://github.com/uds-psl/coq-library-undecidability}{library of undecidable problems~\cite{forstertowards,PSLSyntCT}} in the proof assistant Coq~\cite{Coq}.
Since \Hten is widely used as a seed~\cite{dudenhefner2018simpler, goldfarb1981undecidability} 
for showing the undecidability of problems using \textit{many-one reductions}, this will open further ways of extending the library.
Given that our library already contains a formalisation of Minsky machines~\cite{FLW19}, we follow the approach of Jones and Matijasevi\v{c}~\cite{JonesM84}, who use register machines, being very well-suited since they already work on numbers.
They encode full computations of register machines as Diophantine equations in one single, monolithic step.
To make the proof more tractable for both mechanisation and explanation, we factor out an intermediate language, John Conway's \FRACTRAN~\cite{Conway1987}, which can simulate Minsky~machines.

We first introduce three characterisations of Diophantine equations over natural numbers, namely \textit{Diophantine logic} $\M{DIO\_FORM}$ (allowing to connect basic Diophantine equations with conjunction, disjunction and existential quantification), \textit{elementary Diophantine constraints} $\M{DIO\_ELEM}$ (a finite set of constraints on variables, oftentimes used for reductions~\cite{dudenhefner2018simpler,goldfarb1981undecidability}) and \textit{single Diophantine equations} $\M{DIO\_SINGLE}$, including parameters,  as described above.
$\M{H10}$ then asks about the solvability of single Diophantine equations with no parameters.

\newcommand{\figlink}[3]{\item[\normalfont\setCoqFilename{#1}\coqmathlink{#2}{\M{#3}}]}

\begin{figure}
\begin{description}
\figlink{PCP.PCP}{PCP}{PCP} Post correspondence problem, see e.g.~\cite{FLW19}. (matching)
\figlink{MinskyMachines.MM}{MM_HALTING}{MM} Given $n:\nat$, a Minsky machine $P:\List{\instr_n}$ with $n$ registers,
  and $\vec v:\nat^n$, does $(1,P)$ terminate from input state $(1,\vec v)$? (termination)
\figlink{FRACTRAN.FRACTRAN}{FRACTRAN_HALTING}{FRACTRAN} Given a regular \FRACTRAN program $Q:\List{(\nat\times\nat)}$
  and an input state $s$, does $Q$ terminate from input state $x$? (termination)
\figlink{H10.FRACTRAN_DIO}{DIO_LOGIC_SAT}{DIO\_FORM} Given a Diophantine logic formula $A:\dform$ and a valuation
  $\gval:\Var\cfun\nat$, does $\sem A\gval$ hold? (satisfaction)
\figlink{H10.FRACTRAN_DIO}{DIO_ELEM_SAT}{DIO\_ELEM} Given a list $l:\List\delem$ of elementary Diophantine constraints 
  and a valuation $\gval:\Var\cfun\nat$, does there exist $\genv:\BVar\cfun\nat$ such that 
  $\forall c\in l,\, \semb c\gval\genv$? (simultaneous satisfiability)
\figlink{H10.FRACTRAN_DIO}{DIO_SINGLE_SAT}{DIO\_SINGLE} Given a single Diophantine equation $p\dequal q:\dsingle(\nat,\nat)$
  and a valuation $\gval:\nat\cfun\nat$, does there exist $\genv:\nat\cfun\nat$ s.t.\ 
  $\semb p\gval\genv=\semb q\gval\genv$? (solvability)
\figlink{H10.H10}{H10}{H10} Given a single Diophantine equation $p\dequal q:\dsingle(\fin n,\fin 0)$
  (over $\nat$ with possibly $n$ variables but no parameters), does it have a solution in $\nat$? (solvability)
\figlink{MuRec.recalg}{MUREC_HALTING}{$\mu$-rec} Given $n:\nat$, an $n$-ary $\mu$-recursive function $f :\recalg n$,
  and $\vec v:\nat^n$, does $\vec v$ belong to the domain of $f$? (termination)
\figlink{H10.H10Z}{H10Z}{H10${}_{\mathbb{Z}}$} Given a single Diophantine equation $p\dequal q:\dsingle(\fin n,\fin 0)$
  (over $\mathbb{Z}$ with possibly $n$ variables but no parameters), does it have a solution in $\mathbb{Z}$? (solvability)
\end{description}
\caption{\label{fig:problems}Summary description of some decision problems.}
\end{figure}

Technically, the reduction chain to establish the unsolvability of \Hten starts at the halting problem for single-tape Turing machines $\M{Halt}$, reduced to the Post correspondence problem $\M{PCP}$ in~\cite{FHS18}.
In previous work~\cite{FLW19} we have reduced $\M{PCP}$ to a specialised halting problem for Minsky machines, which we use here in a slightly generalised form as $\M{MM}$.
We then reduce Minsky machine halting to \FRACTRAN termination.
\FRACTRAN is very natural to describe using polynomials, and the encoding 
does not rely on any complicated construction.
The technical difficulty then only lies in the Diophantine encoding of the reflexive-transitive closure of a relation 
which follows from the direct elimination of bounded universal quantification,
given that the proof in~\cite{Matiyasevich1997} involves no detour via models of computation.
In total, we obtain the following chain of reductions\rlap,\footnote{A \emph{many-one reduction $P\reducesto Q$}, later defined formally in this section, is a computable
function from problem $P$ to problem $Q$ that maps instances of $P$ into instances of $Q$, pereserving both validity and invalidity.} 
establishing the undecidability of \Hten and it's many-one interreducibility with several decision problems:
  \[\footnotesize\M{Halt}\reducesto \M{PCP}\reducesto \M{MM}\reducesto \FRACTRAN\reducesto \M{DIO\_FORM}\reducesto \M{DIO\_ELEM}\reducesto \M{DIO\_SINGLE}\reducesto \M{H10} \reducesto \M{$\mu$-rec} \reducesto \M{MM} \]
where Fig.~\ref{fig:problems} lists high-level descriptions of these problems.
Furthermore, we prove that $\M{H10} \reducesto \M{H10}_{\mathbb{Z}}$ via Lagrange's theorem.
In the present paper, we focus on explaining this factorisation of the proof and give some details for the different stages.
While we contribute Coq mechanisations of Matiyasevich's theorem and the elimination of bounded universal quantification, 
we treat them mainly as black-boxes and only elaborate on their challenging formalisation rather 
than the proofs themselves (see Section~\ref{sec:expo_buq}).

To the best of our knowledge, we are the first to give a \emph{full verification} of the DPRM theorem and the undecidability of Hilbert's tenth problem in a proof assistant.
We base the notion of recognisability in the DPRM theorem on Minsky machines.

When giving undecidability proofs via many-one reductions, it is critical to show that all reduction functions are actually computable.
We could in theory verify the computability of all functions involved using an explicit model of computation.
In pen-and-paper proofs, this approach is however almost never used, because implementing high-level mathematical transformations as provably correct low-level programs is a daunting task. 
Instead, we rely on a synthetic approach~\cite{FHS18,0002KS19,FLW19} based on the computability of all functions definable in Coq's constructive type theory, which is closer to the practice of pen-and-paper proofs.
In this approach, a problem $P$ is considered undecidable if there is a reduction from an obviously undecidable problem, 
e.g.\ $\M{Halt} \preceq P$.

The axiom-free Coq formalisation of all the results in this paper is available online~\cite{zenodolibrary}
and the main lemmas and theorems in the pdf version of the paper are hyper-linked 
with the html version of the source code
at
\begin{center}
{\color{ACMDarkBlue}\url{https://github.com/uds-psl/coq-library-undecidability/tree/H10-LMCS-v1.1}}
\end{center}
Starting from our already existing library which included most of the Minsky machine code~\cite{FLW19},
the additional code for proving the undecidability of \Hten and the DPRM theorem consists of about 8k loc including 3k loc for
Matiyasevich's results alone, together with a 4k loc addition to our shared libraries; 
see Appendix~\ref{code_comments_appendix} for more details.
The paper itself can be read without in-depth knowledge of Coq or type theory.

\subsection{Contribution}
This paper is an extended journal version of a conference paper~\cite{fscdversion} which, besides a full mechanisation of the DPRM theorem, contributed a novel refactoring of the proof via \FRACTRAN improving the explainability of the DPRM theorem.
Compared to the conference version, we contribute mechanised proofs showing that
\begin{itemize}
\item $\M{H10}$ reduces to solvability of Diophantine equations over integers, reduction 
      obtained via a low-level implementation of Lagrange's theorem (\Cref{sec:lagrange});
\item Diophantine relations are recognisable by $\mu$-recursive algorithms (\Cref{sec:murec});
\item $\mu$-recursively recognisable relations are $\M{MM}$-recognisable  (\Cref{sec:murec}), 
      thereby proving that all considered problems are in the same many-one reduction class;
\item $\mu$-recursive algorithms can be simulated in the weak call-by-value $\lambda$-calculus (\Cref{sec:wcbv}), thereby proving that all considered problems are in the same many-one reduction class as most problems in the Coq library of undecidable problems~\cite{PSLSyntCT}.
\end{itemize}

Apart from the new results, we have simplified the account of Diophantine logic considerably and expanded various explanations of proofs.

\subsection{Preliminaries}

For the text, while we cannot completely avoid it, we will try to minimize reliance on type theoretic language
and notations. We write $\Prop$ for the (impredicative) type of propositions 
and $\Type$ for the (predicative hierachy of) types of Coq. When $X$ and $Y$ are types,
we write $X\to Y$ for functions from $X$ to $Y$.\footnote{In the case where $X$ and $Y$ are 
moreover propositions, the type/proposition $X\to Y$ is understood as $X$ implies $Y$,
inhabited with functions mapping proofs of $X$ into proofs of $Y$; the type theory of Coq
fully implements the Curry-Howard isomorphism.}
We write $x\mult y$ or $x\cdot y$ for multiplication of natural numbers $x,y : \nat$ and we will leave out the symbol where convenient.
We write $\List X$ for the type of \emph{lists} over $X$ and $l \capp l'$ for the \emph{concatenation} of two lists.
We write $X^n$ for \emph{vectors} $\vec v$ over type $X$ with length $n$, and $\fin n$ for the \emph{finite type} with 
exactly $n$ elements. For $p:\fin n$, we write $\vec v_p$ for the $p$-th component of $\vec v:X^n$.
Notations for lists are overloaded for vectors.

\newcommand\inl{\mathsf{inl}~}
\newcommand\inr{\mathsf{inr}~}
We write $\option X$ for the type of \emph{options} over $X$ with elements $\some x$ where $x : X$ and $\emptyset$.
We write $X + Y$ for the type-theoretic sum of types $X$ and $Y$, consisting of $\inl x$ for $x : X$ and $\inr y$ for $y : Y$.
For a list $l : \List X$, $l_n : \option X$ denotes the $n$-th value in $l$ if it exists.

If $P:X\cfun\Prop$ is a predicate (on $X$) and 
$Q:Y\cfun\Prop$ is a predicate, we write $P \reducesto Q$ if there is a function $f : X \to Y$ 
s.t.\ $\forall x : X,\,P\,x \leftrightarrow Q(f\,x)$, i.e.\ a \emph{many-one reduction} from $P$ to $Q$.
In the synthetic approach~\cite{FHS18,0002KS19,FLW19}, the computability of
the reduction~$f$ is automatically ensured because $f$ is typeable in Coq without relying on any axiom. 
    
\section{Diophantine Relations}

Diophantine relations are composed of polynomials over natural numbers.
There are several equivalent 
approaches to characterise these relations and oftentimes, the
precise definition is omitted from papers. Basically, one can
form equations between polynomial expressions and then combine these 
with conjunctions, disjunctions, and existential quantification\rlap.\footnote{Universal quantification or
negation are not accepted as is.}\enspace 
For instance, these
operations are assumed as Diophantine producing operators
in e.g.~\cite{JonesM84, matijasevic1970enumerable,Matiyasevich1997,Matiyasevich2000}.
Sometimes Diophantine relations 
are restricted to a single polynomial equation.
Sometimes the exponentiation function $x,y\mapsto x^y$ 
is assumed as Diophantine~\cite{JonesM84}.
To complicate the picture, Diophantine relations might equivalently
range over $\mathbb Z$ (instead of $\nat$) but expressions
like $x^y$ implicitly assume that $y$ never gets a negative value.

Although seemingly diverging, these approaches are not contradictory
because in the end, they characterise the same class of relations
on natural numbers. However, mechanisation does not allow for such
implicit assumptions. To give some mechanisable structure to some of 
these  approaches, we propose three increasingly restricted
characterisations of Diophantine relations: \emph{Diophantine logic}, 
\emph{elementary Diophantine constraints} and \emph{single
Diophantine equations}, between which we provide computable transformations 
in Sections~\ref{dio_elem_sect} and~\ref{dio_single_sect}.
In Section~\ref{sec:lagrange},
we also elaborate on the case of polynomials over $\Z$, 
i.e.\ we give an undecidability proof for Hilbert tenth problem over integers.
But before turning to formal definitions, we motivate our approach
for the automated analysis and recognition of Diophantine shapes.

\subsection{Diophantine Shapes}

We introduce the meta-level notion of \emph{Diophantine shape}. On purpose,
this notion does not have a precise formal definition because it is a
\emph{dynamically evolving property} of meta-level expressions
that upgrades itself as more and more closure results are proved about those shapes. 

Of course, we cannot rely on some blurry notion to formally
prove theorems about Diophantine relations. So at some point, we have to
choose one (or several) formal representation(s) of Diophantine relations.
Irrelevant to the actual syntax we finally pick up for the formal description
of Diophantine relations, we call them \emph{object-level representations}. 

The problem we face is the following: how can we minimize the work we have to do 
to actually build an object-level witness corresponding to a given
meta-level expression representing a Diophantine relation. 
Directly manipulating object-level 
syntax is far from the ideal way to proceed for a lazy Coq 
programmer\rlap,\footnote{being able to defer repetitive work to 
computers is critical to the successful completion of mechanizations.}
and indeed, this empirical lesson can be learned the painful way.

\smallskip

Let us illustrate this on the following complex example of meta-level expression 
$$x = z^{x+k}\lconj \forall y,\, y<k \to \exists u\,v, \modeq{u}{2v}p \lconj p = 2^y 
\lconj \modeq{\binomial u z}x v.$$
After a quick analysis of its structure, it appears that before being able to actually
establish that this is a Diophantine relation, we would probably have to show 
that polynomials, identities, arithmetic congruences, the exponential function,
binomial coefficients, conjunction, existential
quantification and bounded universal quantification are all Diophantine 
admissible, hence to give \emph{procedures} to derive object-level
representations for all these meta-level constructions. 
This already amounts to significant work.
But once this is done, we want to avoid both the hand-building of 
the object-level witness for the above expression, and the proofs
that it correctly reflects its semantics.
We essentially require our framework to be able to automatically 
combine those procedures and build a provably correct witness for us. 

To summarise, we aspire at the formal definition of an object-level representation 
and at the same time, at avoiding its direct manipulation. 
This is where the dynamic notion of Diophantine shape plays a central role. 
At first, there are very few basic Diophantine shapes, typically constants, 
variables,
addition, multiplication, equality. But at some point, we would e.g.\ have a 
result stating that if the expressions $f$ and $g$ have a Diophantine 
shape, then so does the expression $f^g$, that is the Diophantine 
admissibility of the exponential function, nowadays called
Matiyasevich's theorem. 
Critically, shapes can  be added
dynamically as they are proved admissible as opposed to be carved in 
the stone of a given object-level syntax. We now describe how to do this in a 
successful way using some of the automation provided by Coq.

\subsection{Diophantine Logic}

\label{dio_logic_sect}

\setCoqFilename{H10.Dio.dio_logic}

We define the type $\dform$ of Diophantine formul{\ae} for the
abstract syntax of Diophantine logic. An atomic Diophantine logic formula
is just expressing basic atomic identities between variables like
 $\dvar i \dequal \dvar j\dplus \dvar k$
or $\dvar i \dequal \dvar j\dmult \dvar k$
and we combine those with binary disjunction, binary conjunction, and existential quantification.
\[
A,B : \coqlink[dio_formula]{\dform} \bnfdef 
{\dvar i \dequal n}
\mid
{\dvar i \dequal \dvar j}
\mid
{\dvar i \dequal \dvar j\dplus \dvar k}
\mid
{\dvar i \dequal \dvar j\dmult \dvar k} 
\mid A\dconj B \mid A\ddisj B \mid \dexists A
\]
The letters $A,B$ range over formul{\ae} and $n:\nat$ represents constant ranging over natural numbers.
We use standard \emph{De Bruijn syntax} with variables $\dvar0, \dvar1, \dots$ of type $\Var := \nat$
for better readability.
If we have $\dvar i : \Var$, we write $\dvar{\dsucc i}$ for the next variable in $\Var$. 
As an example, the meta-level formula $\exists y,(y=0 \land \exists z,\,y= z\plus k)$ would be represented as
e.g.\ $\dexists\bigl(x_0 \dequal 0 \dconj \dexists(x_1 \dequal x_0\dplus x_2)\bigr)$, i.e.\ the variable
$x_i$ refers to the $i$-th binder in the context. Notice that there is no type or syntactic distinction 
between variables and parameters in Diophantine logic. However some variables are bound in their context
and others are free.

We provide a semantics for Diophantine logic.
Given a valuation for variables $\gval:\Var\cfun\nat$,
we define the interpretation $\sem{{\dvar i\dequal\ldots}}\gval:\Prop$ of atomic formul{\ae} by
\[ \begin{array}{l@{~\defeq~}l@{\qquad}l@{~\defeq~}l}
  \sem{\dvar i\dequal n}\gval                      & \gval\,\dvar i = n
& \sem{\dvar i\dequal \dvar j\dplus \dvar k}\gval  & \gval\,\dvar i = \gval\,\dvar j\plus \gval\,\dvar k \\
  \sem{\dvar i\dequal \dvar j}\gval                & \gval\,\dvar i = \gval\,\dvar j 
& \sem{\dvar i\dequal \dvar j\dmult \dvar k}\gval  & \gval\,\dvar i = \gval\,\dvar j\mult \gval\,\dvar k
\end{array} \]
and $\coqlink[df_pred]{\sem A{\gval}}:\Prop$ for a compound formula $A:\dform$ by the following recursive equations:
\[\begin{array}{c@{\qquad}c@{\qquad}c}
\sem{A\dconj B}\gval\cdef\sem A\gval\lconj\sem B\gval 
& \sem{A\ddisj B}\gval\cdef\sem A\gval\ldisj\sem B\gval &  \sem{\dexists A}\gval\cdef\exists n:\nat,\sem A{\dlift\gval n}
\end{array}\]
where $\dlift\gval n:\Var\cfun\nat$ is the standard De Bruijn extension%
\footnote{The notation $\dlift\gval n$ emphasizes that the value $n$ is pushed ahead of the infinite sequence 
$\gval\,\dvar 0;\gval\,\dvar 1;\gval\,\dvar 2;\ldots$}
of a valuation $\gval$ by $n$,
defined by $\dlift\gval n\,(\dvar 0)\cdef n$ and $\dlift\gval n\,(\dvar{\dsucc i})\cdef \gval\,\dvar i$.

We give a first object-level representation of Diophantine relations as  
members of type $(\Var\cfun\nat)\cfun\Prop$ mapping valuations of variables to propositions. 
Moreover, they must be identical to $\abst\gval{\sem A\gval}$
for some Diophantine formula $A$, up to propositional extensionality. 
We give an informative
content to this sub-type of $(\Var\cfun\nat)\cfun\Prop$ to be able to do some computations
with the witness $A:\dform$ of Diophantineness,
typically 
when moving to another formal representation
like elementary Diophantine constraints in Section~\ref{dio_elem_sect}.

\begin{definition}[][dio_rel]
\label{dio_rel_def}
The class of \emph{Diophantine relations} $\diorel:\bigl((\Var\cfun\nat)\cfun\Prop\bigr)\cfun\Type$
is the informative sub-type defined for $R:(\Var\cfun\nat)\cfun\Prop$ by
\[\diorel\,R \,\defeq \sum {A:\dform},\bigl(\forall \gval, \sem A\gval\toot R~\gval\bigr).\]
\end{definition}

Note that $\sum$ denotes type-theoretic dependent pairs.
Hence an inhabitant $w$ of\ $\diorel\,R$ is a (dependent) pair
$(A,H_A)$ where $A=\pi_1(w)$ is a Diophantine formula and $H_A=\pi_2(w)$ a proof that
$\sem A{(\cdot)}$ and $R$ are extensionally equivalent\rlap.\footnote{For the efficiency
of computations, we usually hide the purely logical part $H_A$ into an \emph{opaque proof term.}}\enspace
The informal notion of Diophantine shape will correspond to the dynamically growing collection
of established closure properties of the class $\diorel$ of Diophantine relations.
Definition~\ref{dio_rel_def} of the sub-type $\diorel$ already entails
that $\diorel$ is closed under conjunction, disjunction, existential quantification
and renaming. 

\newcommand{\compose}{\mathbin\circ}

\begin{proposition}
\label{diorel_prop1}
Let\/ $R,S:(\Var\cfun\nat)\cfun\Prop$ be relations,\/
$T:\nat\cfun(\Var\cfun\nat)\cfun\Prop$ be a relation with a singled out variable, and\/
$\gren:\Var\cfun\Var$ be a renaming function. We have the maps:
$$\begin{array}{l@{\qquad}l}
  \text{1.}~\coqlink[dio_rel_conj]{\diorel\,R\cfun\diorel\,S\cfun\diorel(\abst\gval{R\,\gval\lconj S\,\gval})};
& \text{4.}~\coqlink[dio_rel_equiv]{(\forall\gval, S\,\gval\toot R\,\gval)\cfun\diorel\,R\cfun\diorel\,S};\\
  \text{2.}~\coqlink[dio_rel_disj]{\diorel\,R\cfun\diorel\,S\cfun\diorel(\abst\gval{R\,\gval\ldisj S\,\gval})};
& \text{5.}~\coqlink[dio_rel_ren]{\diorel\,R\cfun\diorel\bigl(\abst\gval{R\,(\gval\compose\gren)}\bigr)}.\\[.3ex]
\multicolumn{2}{l}
  {\text{3.}~\coqlink[dio_rel_exst]{\diorel\bigl(\abst \gval{T\,(\gval\,\dvar 0)\,(\dcut\gval)}\bigr) \cfun \diorel\bigl(\abst \gval{\exists u,T\,u~\gval}\bigr)};}
\end{array}$$
\end{proposition}

Understood as Diophantine shapes, maps number~1--3 recognise 
the logical connectives  of conjunction, disjunction 
and existential quantification as newly allowed shapes.
Map number~5 allows renaming (free) variables hence Diophantine
shapes are closed under renaming\rlap.\footnote{$\gval\compose\gren:\Var\to\nat$
denotes the composition $\abst{\dvar i}{\gval\,\dvar{\gren_i}}$ of the valuation 
$\gval:\Var\to\nat$ with the renaming $\gren:\Var\to\Var$.}\enspace 
Map number~4 provides a way to replace the goal $\diorel\,S$ with
$\diorel\,R$ once a proof that they are logically equivalent is
established. Hence, if $S$ cannot be analysed because it does
not currently have a Diophantine shape, 
it can still be replaced by an equivalent 
relation $R$, hopefully better behaved; 
see e.g.\ the proof of Proposition~\ref{prop:diofun_basic}.

Working with Diophantine relations already gives a satisfying implementation
of Diophantine shapes but it is sometimes
more convenient to manipulate Diophantine functions instead of
relations so we define the following specialization.

\begin{definition}[][dio_fun]
\label{dio_fun_def}
The class of \emph{Diophantine functions} $\diofun:\bigl((\Var\cfun\nat)\cfun\nat\bigr)\cfun\Type$
is the informative sub-type defined for $f:(\Var\cfun\nat)\cfun\nat$ by\/
$\diofun\,f \defeq \diorel\bigl(\abst\gval{\gval\,\dvar 0 = f (\dcut\gval)}\bigr)$.
\end{definition}

We extend Diophantine shapes with polynomials expressions and equations between them.
To illustrate the mechanics behind Diophantine shape recognition, for once we
give a detailed account of the steps that are usually performed automatically in this 
framework. 

\begin{proposition}
\label{prop:diofun_basic}
Let\/ $\dvar i:\Var$, $n:\nat$, $\gren:\Var\cfun\Var$ and\/ $f,g:(\Var\cfun\nat)\cfun\nat$.
We have the maps:
$$\begin{array}{l@{\qquad}l}
  \text{1.}~\coqlink[dio_fun_var]{\diofun\,(\abst\gval{\gval\,\dvar i})};
& \text{4.}~\coqlink[dio_fun_plus]{\diofun\,f\to \diofun\, g\to \diofun\,(\abst\gval{f\,\gval\plus g\,\gval})};\\
  \text{2.}~\coqlink[dio_fun_cst]{\diofun\,(\abst\gval n)};
& \text{5.}~\coqlink[dio_fun_mult]{\diofun\,f\to \diofun\, g\to \diofun\,(\abst\gval{f\,\gval\mult g\,\gval})};\\
  \text{3.}~\coqlink[dio_fun_ren]{\diofun\,f\to \diofun\bigl(\abst\gval{f(\gval\compose\gren)}\bigr)};
& \text{6.}~\coqlink[dio_rel_fun_eq]{\diofun\,f\to \diofun\, g\to \diorel\,(\abst\gval{f\,\gval = g\,\gval})}.
\end{array}$$
\end{proposition}

\begin{proof}
Items~1, 2 and map~3 are for projections, constants and renaming functions respectively. 
The non-trivial cases are for $\plus$, $\mult$ and $=$. 
We cover the cases of $+$ and then $=$ in details to illustrate how the representations
of Diophantine relations/functions behave in proof scripts.
In particular, we prove the results reasoning backwards (as is usually done in Coq),
applying established theorems to convert a given proof goal into (hopefully)
simpler proof goals.

\smallskip
  
For the goal $\diofun\,(\abst\gval{f\,\gval\plus g\,\gval})$, unfolding the assumptions $\diofun\,f$
and $\diofun\,g$ we have
\begin{equation}
\label{eq:proof_equiv0}
\diorel\bigl(\abst\gval{\gval\,\dvar 0 = f\,(\dcut\gval)}\bigr)%
\qquad\text{and}\qquad 
\diorel\bigl(\abst\gval{\gval\,\dvar 0 = g\,(\dcut\gval)}\bigr)%
\end{equation}
and we want to establish $\diofun\,(\abst\gval{f\,\gval\plus g\,\gval})$, i.e.\
\begin{equation}
\label{eq:proof_equiv1}
\diorel\bigl(\abst\gval{\gval\,\dvar 0 = f\,(\dcut\gval)+g\,(\dcut\gval)}\bigr).
\end{equation}
By map~4 of Proposition~\ref{diorel_prop1}, we replace Eq.~\eqref{eq:proof_equiv1} with
the (obviously) equivalent goal
\begin{equation}
\label{eq:proof_equiv}
\diorel\bigl(\abst\gval{\exists a\exists b,\, \gval\,\dvar 0 = a+b\lconj a = f\,(\dcuto\gval1)\lconj b = g\,(\dcuto\gval1)}\bigr)
\end{equation}
and we then apply map~3 of Proposition~\ref{diorel_prop1} twice to get the goal
$$\diorel\bigl(\abst\gval{\gval\,\dvar 2 = \gval\,\dvar 1+\gval\,\dvar 0\lconj \gval\,\dvar 1 = f\,(\dcuto\gval3)\lconj \gval\,\dvar 0 = g\,(\dcuto\gval3)}\bigr).$$
We now apply twice map~1 of Proposition~\ref{diorel_prop1} and we get the three following sub-goals:
$$\diorel(\abst\gval{\gval\,\dvar 2 = \gval\,\dvar 1+\gval\,\dvar 0})\quad
\diorel\bigl(\abst\gval{\gval\,\dvar 1 = f\,(\dcuto\gval3)}\bigr)\quad
\diorel\bigl(\abst\gval{\gval\,\dvar 0 = g\,(\dcuto\gval3)}\bigr)$$
\begin{enumerate}
\item For the first sub-goal, we use the formula $\dvar 2\dequal\dvar 1\dplus\dvar 0:\dform$
      as object-level witness;
\item for the second sub-goal, 
      we consider the renaming function $\gren_1:\Var\to\Var$ defined by $\gren_1(\dvar 0)\defeq \dvar 1$ 
      and $\gren_1(\dvar{1+i})\defeq \dvar{3+i}$ and derive 
      $\diorel\bigl(\abst\gval{\gval\,\dvar 1 = f\,(\dcuto\gval3)}\bigr)$ by applying map~5 of Proposition~\ref{diorel_prop1} 
      to the hypothesis $\diorel\bigl(\abst\gval{\gval\,\dvar 0 = f\,(\dcut\gval)}\bigr)$
      in Eqs.~\eqref{eq:proof_equiv0};
\item the third and last sub-goal $\diorel\bigl(\abst\gval{\gval\,\dvar 0 = g\,(\dcuto\gval3)}\bigr)$ 
      is solved similarly with the renaming function $\gren_0:\Var\to\Var$
      defined by $\gren_0(\dvar 0)\defeq \dvar 0$ and $\gren_0(\dvar{1+i})\defeq \dvar{3+i}$.
\end{enumerate}

\smallskip

We now deal with map~6, hence with the goal $\diorel\,(\abst\gval{f\,\gval = g\,\gval})$ under the same previous assumptions $\diofun\,f$
and $\diofun~g$, i.e.\ Eqs.~\eqref{eq:proof_equiv0}. 
We proceed in a somewhat less detailed explanation. We replace the goal by 
$\diorel\,(\abst\gval{\exists a\exists b,\, a = b \lconj a = f\,\gval \lconj b = g\,\gval})$
which is equivalent and then,
after applying the maps of Proposition~\ref{diorel_prop1}, we get three
sub-goals $\diorel\,(\abst\gval{\gval\,\dvar 1=\gval\,\dvar 0})$,
$\diorel\bigl(\abst\gval{\gval\,\dvar 1=f\,(\dcuto\gval2)}\bigr)$
and  $\diorel\bigl(\abst\gval{\gval\,\dvar 0=g\,(\dcuto\gval2)}\bigr)$.
In turn, the first sub-goal corresponds to the witness $\dvar 1\dequal\dvar 0:\dform$, while the second and third 
sub-goals follow from Eqs.~\eqref{eq:proof_equiv0} respectively using straightforward renaming functions.
\end{proof}

On paper these proofs look somehow complicated by the need to infer the renaming 
functions but from a \emph{mechanisation point of view,} Coq's unification algorithm automatically 
solves such goals.
Provided we populate the \emph{hint database} with enough admissible shapes,
we can automate the analysis of the meta-level syntax to establish Diophatineness
and reflect a meta-level expression of Diophantine shape into the
corresponding object-level witness of type $\dform$ together with the proof 
that it is an appropriate witness, hence packed into the types $\diorel\,R$ 
for relational expressions or $\diofun\,f$ for functional expressions.

\newcommand{\dioauto}{\ensuremath{\ctactic{dio}~\ctactic{auto}}}

\smallskip

With Propositions~\ref{diorel_prop1} and~\ref{prop:diofun_basic}, we
populate the hint database for relations with the shapes conjunction, disjunction,
existential quantification, renaming and identity between two functional expressions, 
and for functions, we add the shapes of projections,
constants, addition, multiplication and renaming.
In our implementation, the analysis of Diophantine shapes is performed by
the automatic \coqlink[dio_rel_auto]{\dioauto} tactic.
With such an automated approach, the remaining (and sometimes difficult) work 
occurs when we apply map~4 of Proposition~\ref{diorel_prop1}, that is, we have 
to find an equivalent expression of Diophantine shape, like in Eq.~\eqref{eq:proof_equiv}
and to prove it is indeed equivalent to Eq.~\eqref{eq:proof_equiv1}, which, unlike that
specific example, might be non-trivial; see e.g.\ the discussion in
Section~\ref{discuss_matiya_sect}.

\begin{proposition}
\label{diorel_prop}
Let\/ 
$f,g :(\Var\cfun\nat)\cfun\nat$. We have the maps:
$$\begin{array}{l@{\qquad\qquad}l}
  \text{1.}~\coqlink[dio_rel_True]{\diorel\,(\abst\gval\True)};
& \text{3.}~\coqlink[dio_rel_le]{\diofun\,f\to \diofun\, g\to \diorel\,(\abst\gval{f\,\gval \leq g\,\gval})};\\
  \text{2.}~\coqlink[dio_rel_False]{\diorel\,(\abst\gval\False)};
& \text{4.}~\coqlink[dio_rel_lt]{\diofun\,f\to \diofun\, g\to \diorel\,(\abst\gval{f\,\gval < g\,\gval})};\\
& \text{5.}~\coqlink[dio_rel_neq]{\diofun\,f\to \diofun\, g\to \diorel\,(\abst\gval{f\,\gval \neq g\,\gval})}.
\end{array}$$
\end{proposition}

\begin{proof}
For e.g.\ $\diorel\,(\abst\gval{\,f\,\gval < g\,\gval})$, 
we shift to the equivalent $\diorel\,(\abst\gval{\exists a,1+a+f\,\gval = g\,\gval})$
using map~4 of Proposition~\ref{diorel_prop1} and finish the proof calling \dioauto\rlap.%
\footnote{Notice that the actual implemented proofs might differ slightly because we
sometimes optimize the shape of expressions for smaller witnesses,
especially for these basic shapes which pop up over and over again.}
\end{proof}

Again, we populate the hint database with the new shapes of 
Proposition~\ref{diorel_prop}.
We follow up with the slightly more complex example of the ``does not divide''
relation defined by $u\ndivides v\cdef\lneg(\exists k, v=k\times u)$. At this point, 
this expression cannot be recognized as a Diophantine shape 
because it contains a negation.

\begin{proposition}[][dio_rel_ndivides]
\label{ndiv_dio_prop}
$\forall f\,g:(\Var\cfun\nat)\cfun\nat,\,\diofun\,f\cfun\diofun\,g\cfun\diorel\,(\abst\gval{f\,\gval \ndivides g\,\gval})$.
\end{proposition}

\begin{proof}
However, using  Euclidean division, we (easily) prove the equivalence
\begin{equation*}
u\ndivides v\toot (u=0\lconj v\neq 0 \ldisj \exists a\,b,\, v = a\mult u\plus b\lconj 0 < b < u)
\end{equation*}
and this new expression can now be recognised as a Diophantine shape.
Using this equivalence in combination with map~4 of Proposition~\ref{diorel_prop1},
we replace the goal $\diorel\,(\abst\gval{f\,\gval\ndivides g\,\gval})$ with
$$\diorel\,(\abst\gval{f\,\gval=0\lconj g\,\gval\neq 0 
\ldisj \exists a\,b,\, g\,\gval = a\mult f\,\gval\plus b\lconj 0 < b \lconj b < f\,\gval})$$
and then let the magic of \dioauto\ unfold. 
\end{proof}

Again, once established, we can add the map
$\diofun\,f\cfun\diofun\,g\cfun \diorel\,(\abst\gval{f\,\gval\ndivides g\,\gval})$
in the Diophantine hint database so that later encountered proof goals $\diorel\,(\abst\gval{f\,\gval\ndivides g\,\gval})$ 
can be immediately solved by \dioauto.

\smallskip

In this above described approach, the recovery of the 
object-level witness $A$ of Definition~\ref{dio_rel_def} from meta-level syntax 
is automatic and hidden by the use of the \dioauto\ tactic associated 
with the ever growing hint database. This allows us to proceed as in e.g.\
Matiyasevich papers where he usually
transforms a relation into an equivalent Diophantine shape,
accumulating more and more Diophantine shapes on the way.
Instead of having to manipulate object-level witnesses by hand,
obfuscating sometimes simple to understand proofs, we use Diophantine
shapes as the cornerstone of the
faithful implementation of existing pen and paper scripts.

\subsection{Exponentiation and Bounded Universal Quantification}

\label{sec:expo_buq}

For now, we introduce the \emph{elimination of the exponential relation} 
and then of \emph{bounded universal quantification} as 
black boxes expressed in the framework of Diophantine shapes, i.e.\
new closure properties of the classes $\diofun/\diorel$. 

While we do contribute implementations for
both of these hard results, on purpose, 
we choose to avoid the detailed presentation 
of these mechanised proofs for several reasons: 
\begin{itemize}
\item first of all, there are already
fully detailed pen and paper accounts of these results and
we implemented two of these somewhat faithfully; 
\item then, in our modular approach, the proof of these admissibility results 
can be ignored without hindering the understanding of the overall
structure of the main results, e.g.\ \Hten;
\item finally, already the above cited pen and paper proofs assume some
not so standard results in arithmetic like e.g.\ Lucas's theorem,
and we favoured giving accounts of those assumed theorems instead
of simply reproducing the rest of the existing arguments. 
\end{itemize}
Hence, for the moment, we postpone remarks and
discussions about the Diophantineness of the exponential 
function and the Diophantine admissibility of  
bounded universal quantification to Section~\ref{discuss_matiya_sect}.

\setCoqFilename{H10.Dio.dio_expo}
\begin{theorem}[Exponential][dio_fun_expo]
\label{expo_dio_thm}
$\forall f\,g:(\Var\to\nat)\to\nat,\, 
\diofun\,f\cfun\diofun\,g\cfun\diofun\bigl(\abst\gval{(f\,\gval)^{g\,\gval}}\bigr)$.
\end{theorem}

To prove it, one needs a meta-level Diophantine shape for
the exponential relation, \emph{the proof of which is nothing short 
of extraordinary.} This landmark result is due to Matiyasevich~\cite{matijasevic1970enumerable},
but we have implemented the shorter and more up-to-date proof 
of~\cite{Matiyasevich2000}.

\setCoqFilename{H10.Dio.dio_bounded}
\begin{theorem}[Bounded Univ.\ Quantification][dio_rel_fall_lt]
For\/ $f:(\Var\cfun\nat)\cfun\nat$ and\/ $T:\nat\cfun(\Var\cfun\nat)\cfun\Prop$,
we have\/
$     \diofun\,f
 \cfun \diorel\bigl(\abst \gval{T\,(\gval\,\dvar 0)~(\abst{\dvar i}{\gval\,\dvar{\dsucc i}})}\bigr)
 \cfun \diorel\bigl(\abst \gval{\forall u, u < f\,\gval\cfun T\,u~\gval}\bigr)$.
\end{theorem}

This map can be compared with map~3 of Proposition~\ref{diorel_prop1} and
allows to recognise bounded universal quantification as a legitimate Diophantine
shape. We have implemented the direct proof 
of Matiyasevich~\cite{Matiyasevich1997} which does not involve a detour through a 
model of computation. 
Notice that the bound $f\,\gval$ in $\forall u,\, u<f\,\gval \cfun\ldots$ 
is not assumed constant otherwise the elimination of the quantifier would 
proceed as a simple reduction to a finitary conjunction.

\subsection{Reflexive-Transitive Closure is Diophantine}

With these tools~--~elimination of the exponential relation and 
of bounded universal quantification~--~we can show that the reflexive and transitive closure
of a Diophantine binary relation is itself Diophantine.
We assume a binary relation $R:\nat\cfun\nat\cfun\Prop$ over natural numbers.
The Diophantineness of $R$ can be formalised by assuming that e.g.\
$\abst\gval{R~(\gval\,\dvar 1)~(\gval\,\dvar 0)}$ is a Diophantine relation.
We show that the $i$-th iterate of $R$ is Diophantine (where $i$ is non-constant).

\setCoqFilename{H10.Dio.dio_rt_closure}
\begin{lemma}[][dio_rel_rel_iter]
For any binary relation\/ $R:\nat\to\nat\to\Prop$ and\/ $f,g,h:(\Var\cfun\nat)\cfun\nat$, we have
$$\diofun\,f\cfun\diofun\,g\cfun\diofun\,h\cfun
\diorel\bigl(\abst\gval{R~(\gval\,\dvar 1)~(\gval\,\dvar 0)}\bigr)\cfun
\diorel\bigl(\abst\gval{R^{h\,\gval}\,(f\,\gval)~(g\,\gval)}\bigr).$$
\end{lemma}

\begin{proof}
Using Euclidean division, we define 
the $\cstfont{is\_digit}~c~q~n~d$ predicate stating that $d$ is the $n$-th
digit of the base $q$ development of number $c$, as a Diophantine sentence:
\[\cstfont{is\_digit}~c~q~n~d\cdef d < q\lconj
\exists a\,b,\, c = (aq+d)q^n+b\lconj b < q^n.\]
The Diophantineness of this follows from 
previous Diophantine shapes, including the exponential (\Cref{coq:dio_fun_expo}).
Then we define the $\cstfont{is\_seq}~R~c~q~i$ predicate stating that the first $i+1$ digits of
$c$ in base $q$ form an $R$-chain, again with a Diophantine shape, established using hypothesis 
$\diorel\bigl(\abst\gval{R~(\gval\,\dvar 1)~(\gval\,\dvar 0)}\bigr)$ and the 
Diophantine admissibility of bounded universal quantification (\Cref{coq:dio_rel_fall_lt}):
\[\cstfont{is\_seq}~R~c~q~i\cdef \forall n,\,n<i\cfun\exists u\,v,\,
\cstfont{is\_digit}~c~q~n~u\lconj
\cstfont{is\_digit}~c~q~(1+n)~v\lconj
R~u~v\]
Then we encode $R^i~u~v$ by stating that there exists a (large enough)  
$q$ and a number $c$ such that the first $i+1$ digits of $c$ in base $q$ form
an $R$-chain starting at $u$ and ending at $v$:
\[ R^i~u~v~\toot~ \exists q\,c,\,
\cstfont{is\_seq}~R~c~q~i
\lconj \cstfont{is\_digit}~c~q~0~u
\lconj \cstfont{is\_digit}~c~q~i~v\]
and this 
expression is accepted as a Diophantine shape by \dioauto.
Then assuming Diophantineness of $f$, $g$ and $h$, we
easily derive that $\abst\gval{R^{h\,\gval}\,(f\,\gval)~(g\,\gval)}$
is Diophantine.
\end{proof}

We fill in \Cref{coq:dio_rel_rel_iter} in the Diophantine hint database
and we derive the Diophantineness of the reflexive-transitive closure as a
direct consequence of the equivalence ${R^\kleenestar}~u~v\toot\exists i, {R^i}~u~v$.

\begin{theorem}[RT-closure][dio_rel_rt]
For any binary relation\/ $R:\nat\cfun\nat\cfun\Prop$, we have the map\/
$$\forall f\,g:(\Var\cfun\nat)\cfun\nat,\,\diofun\,f\cfun\diofun\,g\cfun
\diorel\bigl(\abst\gval{R~(\gval\,\dvar 1)~(\gval\,\dvar 0)}\bigr)
\cfun\diorel\bigl(\abst\gval{R^\kleenestar\,(f\,\gval)~(g\,\gval)}\bigr).$$
\end{theorem}

\section{Elementary Diophantine Constraints}

\label{dio_elem_sect}

We now shift to another, seemingly less expressive, object-level representation of Diophantine relations.
Elementary Diophantine constraints are very simple equations where only one instance of
either $\dplus$ or $\dmult$ is allowed. Schematically, starting from Diophantine logic, 
we remove disjunction and existential quantification and encode conjunctions into the
structure of a list. We give a direct proof that any Diophantine logic formula
is semantically equivalent to the simultaneous satisfiability of a list of elementary Diophantine constraints.

\setCoqFilename{H10.Dio.dio_elem}

Starting from two copies of $\nat$, one called $\BVar$ with $u,v,w$ ranging over $\BVar$ for existentially quantified variables,
and another one $\Var=\{\dvar 0,\dvar 1,\ldots\}$ for parameters, we define the type of elementary 
Diophantine constraints by:\footnote{The equation $u\dequal v$ is redundant because
it could be replaced with $z\dequal 0\dconj u\dequal z\dplus v$, for some fresh $z$. However we keep 
$u\dequal v$ in the syntax because this
simplifies arguments when parameters $\dvar i$ are mapped to existential variables $v$
in the proof of \Cref{coq:dio_repr_at_form}, the
type $\delem$ being thus closed under this transformation.} 
\[
c:\coqlink[dio_constraint]{\delem}
\bnfdef
     u\dequal n 
\mid u\dequal v 
\mid u\dequal \dvar i
\mid u\dequal v\dplus w 
\mid u\dequal v\dmult w
\qquad \text{where $n : \nat$}
\]
Notice that these constraints do not have a ``real'' inductive structure,
they are flat and of size (number of symbols) either 3 or 5.
Given two interpretations, $\genv:\BVar\cfun\nat$ for variables
and $\gval:\Var\cfun\nat$ for parameters, it is trivial to define the
semantics $\coqlink[dc_eval]{\semb c\gval\genv}:\Prop$ of a single constraint $c$ of type $\delem$:
\[\begin{array}{c@{\qquad}c@{\qquad}c}
  \semb{u\dequal n}\gval\genv\cdef \genv\,u=n
& \semb{u\dequal v}\gval\genv\cdef \genv\,u=\genv\,v
& \semb{u\dequal v\dplus w}\gval\genv\cdef \genv\,u=\genv\,v\plus\genv\,w \\
& \semb{u\dequal \dvar i}\gval\genv\cdef \genv\,u=\gval\,{\dvar i}
& \semb{u\dequal v\dmult w}\gval\genv\cdef \genv\,u=\genv\,v\mult \genv\,w
\end{array}\]

Given a list $l:\List \delem$ of constraints, we write $\semb l\gval\genv$
when all the constraints in $l$ are simultaneously satisfied, i.e.\
$\semb l\gval\genv\cdef \forall c,\, c\in l\cfun\semb c\gval\genv$.
We show the following result:

\setCoqFilename{H10.Dio.dio_elem}
\begin{theorem}[][dio_formula_elem]
\label{form_to_elem_thm}
For any Diophantine formula\/ $A:\dform$ one can compute a list of
elementary Diophantine constraints\/ $l:\List\delem$ such that\/
$\forall \gval:\Var\cfun\nat,\,\sem A\gval\toot\exists\genv:\BVar\cfun\nat,\, \semb l\gval\genv$.
\end{theorem}
Put in other terms, for any given interpretation $\gval$ of parameters, $\sem A\gval$ holds if and only if  
the constraints in $l$ are simultaneously satisfiable. Any Diophantine logic formula is 
equivalent to the satisfiability of the conjunction of finitely many elementary Diophantine 
constraints.
The proof of~\Cref{coq:dio_formula_elem} spans the rest of this section. We will strengthen
the result a bit to be able to get an easy argument by induction on $A$. 

\newcommand{\dseqns}{\ensuremath{\mathcal{E}}}
\newcommand{\dsref}{\ensuremath{\mathfrak{r}}}

\begin{definition}[][dio_repr_at]
Given a relation $R:(\Var\cfun\nat)\cfun\Prop$ and an interval ${[u_a,u_{a+n}[}\subseteq\BVar$,
an\/ \emph{elementary representation of\/ $R$ in\/ $[u_a,u_{a+n}[$} is given by:
\begin{enumerate}
\item a list\/ $\dseqns: \List\delem$ of constraints and a reference variable\/ $\dsref : \BVar$;
\item proofs that\/ $\dsref$ and the (existentially quantified) variables of\/ $\dseqns$ belong to\/ $[u_a,u_{a+n}[$;
\item a proof that the constraints in\/ $\dseqns$ are always (simultaneously) satisfiable,
      i.e.\/ $\forall\gval\exists\genv\,\semb\dseqns\gval\genv$;
\item a proof that the list\/ $(\dsref\dequal 0)\ccons\dseqns$ is equivalent to\/ $R$, i.e.\/
      $\forall\gval,\, R~\gval \toot \bigl(\exists\genv,\,\genv\,{\dsref}=0\lconj
\semb \dseqns\gval\genv\bigr)$.
\end{enumerate}
\end{definition}

It is obvious that an elementary representation of $\abst\gval{\sem A\gval}$ in any interval $[u_a,u_{a+n}[$ is enough
to prove \Cref{coq:dio_formula_elem} because of item~4 of \Cref{coq:dio_repr_at}.
But actually, computing such a representation is simpler 
than proving \Cref{coq:dio_formula_elem} directly\rlap.%
\footnote{Proving \Cref{coq:dio_formula_elem} directly involves renamings of
existential variables and might produce exponential blow-up in the number of constraints
when handled naively.}\enspace
Below, the size of a Diophantine formula $A$, denoted
\setCoqFilename{H10.Dio.dio_logic}%
\coqlink[df_size]{$\dsize A$}, is defined as the number of nodes of its syntactic tree.

\setCoqFilename{H10.Dio.dio_elem}
\begin{lemma}[][dio_repr_at_form]
For any\/ $a:\nat$ and any\/ $A:\dform$, one can compute\/ $n\leq 8\dsize A$
and an elementary representation of the relation\/ $\abst\gval{\sem A\gval}$ 
in\/ $[u_a,u_{a+n}[$.
\end{lemma}

\begin{proof}
We show the result by structural induction on $A$. 
\begin{itemize}
\item If $A$ is e.g.\  $\dvar i\dequal\dvar j\dmult \dvar k$, we get the representation 
      with $n\defeq 8$ and the pair $(\dseqns,\dsref)$ with
      $$ 
         \dseqns\defeq \left\lbrack\begin{array}{@{\,}l@{\,}}
            u_{a+7}\dequal u_a\dplus u_{a+1}      ~;~
            u_{a+6}\dequal u_a\dplus u_{a+2}      ~;~
            u_{a+6}\dequal u_{a+1}\dplus u_{a+5}  ~; \\
            u_{a+5}\dequal u_{a+3}\dmult u_{a+4}  ~;~
            u_{a+4}\dequal \dvar k                ~;~
            u_{a+3}\dequal \dvar j                ~;~
            u_{a+2}\dequal \dvar i              
         \end{array}\right\rbrack  
         \qquad
         \dsref\defeq u_{a+7}
      $$
      Property (2) is obviously satisfied.
      Property (3), i.e.\ the satisfiability of $\dseqns$ whatever the values of $\dvar i$, $\dvar j$
      and $\dvar k$, is simple to establish: indeed, the values of $u_{a+2}=\dvar i$, $u_{a+3}=\dvar j$,
      $u_{a+4}=\dvar k$ and $u_{a+5}=\dvar j\dvar k$ are uniquely determined. 
      Pick $u_a\defeq \dvar j\dvar k$ and $u_{a+1}\defeq \dvar i$ and then 
      again $u_{a+6} = u_{a+7} = \dvar i+\dvar j\dvar k$ 
      are both uniquely determined. This assignment of variables satisfies all the constraints in $\dseqns$.

      For property~(4), let us now add the extra constraint $u_{a+7}\dequal 0$ and consider a 
      valuation satisfying the constraints in $(u_{a+7}\dequal 0)\ccons \dseqns$. We must have
      $u_{a+7} = u_a+u_{a+1} =0$, hence $u_a=u_{a+1}=0$ which entails
      $u_{a+6} = u_{a+2} = u_{a+5}$. But $u_{a+2} = \dvar i$
      and $u_{a+5} = u_{a+3}u_{a+4} =\dvar j\dvar k$. Hence 
      $\dvar i\dequal\dvar j\dmult \dvar k$ is satisfied.

      Conversely, parameter values satisfying $\dvar i\dequal\dvar j\dmult \dvar k$
      can be extended to a valuation of variables satisfying $(u_{a+7}\dequal 0)\ccons \dseqns$
      in a unique way with $u_{a+7},u_a,u_{a+1} \defeq 0$,
      $u_{a+6}, u_{a+5}, u_{a+2} \defeq \dvar i$, $u_{a+3}\defeq\dvar j$ and
      $u_{a+4}\defeq\dvar k$;

\item We proceed similarly for the other atomic cases $\dvar i\dequal n$, $\dvar i\dequal\dvar j$ and $\dvar i\dequal\dvar j\dplus \dvar k$;

\item When $A$ is $B\dconj C$, we get a representation in $[u_a,u_{a+n_A}[$ by induction. 
      Hence, let $(\dseqns_B,\dsref_B)$ be the representation
      of $B$ in $[u_a,u_{a+n_B}[$.
      Then, inductively again, let $(\dseqns_C,\dsref_C)$ be a representation 
      of $C$ in $[u_{a+n_B},u_{a+n_B+n_C}[$. We define 
      $\dsref_A \cdef u_{a+n_A+n_B}$ and 
         $\dseqns_A \cdef (\dsref_A\dequal \dsref_B\dplus \dsref_C)\ccons \dseqns_B\capp \dseqns_C$
      and then $(\dseqns_A,\dsref_A)$ represents $A=B\dconj C$ in $[u_a,u_{a+1+n_B+n_C}[$;%
      \footnote{Since the intervals $[u_a,u_{a+n_B}[$ and $[u_{a+n_B},u_{a+n_B+n_C}[$ are built
      disjoint, there is no difficulty in merging valuations whereas this usually
      involves renamings when existential variables are not carefully chosen.}
\item The case of $B\ddisj C$ is similar: simply replace $\dsref_A\dequal \dsref_B\dplus\dsref_C$
      with $\dsref_A\dequal \dsref_B\dmult\dsref_C$;
\item We finish with the case when $A$ is $\dexists B$. Let
      $(\dseqns_B,\dsref_B)$ be a representation of $B$ in $[u_a,u_{a+n_B}[$.
      Let $\gsubst$ be the substitution mapping parameters in $\Var$ and
      defined by $\sigma(\dvar 0)\cdef u_{a+n_B}$
      and $\sigma(\dvar{1+i})\cdef\dvar i$; existential variables in $\BVar$ are left unmodified 
      by this substitution. 
      Then $\bigl(\gsubst(\dseqns_B),\dsref_B\bigr)$ is a representation of $A=\dexists B$ in $[u_a,u_{a+1+n_B}[$.
\end{itemize}
This concludes the recursive construction of a representation of  $\abst\gval{\sem A\gval}$.
\end{proof}

\section{Single Diophantine Equations}

\label{dio_single_sect}

We now give our last and most naive object-level representation of Diophantine relations as
a single polynomial equation.  
In this section, we show how a list of elementary Diophantine constraints can be simulated by
a single identity between two Diophantine polynomials. We use the following well known
convexity identity to achieve the reduction.

\setCoqFilename{H10.Dio.dio_single}
\begin{proposition}[][convex_n_eq]
\label{convexity2}
Let\/ $(p_1,q_1),\ldots,(p_n,q_n)$ be a sequence of pairs in $\nat\times\nat$. Then
\[\hfil \sum_{i=1}^n 2p_iq_i =\! \sum_{i=1}^n p_i^2+q_i^2~\toot~p_1=q_1\lconj\cdots\lconj p_n=q_n.\]
\end{proposition}

There are many possible justifications for this equivalence which, if ranging
over $\Z$, would more conventionally be written as $\sum_i(p_i-q_i)^2 = 0\toot\forall i,\, p_i=q_i$.
This form does not hold directly for $\nat$ however, hence the reformulation without
using the minus/$-$ operator\rlap.\footnote{Remember that $0-1=0$ holds in $\nat$.}\enspace 
See Appendix~\ref{app_convexity2} for an elementary justification of the result.

\smallskip

Similarly to elementary Diophantine constraints, 
we define Diophantine polynomials distinguishing the types of
$\BVar$ of bound variables and $\Var$ of parameters (or free variables) 
but the types $\BVar$ and $\Var$ are not fixed copies of $\nat$ anymore,
but type parameters of arbitrary value. 

\begin{definition}[][dio_polynomial]
\label{dio_polynomial_nat}
The type of\/ \emph{Diophantine polynomials} $\dpoly(\BVar,\Var)$ and
the type of\/ \emph{single Diophantine equations} $\dsingle(\BVar,\Var)$ are defined by:
\[p,q : \dpoly(\BVar,\Var) \bnfdef u : \BVar \mid \genvar :\Var \mid n :\nat \mid p\dplus q\mid p\dmult q
\qquad 
E : \dsingle(\BVar,\Var) \bnfdef p\dequal q.\]
\end{definition}

For $\genv:\BVar\cfun\nat$ and $\gval:\Var\cfun\nat$ we define the semantic
interpretations of polynomials $\semb p\gval\genv:\nat$ and single Diophantine 
equations $\semb{p\dequal q}\gval\genv:\Prop$ in the obvious way.

\begin{theorem}[][dio_elem_single]
\label{elem_single_thm}
For any list\/ $l:\List\delem$ of elementary Diophantine constraints, one can compute a single Diophantine equation\/
$E:\dsingle(\nat,\nat)$ such that\/ $\forall\gval\forall\genv,\,\semb E\gval\genv\toot\semb l\gval\genv$.
\end{theorem}

\begin{proof}
We write $l = [p_1\dequal q_1;\ldots;p_n\dequal q_n]$ and then
use Proposition~\ref{convexity2}.
In the code, we moreover show that the size of $E$ is linear in the length of $l$.
If needed, one could also show that the degree of the polynomial is less than 4.
\end{proof}

\begin{corollary}[][dio_rel_single]
Let\/ $R:(\nat\cfun\nat)\cfun\Prop$. Assuming\/ $\diorel\,R$, 
one can compute a single Diophantine equation\/ $p\dequal q:\dsingle(\nat,\nat)$ such that\/ 
$\forall \gval,\,R~\gval\toot\exists\genv,\,\semb p\gval\genv = \semb q\gval\genv$.
\end{corollary}

\begin{proof}
Direct combination of Definition~\ref{dio_rel_def} and Theorems~\ref{form_to_elem_thm} and~\ref{elem_single_thm}.
In the formalisation, we also show that the size of the obtained single Diophantine equation is linearly
bounded by the size of the witness formula contained in the proof of $\diorel\,R$.
\end{proof}

We have shown that the automation we designed to recognise relations
of Diophantine shape entail that these relations are also definable by satisfiability of 
a single equation between Diophantine polynomials, so these tools are sound w.r.t.\ 
a formally restrictive characterisation of Diophantineness. 
One could argue that the above existential quantifier $\exists\genv$
encodes infinitely many existential quantifiers but it can easily be replaced by finitely many
existential quantifiers over the bound variables that 
actually occur in $p$ or $q$.

\begin{proposition}[][dio_poly_eq_pos]
For any single Diophantine equation\/ $p\dequal q:\dsingle(\nat,\Var)$, one can compute\/
$n:\nat$ and a new single Diophantine equation\/ $p'\dequal q':\dsingle(\fin n,\Var)$ such that for any\/ $\gval:\Var\cfun\nat$,
$(\exists\genv:\nat\cfun\nat,\,\semb p\gval\genv = \semb q\gval\genv)
\toot(\exists\genv:\fin n\cfun\nat,\,\semb{p'}\gval\genv = \semb{q'}\gval\genv)$.
\end{proposition}

\begin{proof}
We pick $n$ greater that the number of bound variables which occur in either $p$ or $q$. 
This subset of $\nat$ can be faithfully embedded into the finite type $\fin n$ and we use such a renaming
to compute $(p',q')$. Remark that the size of $(p',q')$ is the same as that of $(p,q)$.
\end{proof}

By~\Cref{coq:dio_rel_single} and~\Cref{coq:dio_poly_eq_pos}, we see that a Diophantine logic formula
$A:\dform$ potentially containing inner existential quantifiers and representing 
the Diophantine relation $\abst\gval{\sem A\gval}$ can effectively be reduced to a single
Diophantine equation $p'\dequal q':\dsingle(\fin n,\Var)$ such that
$\sem A\gval\toot \exists\genv:\fin n\cfun\nat,\,\semb{p'}\gval\genv = \semb{q'}\gval\genv$.
Because $\fin n$ is the finite type of $n$ elements, the (higher order) existential quantifier 
$\exists\genv$ simply encodes $n$ successive (first order) existential quantifiers.
The existential quantifiers that occur deep inside $A$ are \emph{not erased} by the
reduction, they are  \emph{moved at the outer level} and to be
ultimately understood as \emph{solvability} for some polynomial equation of which
the parameters match the free variables of $A$.

\section{Remarks on the Implementation of Matiyasevich's Theorems}

\label{discuss_matiya_sect}

\emph{Matiyasevich's theorem} stating that there is a Diophantine
description of the exponential relation $x=y^z$
is a masterpiece which concluded the line of work by Davis, Putnam and Robinson, starting at Davis's conjecture in 1953.
Already in 1952, Julia Robinson discovered that in order to show the exponential relation Diophantine, it suffices to find a single binary Diophantine relation exhibiting exponential growth~\cite{robinson1952}, a so-called Robinson predicate,
i.e.\ a predicate $J(u,v)$ in two variables s.t.\ $J(u,v)$ implies $v < u^u$ and for every $k$ there are $u,v$ with $J(u,v)$ and $v > u^k$.
Robinson's insight meant the only thing missing to prove what is nowadays called the DPRM theorem, 
was a \emph{single} polynomial equation capturing \emph{any} freely chosen Robinson predicate.
Similar to other famous hard problems of mathematics, the question is easy to state, but from the start of the study of Diophantine equations to the late 60s, no such relation was known, rendering the problem one of the most baffling questions for mathematicians and computer scientists alike.

\newcommand{\fib}[1]{\mathrm{fib}_{#1}}

In 1970, Yuri Matiyasevich~\cite{matijasevic1970enumerable} discovered that $v = \fib{2 u}$ is both a
Robinson predicate and Diophantine. Here $(\fib n)_{n\in\nat}$ is the well known Fibonacci sequence defined by the second order 
recurrence relation $\fib 0 = 0$, $\fib 1 = 1$ and $\fib{n+2}=\fib{n+1}+\fib n$. Combined with previous results, 
this concluded the multi-decades effort to establish the Diophantineness of all recursively enumerable predicates,
implying a negative solution to Hilbert's tenth problem. That proof which included the original proof 
of Matiyasevich~\cite{matijasevic1970enumerable} was later simplified. For instance, exploiting similar 
ideas but in the easier context of the solutions of another second order 
equation~--~namely Pell's equations $x^2-(a^2-1)y^2=1$ with parameter $a>1$,~--~Martin
Davis~\cite{davis73} gave a standalone proof of the DPRM-theorem where recursively enumerable predicates
are characterised by a variant of  $\mu$-recursive functions. In that
paper, Davis also provided a proof of the admissibility of bounded universal quantification using
the Chinese remainder theorem to encode finite sequences of numbers. There exists more recent and 
simpler proofs of this admissibility result as well, see e.g.~\cite{Matiyasevich1997}.

\smallskip

Before we discuss the mechanisation of the Diophantineness of both the exponential relation and of bounded universal quantification, 
we want to remark on the difficulty of mechanising the former proof.
Both on its own and as a stepping stone towards the negative solution to Hilbert's tenth problem, it is clear that Matiyasevich's theorem was an extremely difficult question which required superior intellectual resources to be solved.
The mechanisation of a modernised form of the proof, although not trivial, cannot be compared to the difficulty of finding a solution.
In particular, the modern proof relies on very mature background theories, lowering the number of possible design choices for the mechanisation.
Moreover, very detailed pen and paper accounts of the proof are available, which can be followed closely.

An aspect that is more challenging in mechanisation than on paper are proofs regarding the computability of certain functions.
Since paper proofs oftentimes rely on a vague notion of algorithm, most of the reasoning about these algorithms is hand-waved away by computer scientists,
relying on the implicit understanding of what is an algorithm.
By using a synthetic approach to computability~\cite{FHS18,0002KS19,FLW19}, we make the notion of an algorithm precise and thus enable mechanisation, at the same time circumventing the verification of low-level programs.

\subsection{Exponential is Diophantine\texorpdfstring{ (\Cref{coq:dio_fun_expo})}{}}

\newcommand{\modulus}{\mathrel{\mathrm{mod}}}

\setCoqFilename{H10.Dio.dio_expo}

For our mechanised proof, we rely on a more recent account of Matiyasevich's theorem from~\cite{Matiyasevich2000},
which, among the many options we considered, seemed the shortest.
The proof employs the equation $x^2-bxy+y^2=1$ for $b\geq 2$, also
called Pell's equation in~\cite{davis73}. 
We use the second order recurrence relation
$\alpha_b(-1) = -1$, $\alpha_b(0) = 0$ and
$\alpha_b(n+2) = b\alpha_b(n+1)-\alpha_b(n)$
to describe the set of solutions of Pell's equation
by $\bigl\{(\alpha_b(n),\alpha_b(n+1))\mid n\in\nat\bigr\}$.
The recurrence can be characterised by the following square $2\times2$ matrix equation:
\[\hfil
A_b(n) = (B_b)^n
\quad
\text{with~~}
A_b(n)\cdef \left(\begin{array}{@{}c@{}c@{}}\alpha_b(n+1) & -\alpha_b(n) \\  \alpha_b(n)   & -\alpha_b(n-1)\end{array}\right)
\text{~~and~~}
B_b\cdef \left(\begin{array}{@{}c@{~\,}c@{}}b & -1 \\ 1 & 0\end{array}\right)
\] 
Then, studying the properties of the sequence
$n\mapsto\alpha_b(n)$ in $\nat$ or $\mathbb Z$, one can show that
$\alpha_2(n) = n$ for all $n$ and $n\mapsto\alpha_b(n)$ grows
exponentially for $b\geq 3$. Studying the properties
of the same sequence  in $\mathbb Z/p\mathbb Z$
(for varying values of the modulus $p$),
one can for instance show that 
$\modeq{n=\alpha_2(n)}{\alpha_b(n)}{b-2}$,
which relates $n$ and $\alpha_b(n)$ modulo $(b-2)$.
With various intricate but elementary 
results\rlap,\footnote{by elementary we certainly do not mean 
either simple or obvious, but we mean that they only involve
standard tools from modular and linear algebra.}
such as e.g.\ $\alpha_b(k)\divides\alpha_b(m)\toot k\divides m$
and $\alpha_b^2(k)\divides\alpha_b(m)\toot k\alpha_b(k)\divides m$
(both for $b\geq 2$ and any $k,m\in\nat$),
one can show that $a,b,c\mapsto 3 < b\lconj a = \alpha_b(c)$
has a Diophantine representation. In our formalisation, we
get a Diophantine logic formula of size 1445 as
a witness (see~\coqlink[dio_rel_alpha_size]{\cstfont{dio\_rel\_alpha\_size}}). 

Once $\alpha_b(n)$ is proven Diophantine, one can
recover the exponential relation 
$x,y,z\mapsto x=y^z$ using the 
eigenvalue $\lambda$ of the matrix $B_b$ 
which satisfies $\lambda^2-b\lambda-1=0$. By
wisely choosing $m\defeq bq-q^2-1$, one gets
$\modeq\lambda q m$ and thus, using the corresponding 
eigenvector, one derives $\modeq{q\alpha_b(n)-\alpha_b(n-1)}{q^n} m$.
For a large enough value of $m$, 
hence a large enough value%
\footnote{the largeness of which 
is secured using $\alpha$ itself again, but with other 
input values. But this works only in the case where $n>0$ and $q>0$. 
The cases where $n=0$ (and hence $q^n=1$) or 
$q=0$ and $n>0$ (and hence $q^n=0$) are trivial 
and treated separately.}\enspace
of $b$, 
this gives a Diophantine representation of $q^n$.
In our code, we
get a Diophantine logic formula of size 4903 as
a witness (see~\coqlink[dio_fun_expo_example_size]{\cstfont{dio\_fun\_expo\_example\_size}}). 

The main libraries 
which are needed to solve Pell's equation and characterise its
solutions are linear algebra (or at least  square $2\times 2$ matrices) 
over commutative rings such as $\mathbb Z$ and $\mathbb Z/p\mathbb Z$, 
a good library for modular algebra ($\mathbb Z/p\mathbb Z$),
and the binomial theorem over rings. Without the help
of the Coq \cstfont{ring} tactic, such a development 
would be extremely painful. These libraries are then used again
to derive the Diophantine encoding of the exponential.

\subsection{Admissibility of Bounded Universal Quantification\texorpdfstring{ (\Cref{coq:dio_rel_fall_lt})}{}}

As hinted earlier, we provide an implementation of the algorithm 
for the elimination of bounded universal quantification described in~\cite{Matiyasevich1997}.
It does not involve the use of a model of computation, hence does
not create a chicken-and-egg problem when used for the proof
of the DPRM theorem. The technique of~\cite{Matiyasevich1997} uses 
the exponential function and thus \Cref{coq:dio_fun_expo} (a lot), and a combination
of arithmetic and bitwise operations over $\nat$ through base 
$2$ and base $2^q$ representations of natural numbers. 

The Diophantine admissibility of
bitwise operations over $\nat$ is based on the 
relation stating that every bit of $a$ is lower or
equal than the corresponding bit in $b$ and denoted $a\bwleq b$. 
The equation $a\bwleq b\toot\binomial ba\text{ is odd}$\footnote{where
$\binomial ba$ denotes the binomial coefficient with the usual
convention that $\binomial ba=0$ when $a > b$.} gives a
Diophantine representation for $a\bwleq b$ and then
bitwise operators are derived from $\bwleq$ in
combination with regular addition $\plus$,
in particular, the \emph{digit by digit A\!N\!D} operation
called ``projection\rlap.''\enspace 
To obtain  that $a\bwleq b$ holds if and only if 
$\modeq{\binomial ba}{1}2$,
we prove Lucas's theorem~\cite{lucas1878}
which allows for the computation of the binomial
coefficient in base $p$. It states that
$\modeq{\binomial ba}{\binomial{b_n}{a_n}\times\cdots\times\binomial{b_0}{a_0}} p$
holds when $p$ is prime and $a=a_np^n+\cdots+a_0$ and 
$b=b_np^n+\cdots+b_0$ are the respective 
base $p$ representations of $a$ and $b$;
see Appendix~\ref{append:lucas} for an elementary combinatorial proof of
Lucas's theorem.

A Diophantine representation of the binomial
coefficient can be obtained via e.g.\ the binomial
theorem: $\binomial nk$ is the $k$-th
digit of the development of $(1+q)^n=\sum_{i=0}^n \binomial niq^i$ in
base $q=2^{n+1}$. This gives a Diophantine representation
using the \cstfont{is\_digit} relation of \Cref{coq:dio_rel_rel_iter}.

The rest of the admissibility proof for bounded
universal quantification $\forall i,\, i < n \cfun A$ is a very nice encoding of 
vectors of natural numbers of type $\nat^n$ into natural numbers $\nat$
such that regular addition $\plus$ (resp.\ multiplication 
$\mult$) somehow performs parallel/simultaneous additions (resp.\ multiplications)
on the encoded vectors. More precisely, a vector
$(a_1,\ldots,a_n)\in[0,2^q-1]^n$ of natural numbers is encoded 
as the ``cipher'' $a_1r^2+a_2r^4+a_3r^8+\cdots+a_nr^{2^n}$ with $r=2^{4q}$.
In these \emph{sparse ciphers,} only the digits occurring at $r^{2^i}$ are non-zero.
We remark that none of the parameters, including $n$ or $q$, are constant in the encoding.

Besides the low-level inductive proof of Lucas's theorem presented in Appendix~\ref{append:lucas},
the essential library for the removal of bounded universal quantification
consists of tools to manipulate the type $\nat$ simultaneously and smoothly both 
as (a)~usual natural numbers and (b)~sparse base $r=2^{4q}$ 
encodings of vectors of natural numbers in $[0,2^q-1]$. 
 Notice that $r$ is defined as $r=2^{2q}$ in~\cite{Matiyasevich1997} 
but we favour the alternative choice $r=2^{4q}$ which allows for an easier soundness
proof for vector multiplication because there is no need to manage for 
digit overflows (see Appendix~\ref{overflow_appendix}).

A significant step in the Diophantine
encoding of $\plus$ and $\mult$ on $\nat^n$ 
is the Diophantine encoding of $u=\sum_{i=1}^n r^{2^i}$ and $u'=\sum_{i=2}^{n+1} r^{2^i}$ as the 
ciphers of the constant vectors $[1;\ldots;1]\in\nat^n$ and
 $[0;1;\ldots;1]\in\nat^{n+1}$ respectively, obtained by
masking $u^2$ with $w=\sum_{i=0}^{2^{n+1}}\!\!r^i$ and $2w$. 

Finally, it should be noted that
prior to the elimination of the quantifier in $\forall i,\, i < n \cfun A$,
the Diophantine formula $A$ is first normalised into a
conjunction of elementary constraints using \Cref{coq:dio_formula_elem},
and then the elimination is performed on that list of elementary constraints,
encoding e.g.\ $v_0\dequal v_1\dplus v_2$ and $v_0\dequal v_1\dmult v_2$ with their respective
sparse cipher counterparts.

\section{Minsky Machines Reduce to \texorpdfstring{\FRACTRAN}{FRACTRAN}}

\newcommand{\mmstep}[5]{\sssstepX{\!M}{#1}{(#2,#3)}{(#4,#5)}}
\newcommand{\mmsteps}[6]{\sssstepsX{\!M}{#1}{(#2,#3)}{#4}{(#5,#6)}}
\newcommand{\mmeval}[6]{\ssscomputeX{\!M}{(#1,#2)}{(#3,#4)}{(#5,#6)}}
\newcommand{\mmoutput}{\sssoutputX{\!M}}
\newcommand{\mmterm}{\sssterminatesX{\!M}}
\newcommand{\outcode}{\cstfont{out}}
\newcommand\PC{PC}

In previous work, we have reduced the halting problem for Turing machines to \M{PCP}~\cite{FHS18} and on to a specialised halting problem for Minsky machines~\cite{FLW19} in Coq.
The specialised halting problem asked whether a machine on a given input halts in a configuration with all registers containing zeros.
In order to define Minsky machine recognisability, we consider a general halting problem which allows \textit{any} final configuration,
final meaning that computation cannot further proceed.
The adaptation of the formal proofs reducing \M{PCP} via binary stack machines to Minsky machines is quite straightforward and reuses the certified compiler for low-level languages defined in~\cite{FLW19}.

We first show that one can remove self loops from Minsky machines, i.e.\ instructions which jump to their own location, using the compositional reasoning techniques developed 
in~\cite{FLW19}.
We then formalise the \FRACTRAN language~\cite{Conway1987} and show how the halting problem for Minsky machines can be encoded into the halting problem for \FRACTRAN programs.
While the verification of Minsky machines can be complex and needs preliminary thoughts on compositional reasoning, the translation from Minsky machines to \FRACTRAN is elementary and needs no heavy machinery.

\subsection{Minsky Machines}

\label{mm_sect}

\setCoqFilename{MinskyMachines.MM}

We employ Minsky machines~\cite{Minsky} with instructions
$\iota : \instr_n \bnfdef \INC {(\alpha : \fin n)} \mid \DEC {(\alpha : \fin n)} {(p : \nat)}$.
A Minsky machine with $n$ registers is a sequence of consecutively indexed instructions
$s : \iota_0;~\ldots~s+k : \iota_k;$
represented as a pair $(s : \nat, [\iota_0; \ldots; \iota_k] : \List \instr_n)$.
Its state $(i,\vec v)$ is a program counter (PC) value $i : \nat$ and a vector of values for
registers $\vec v : \nat^n$.
$\INC \alpha$ increases the value of register $\alpha$ and the \PC\ by
one. $\DEC \alpha p$ decreases the value of register $\alpha$ by one
if that is possible and increases the \PC, or, if the
register is already $0$, jumps to \PC\ value $p$.
Given a Minsky machine $(s,P)$, we write $\mmsteps{(s,P)} {i_1} {\vec v_1} k {i_2} {\vec v_2}$ when $(s,P)$ transforms state $(i_1,{\vec v_1})$ into $(i_2,{\vec v_2})$ in $k$ steps 
of computation. 
For $(s,P)$ to do a step in state $(i, \vec v)$ the instruction at label $i$ in $(s, P)$ is considered.
We define 
{\setCoqFilename{Shared.Libs.DLW.Code.subcode}%
$\coqlink[out_code]{\outcode~i~(s,P)}\cdef i < s \ldisj \clength P+s \leq i$}
 to characterize when
label $i$ is outside of the code of $(s,P)$. In that case (and only that case), no computation step can occur. 
We define the halting problem for Minsky Machines as 
$$\begin{array}{c}
\coqlink[MM_HALTING]{\M{MM}\bigl(n : \nat, P : \List{\instr_n}, \vec v : \nat^n\bigr) \cdef \mmterm{(1,P)}{(1,\vec v)}}
\\[1ex]
\null\text{where}~
  \left\{\begin{array}{@{}l@{}} 
      \mmoutput{(s,P)}{(i,\vec v)}{(j,\vec w)}\cdef \exists k,\,\mmsteps{(s,P)} {i} {\vec v} k {j} {\vec w}\lconj \outcode~j~(s,P)\\
      \mmterm{(s,P)}{(i,\vec v)}\cdef \exists j\,\vec w,\,\mmoutput{(s,P)}{(i,\vec v)}{(j,\vec w)}
  \end{array}\right.
\end{array}$$
meaning that the machine $(1,P)$ has a terminating computation starting at state $(1,\vec v)$,
the value of the final state being irrelevant.
Notice that since $\M{MM}(n,P,\vec v)$ considers only Minsky machines starting at \PC\ value $1$, the 
\PC\ value $0$ is always outside of their code (i.e.\ $\outcode~0~(1,P)$ holds), hence
computations can be halted by jumping there. 
We refer to~\cite{FLW19} for a more in-depth formal description of those counter machines.
Note that the halting problem defined there as
\coqlink[MM_HALTS_ON_ZERO]{$\M{MM}_0\bigl(n : \nat, P : \List{\instr_n}, \vec v : \nat^n\bigr) \cdef \mmoutput{(1,P)}{(1,\vec v)}{(0,\vec 0)}$}
is more specific than the problem $\M{MM}$ above defined but
\setCoqFilename{MinskyMachines.Reductions.PCPb_to_MM}%
both are \coqlink[PCPb_MM_HALTS_ON_ZERO]{proved undecidable in our library}.

\setCoqFilename{MinskyMachines.MM.mm_no_self}

We say that a machine \emph{has a self loop} if it contains an instruction of the form $i : \DEC \alpha i$, i.e.\ 
jumps to itself in case the register $\alpha$ has value $0$, leading necessarily to non-termination (in case the \PC\
reaches value $i$).
For every machine $P$ with self loops, we can construct an equivalent machine $Q$ using one additional register 
$\alpha_0$ with constant value $0$, which has the same behaviour but no self loops.
Since the effect of a self loop 
$i:\DEC \alpha i$ is either
decrement and move to the next instruction at $i+1$ if $\alpha>0$
or else enter in a forever loop at $i$, it is easily simulated by a jump
to a length-2 cycle, i.e.\ replacing $i:\DEC \alpha i$ with
$i:\DEC \alpha j$ and adding $j:\DEC {\alpha_0} {(j+1)}; j+1:\DEC {\alpha_0} j$
somewhere near the end of the program. 

\begin{theorem}[][mm_remove_self_loops]
  Given a Minsky machine\/ $P$ with\/ $n$ registers one can compute a machine\/ $Q$ with\/ 
  $1+n$ registers and no self loops s.t.\ for any\/ $\vec v$, 
  $$\mmterm{(1,P)}{(1,\vec v)}~\toot~\mmterm{(1,Q)}{(1,0\ccons\vec v)}.$$
\end{theorem}

\begin{proof}
We explain how any Minsky machine $(1,P)$ with $n$ registers can be transformed into an equivalent one
that uses an extra 0 valued spare register $\alpha_0=0\in\fin{1+n}$ and avoids self loops.
Let $k$ be the length of $P$ and let $P'$ be the Minsky machine with $1+n$ registers
defined by performing a 1-1 replacement of instructions of $(1,P)$: 
\begin{itemize}
\item instructions of the form $i : \INC \alpha$ are replaced by $i : \INC {(1+\alpha)}$;
\item self loops $i : \DEC \alpha i$ are replaced by $i : \DEC {(1+\alpha)} (2+k)$;
\item proper inside jumps $i : \DEC \alpha j$ for $i \neq j$ and $1\leq j\leq k$ are replaced by $i : \DEC {(1+\alpha)} j$;
\item and outside jumps $i : \DEC \alpha j$ for $j=0\ldisj k < j$ are replaced by $i : \DEC {(1+\alpha)} 0$.
\end{itemize} 
Then we define $Q\cdef P'\capp[\DEC {\alpha_0} 0;\DEC {\alpha_0} {(3+k)}; \DEC {\alpha_0} {(2+k)} ]$.
Notice that $P'$ is immediately followed $\DEC {\alpha_0} 0$, 
i.e.\ by an unconditional jump to $0$ (because $\alpha_0$ has value $0$),
and that $(1,Q)$ ends with the length-2 cycle composed of $2+k:\DEC {\alpha_0} {(3+k)}; 3+k:\DEC {\alpha_0} {(2+k)}$.
We show that $(1,Q)$ is a program without self loops (obvious) that satisfies the required simulation equivalence.
Indeed, self loops are replaced by jumps to the length-2 cycle that uses the unmodified
register $\alpha_0$ to loop forever. One should just be careful that the outside jumps of $(1,P)$ 
do not accidentally fall into that cycle and this is why we redirect them all to PC value $0$,
which halts the computation because $\outcode~0~(1,Q)$ holds.
\end{proof}

A predicate \emph{$R:\nat^n\cfun\Prop$ is $\M{MM}$-recognisable} if there exist $m :\nat$ and a
Minsky machine $P : \List{\instr_{n + m}}$ of $(n+m)$ registers such that for any $\vec v : \nat^n$ we have $R~\vec v\toot\mmterm{(1,P)}{(1,\vec v\capp\vec 0)}$. The last $m$ registers serve as spare registers during the computation. Notice that
not allowing for spare registers would make e.g.\ the empty predicate un-recognisable\rlap.%
\footnote{For any Minsky machine $(1,P)$, if it starts on large enough register values, 
for instance if they are all greater than the length of $P$, 
then no jump can occur and the machine terminates after its last instruction executes.
Such unfortunate behavior can be circumvented with a $0$-valued spare register.}\enspace
It is possible to limit the number of (spare) registers but that question is not essential in our development.

\subsection{The FRACTRAN language}

\newcommand{\fraction}[2]{{#1}/{#2}}
\newcommand{\ftstep}{\sssstepX{\!F}}
\newcommand{\ftterm}{\sssterminatesX{\!F}}
\newcommand{\ftcompute}{\ssscomputeX{\!F}}

\newcommand{\myprimes}[2]{\mathfrak{#1}_{#2}}
\newcommand{\pprime}[1]{\myprimes{p}{#1}}
\newcommand{\qprime}[1]{\myprimes{q}{#1}}

\setCoqFilename{FRACTRAN.FRACTRAN}

We formalise the language \FRACTRAN, introduced as a universal programming language for arithmetic by
Conway~\cite{Conway1987}.
A \FRACTRAN program $Q$ consists of a list of positive fractions 
$[\fraction {p_1} {q_1}; \ldots; \fraction {p_n} {q_n}]$.
The current state of a \FRACTRAN program is just a natural number~$s$.
The first fraction $\fraction{p_i}{q_i}$ in $Q$ such that $s\cdot(\fraction{p_i}{q_i})$ is still integral determines the successor state, 
which then is $s\cdot(\fraction{p_i}{q_i})$.
If there is no such fraction in $Q$, the program terminates.
 
We make this precise inductively for $Q$ 
being a list of fractions $\fraction pq:\nat \times \nat$:
\[
  \inferrule{q \cdot y = p \cdot x}{\ftstep{(\fraction p q\ccons Q)}{x}{y} }
  \qquad\qquad
  \inferrule{q \ndivides p \cdot x \and \ftstep Q x y}{\ftstep{(\fraction p q\ccons Q)}{x}{y}}
\]
i.e.\ at state $x$ the first fraction $\fraction p q$ in $Q$ where $q$
divides $p\cdot x$ is used, and $x$ is multiplied by $p$ and divided
by $q$. For instance, the \FRACTRAN program $[\fraction 5 7; \fraction 2 1]$ runs forever when starting 
from state $7$, producing the sequence $5 = 7\cdot(\fraction57)$, $10 = 5\cdot(\fraction21)$, 
$20 = 10\cdot(\fraction21)$ ...\footnote{No \FRACTRAN program can ever stop when it contains a 
fraction having an integer value like $\fraction 21$.}

We say that a \FRACTRAN program $Q = [\fraction {p_1}{q_1}; \ldots; \fraction{p_n}{q_n}]$ is \emph{regular} if none of its denominators is $0$, i.e.\ if $q_1 \neq 0,\ldots,q_n\neq 0$.
For a \FRACTRAN program $Q : \List{(\nat \times \nat)}$ and $s:\nat$, 
we define the decision problem as 
the question ``does $Q$ halt when starting from $s$'':
\[
  \coqlink[FRACTRAN_HALTING]{\FRACTRAN(Q,s)\cdef \ftterm{Q}{s}}\quad\text{with}~
  \ftterm{Q}{s} \cdef \exists x,\,\ftcompute Q s x\lconj\forall y,\,\lneg\ftstep Q x y
\]
Following~\cite{Conway1987}, we now show how (regular) \FRACTRAN halting can be used to simulate Minsky machines halting.
The idea is to use a simple Gödel encoding of the states of a Minsky machine.
We first fix two infinite sequences of prime numbers $\pprime 0, \pprime 1, \dots$ and $\qprime 0,\qprime 1, \dots$ all distinct from each other.
We define the encoding of $n$-register Minsky machine states as 
$\ol {(i,\vec v)} \cdef \pprime i\qprime 0^{x_0} \cdots \qprime {n-1}^{x_{n-1}}$ where 
$\vec v = [x_0, \dots, x_{n-1}]$:
\begin{itemize}
\item
To simulate the step semantics of Minsky machines for $i:\INC \alpha$, we divide the encoded state 
by $\pprime i$ and multiply by $\pprime{i+1}$ for the change in PC value, and increment the register $\alpha$ 
by multiplying with $\qprime\alpha$, hence we add the fraction 
$\fraction{\pprime{i+1}\qprime\alpha}{\pprime i}$;
\item
To simulate $i:\DEC \alpha j$ when $\vec v_\alpha = 1+n$ we divide by $\pprime i$, 
multiply by $\pprime{i+1}$ and decrease register $\alpha$ by dividing by $\qprime\alpha$, hence 
we add the fraction $\fraction{\pprime{i+1}}{\pprime i\qprime\alpha}$;
\item
To simulate $i:\DEC \alpha j$ when $\vec v_\alpha = 0$ we divide by $\pprime i$ and multiply by 
$\pprime j$. To make sure that this is only executed when the previous rule does not apply, 
we add the fraction $\fraction{\pprime j}{\pprime i}$ \emph{after} the fraction
$\fraction{\pprime{i+1}}{\pprime i\qprime\alpha}$.
\end{itemize}
In short, we define the encoding of labelled instructions and then programs as
\[\begin{array}{@{}r@{~\cdef~}l@{}}
      \ol{(i,\INC \alpha)} & [\fraction{\pprime{i+1}\qprime\alpha}{\pprime i}] \\
\ol{(i,\DEC \alpha j)} & [ \fraction{\pprime{i+1}}{\pprime i\qprime\alpha}; \fraction{\pprime j}{\pprime i} ]
\end{array}\qquad
\ol {(i, [\iota_0; \ldots; \iota_k])} \cdef 
\ol{(i,\iota_0)}  \capp \cdots \capp \ol{(i+k,\iota_k)}.\]
Notice that we only produce regular programs and
that a self loop like $i:\DEC \alpha i$, jumping on itself when $\vec v_\alpha = 0$,
will generate the fraction $\fraction{\pprime i}{\pprime i}$ potentially capturing any
state $\ol {(j,\vec v)}$ even when $j\neq i$.
So this encoding does not work on Minsky machines containing self loops because the
corresponding \FRACTRAN\ program would never terminate, even when the \PC\ never 
reaches self loops. 

\setCoqFilename{FRACTRAN.FRACTRAN.mm_fractran}

\begin{lemma}[][mm_fractran_simulation]
  If\/ $(1,P)$ has no self loops then\/ $\mmterm{(1,P)}{(1,\vec v)}\toot\ftterm{\ol{(1,P)}}{\ol{(1,\vec v)}}$.
\end{lemma}

\begin{proof}
  Let $(i,P)$ be a Minsky machine with no self loops. 
  We show that the simulation of $(i,P)$ by $\ol{(i,P)}$ is 1-1, 
  i.e.\ each step is simulated by one step. We first show the forward
  simulation, i.e.\ that
  $\mmstep{(i,P)}{i_1}{\vec v_1}{i_2}{\vec v_2}$ entails
  $\ftstep{\ol{(i,P)}}{\ol{(i_1,\vec v_1)}}{\ol{(i_2,\vec v_2)}}$,
  by case analysis.
  Conversely we show that if $\ftstep{\ol{(i,P)}}{\ol{(i_1,\vec v_1)}}{\mathit{st}}$
  holds then $\mathit{st} = \ol{(i_2,\vec v_2)}$ for some $(i_2, \vec v_2)$ 
  such that $\mmstep{(i,P)}{i_1}{\vec v_1}{i_2}{\vec v_2}$. Backward simulation
  involves the totality of \M{MM} one step semantics and the determinism of regular 
  \FRACTRAN one step semantics combined
  with the forward simulation.

  Using these two simulation results, the desired equivalence follows by induction 
  on the length of terminating computations.
\end{proof}

\begin{theorem}[][mm_fractran_n]
For any\/ $n$-register Minsky machine\/ $P$ one can compute a regular \FRACTRAN program\/
$Q$ s.t.\/ $\mmterm{(1,P)}{(1,[x_1;\ldots;x_n])}\toot\ftterm Q{\pprime 1\qprime 1^{x_1}\cdots\qprime n^{x_n}}$
holds for any\/ $x_1,\ldots,x_n$.
\end{theorem}

\begin{proof}
Using \Cref{coq:mm_remove_self_loops}, 
we first compute a Minsky machine $(1,P_1)$ equivalent to $(1,P)$ but with one extra 0-valued spare register
and no self loops. Then we apply~\Cref{coq:mm_fractran_simulation} to $(1,P_1)$
and let $Q\cdef\ol{(1,P_1)}$. The program $Q$ is obviously regular and 
given $\vec v=[x_1;\ldots;x_n]$, the encoding of
the starting state $(1,0\ccons\vec v)$ for $(1,P_1)$  
is $\pprime 1\qprime 0^0\qprime 1^{x_1}\cdots\qprime n^{x_n}$ hence the result.
\end{proof}

This gives us a formal constructive proof that (regular) \FRACTRAN is Turing complete as a model of computation
and is consequently undecidable.

\setCoqFilename{FRACTRAN.FRACTRAN_undec}
\begin{corollary}[][Fractran_UNDEC]
  $\M{Halt}$ reduces to $\FRACTRAN$.
\end{corollary}

\begin{proof}
\Cref{coq:mm_fractran_n} gives us a reduction from $\M{MM}$ to $\FRACTRAN$
which can be combined with the reduction of $\M{Halt}$ to $\M{PCP}$ from~\cite{FHS18}
and a slight modification of $\M{PCP}$ to $\M{MM}$ from~\cite{FLW19}.
\end{proof}

\section{Diophantine Encoding of \FRACTRAN}

\setCoqFilename{H10.Fractran.fractran_dio}

We show that a single step of \FRACTRAN computation is a Diophantine relation.

\begin{lemma}[][dio_rel_fractran_step]
  For any\/ \FRACTRAN program\/ $Q:\List(\nat\times\nat)$, one can compute a map
  $$\forall f\,g:(\Var\cfun\nat)\cfun\nat,\,\diofun\,f\cfun\diofun\,g\cfun \diorel\,(\abst \gval {\,\ftstep Q {f\,\gval} {g\,\gval}}).$$
\end{lemma}

\begin{proof}
The map is built by induction on $Q$.
If $Q=\cnil$, then we show ${\ftstep \cnil {f\,\gval} {g\,\gval}}\toot\False$, and 
thus $\diorel\,(\abst \gval { \ftstep Q {f\,\gval} {g\,\gval}})$ by map~4 of
Proposition~\ref{diorel_prop1} followed by \dioauto.
If $Q$ is a composed list $Q = \fraction p q \ccons Q'$, then we show the equivalence
\[{\Bigl(\ftstep {\bigl(\fraction p q \ccons Q'\bigr)} {f\,\gval} {g\,\gval}\Bigr)}~~\toot~~{\Bigl(q\cdot(g\,\gval) = p\cdot(f\,\gval)\Bigr) \ldisj \Bigl(q \ndivides \bigl(p\cdot(f\,\gval)\bigr)\Bigr) \lconj \Bigl(\ftstep {Q'} {f\,\gval} {g\,\gval}\Bigr)}\]
and we derive $\diorel\,(\abst \gval { \ftstep Q {f\,\gval} {g\,\gval}})$
by map~4 of Proposition~\ref{diorel_prop1} followed by \dioauto, 
the induction hypothesis being used locally as a hint for the tactic.
\end{proof}

In addition, the ``$Q$ has terminated at $x$'' predicate 
is Diophantine for any \FRACTRAN program $Q$. The proof 
is similar to the previous one:

\begin{lemma}[][dio_rel_fractran_stop]
  For any\/ \FRACTRAN program\/ $Q:\List(\nat\times\nat)$, one can compute a map
 $$\forall f:(\Var\cfun\nat)\cfun\nat,\,\diofun\,f\cfun \diorel\,(\abst \gval {\forall y,\lneg\ftstep Q {f\,\gval} y}).$$
\end{lemma}

\begin{proof}
The map $\forall f,\,\diofun\,f\cfun \diorel\,(\abst \gval {\forall y,\lneg{\ftstep Q {f\,\gval} y}})$ 
is built by induction on $Q$. 
If $Q=\cnil$, then we show $(\forall y,\lneg{\ftstep \cnil {f\,\gval} y})\toot\True$, and 
thus $\diorel\,(\abst \gval {\forall y,\lneg{\ftstep Q {f\,\gval} y}})$ by map~4 of
Proposition~\ref{diorel_prop1} followed by \dioauto.
If $Q = \fraction p q \ccons Q'$, then we show the equivalence
\[\Bigl({\forall y,\lneg{\ftstep Q {f\,\gval} y}}\Bigr)~~\toot~~
\Bigl(q\ndivides \bigl(p\cdot(f\,\gval)\bigr)\Bigr) \lconj \Bigl(\forall y,\lneg{\ftstep {Q'} {f\,\gval} y}\Bigr)\]
and we get $\diorel\,(\abst \gval {\forall y,\lneg{\ftstep Q {f\,\gval} y}})$
by map~4 of Proposition~\ref{diorel_prop1} followed by \dioauto,
the induction hypothesis being used as a hint again.
\end{proof}

We can now deduce a core result of the paper which states that
\FRACTRAN programs have Diophantine termination predicates.

\setCoqFilename{H10.Fractran.fractran_dio}
\begin{theorem}[][FRACTRAN_HALTING_on_diophantine]
  If\/ $Q:\List(\nat\times\nat)$ is a \FRACTRAN program then
  one can compute a map
  $$\forall f:(\Var\cfun\nat)\cfun\nat,\, \diofun\,f\to 
  \diorel\,(\abst \gval {\,\ftterm Q{f\,\gval}}).$$
\end{theorem}

\begin{proof}
  By definition we have $\ftterm Q{f\,\gval}\toot\exists x\,(\ftcompute Q {f\,\gval} x\lconj\forall y,\,\lneg\ftstep Q x y)$
  and hence we obtain the claim using \Cref{coq:dio_rel_rt} together with \Cref{coq:dio_rel_fractran_step}
  and \Cref{coq:dio_rel_fractran_stop}.
\end{proof}

We conclude with the undecidability of Hilbert's tenth problem
by a reduction chain starting from the Halting problem for single tape Turing 
machines:

\setCoqFilename{H10.H10_undec}
\begin{theorem}[Hilbert's tenth problem][Hilberts_Tenth]
We have the following reduction chain
\[         \M{Halt}
\reducesto \M{PCP}
\reducesto \M{MM}
\reducesto \FRACTRAN
\reducesto \M{DIO\_FORM}
\reducesto \M{DIO\_ELEM}
\reducesto \M{DIO\_SINGLE}
\reducesto \M{H10}\]
and as a consequence, \coqlink[H10_undec]{\Hten is undecidable}.
\end{theorem}

\begin{proof} The proof combines previous results as Theorems~\ref{coq:mm_fractran_n} and~\ref{coq:FRACTRAN_HALTING_on_diophantine} and \Cref{coq:dio_rel_single}. \end{proof}

\section{The Davis-Putnam-Robinson-Matiyasevich Theorem}

We give a proof of an instance of the DPRM theorem stating that
recursively enumerable predicates are Diophantine\rlap.\footnote{By instance,
we mean that the DPRM is an open theorem bound to be extended
for any newly proposed Turing complete model of computation.}\enspace 
Here we assume that the informal notion of ``recursive enumerability''
(justified by Church's thesis) can be characterised by Minsky machines recognisability
as defined in Section~\ref{mm_sect}.

\setCoqFilename{H10.Fractran.fractran_dio}

\begin{proposition}[][fractran_exp_diophantine]
The Gödel encoding is Diophantine, i.e.\ we have
$\diofun\,(\abst\gval{\qprime 1^{\gval\,x_0}\cdots\qprime n^{\gval\,x_{n-1}}})$.
\end{proposition}

\begin{proof}
By induction on $n:\nat$ using Proposition~\ref{prop:diofun_basic} and \Cref{coq:dio_fun_expo}.
Notice that the $\qprime i$'s are hard-coded in the Diophantine representation,
which means we do not need to encode the algorithm that actually computes them,
which would otherwise be very painful.
\end{proof}

\begin{lemma}[][FRACTRAN_HALTING_on_exp_diophantine]
  For any \FRACTRAN program $Q$ we have
  $\diorel\,\bigl(\abst \gval{\,\ftterm Q{\pprime 1\qprime 1^{\gval\,x_0}\cdots\qprime n^{\gval\,x_{n-1}}}}\bigr).$
\end{lemma}

\begin{proof}
By \Cref{coq:FRACTRAN_HALTING_on_diophantine}, we only have to show
$\diofun\,f$ for $f\,\gval\defeq \pprime 1\qprime 1^{\gval\,x_0}\cdots\qprime n^{\gval\,x_{n-1}}$.
This follows from  Propositions~\ref{prop:diofun_basic} and~\ref{coq:fractran_exp_diophantine}.
\end{proof}

To simplify the notation $\semb p{\vec v}{\vec w}$ below, we abusively identify 
the vector $\vec v:\nat^n$ (resp.\ $\vec w:\nat^m$) with the valuation 
$\abst{(i:\fin n)}{\vec v_i}$ (resp.\ $\abst{(j:\fin m)}{\vec w_j}$)
that accesses the components of the vector $\vec v$ (resp.\ $\vec w$).

\setCoqFilename{H10.DPRM}
\begin{theorem}[DPRM][DPRM_n]
  Any $\M{MM}$-recognisable relation\/ $R:\nat^n\cfun\Prop$ is Diophantine: 
  one can compute a single Diophantine equation\/ $p\dequal q:\dsingle(\fin m,\fin n)$ 
with\/ $n$ parameters and\/ $m$ variables s.t.\/ $\forall \vec v : \nat^n,\, R~\vec v \toot \exists \vec w : \nat^m,\,
\semb p{\vec v}{\vec w} = \semb q{\vec v}{\vec w}$.
\end{theorem}

\begin{proof}
  By definition, $R:\nat^n\cfun\Prop$ is recognised by some Minsky 
  machine $P$ with $(n+m)$ registers, i.e.\ $R~\vec v \toot \mmterm{(1,P)}{(1,\vec v\capp\vec 0)}$.
  By \Cref{coq:mm_fractran_n}, we compute a \FRACTRAN program $Q$ s.t.\ 
  $\mmterm{(1,P)}{(1,[v_1;\ldots;v_n;w_1;\ldots;w_m])}\toot\ftterm 
  Q{\pprime 1\qprime 1^{v_1}\cdots\qprime n^{v_n}
    \qprime {n+1}^{w_1}\cdots\qprime{n+m}^{w_m}}$. 

  Hence
  we deduce $R~[v_1;\ldots;v_n] \toot \ftterm Q{\pprime 1\qprime 1^{v_1}\cdots\qprime n^{v_n}}$.
  As a consequence, the relation  $\abst\gval{R~[\gval\,x_0;\ldots;\gval\,x_{n-1}]}$ is Diophantine
  by \Cref{coq:FRACTRAN_HALTING_on_exp_diophantine}. 
  By \Cref{coq:dio_rel_single},
  there is a Diophantine equation $p\dequal q:\dsingle(\nat,\Var)$ such that
  $R~[\gval\,x_0;\ldots;\gval\,x_{n-1}]\toot \exists\genv,\, \semb p\gval\genv = \semb q\gval\genv$.
  Notice that the value $\gval\,x_i$ of any parameter of $p\dequal q$ greater than 
  $x_n$ does not influence its solvability.

  Now let $m$ be an upper bound of the number of (existentially quantified) variables 
  in $p\dequal q$. We injectively map those variables in $\fin m$ and we project
  the parameters of $p\dequal q$  onto $\fin n$ by replacing every parameter
  greater than $x_n$ with the $0$ constant. We get a Diophantine equation 
  $p'\dequal q':\dsingle(\fin m,\fin n)$ of which the solvability at $\vec v$ is equivalent 
  to $R~\vec v$.
\end{proof}

\section{Hilbert's Tenth Problem Over Integers}
\label{sec:lagrange}

In our formalisation, polynomials are defined over natural numbers, that is both constants and solutions come from $\nat$.
The standard way to extend the undecidability of $\M{H10}$ to a formalisation based on integers is via Lagrange's theorem, stating that an integer is positive if and only if it is the sum of four squares.

Similar to Definition~\ref{dio_polynomial_nat}, we define define polynomials over integers:

\setCoqFilename{H10.H10Z}
\begin{definition}[][dio_polynomial]
  The type of\/ \emph{Diophantine polynomials} $\dpolyz(\BVar,\Var)$ over\/ $\Z$ is defined by:
  \[p,q : \dpolyz(\BVar,\Var) \bnfdef u : \BVar \mid \genvar :\Var \mid z :\Z \mid p\dplus q\mid p\dmult q. \]
\end{definition}

The interpretation of a polynomial $p: \dpolyz(\BVar,\Var)$ in $\Z$ given $\genv:\BVar\to\Z$, $\gval:\Var\to\Z$ and denoted
$\semb p\gval\genv$, is defined in the obvious way. Again, if $p:\dpolyz(\fin m,\fin n)$ we abusively write 
$\semb p{\vec v}{\vec w}$ when $\vec v:\nat^n$ and $\vec w:\nat^m$.
We can then define 
$$\coqlink[H10Z]{\Htenz}\bigl(n, p : \dpolyz(\fin n, \fin 0)\bigr) := \exists \vec w : \nat^n,\, \semb{p}\cnil{\vec w} = 0$$
that is ``does the polynomial equation $p(x_1,\ldots,x_n)=0$ in (at most) $n$ variables have a solution in 
$\Z$\rlap.''\enspace
We first outline a proof of Lagrange's theorem and then reduce \Hten to \Htenz.

\subsection{Lagrange's theorem}

The proof we have implemented roughly follows 
the ``classical proof'' in {\color{ACMDarkBlue}\href{https://en.wikipedia.org/wiki/Lagrange's_four-square_theorem}{Wikipedia's account}} of the theorem. Their use of the ``classical'' qualifier should be understood 
as typical/standard, and certainly not as opposed to constructive/intuitionistic.
The below proof perfectly fits in our constructive setting.

\setCoqFilename{H10.ArithLibs.lagrange}

\begin{proposition}[Euler, 1748][Euler_squares]
\label{prop:euler}
Let us assume the two equations\/ 
      $n = a_1^2+b_1^2+c_1^2+d_1^2$ 
and\/ $m = a_2^2+b_2^2+c_2^2+d_2^2$ hold
in\/ $\Z$. Let us define the four relative
intergers:
$$\begin{array}{c@{\qquad}c}
  a\defeq a_1a_2+b_1b_2+c_1c_2+d_1d_2 
& b\defeq a_1b_2-b_1a_2+d_1c_2-c_1d_2 \\
  c\defeq a_1c_2-c_1a_2+b_1d_2-d_1b_2
& d\defeq a_1d_2-d_1a_2+c_1b_2-b_1c_2. \\
\end{array}$$ 
Then the identity\/ $n m = a^2+b^2+c^2+d^2$ holds.
\end{proposition}

\begin{proof}
This holds in any commutative ring and the
proof just calls the \cstfont{ring} tactic. 
\end{proof}

Because of \Cref{coq:Euler_squares}, ``being the sum of four squares'' is
a multiplicative property. Hence to show that it holds for every
natural number, it is enough to establish it for primes.

\setCoqFilename{Shared.Libs.DLW.Utils.prime}

\begin{theorem}[Prime induction][prime_rect]
\label{thm:primeind}
Let\/ $P:\nat\cfun\Prop$ be a predicate.\/ To establish\/ $\forall n:\nat,\,P\,n$,
it is enough to prove these four induction steps: 
$$ P\,0\qquad P\,1
\qquad  \forall a\,b:\nat,\, P\,a\cfun P\,b\cfun P(a b)
\qquad  \forall p:\nat,\,\cstfont{prime}\,p\cfun P\,p.$$ 
\end{theorem}

\begin{proof}
This is one possible form of the \emph{fundamental theorem of arithmetic}.
For the proof, first show by strong induction on $n:\nat$ that 
one can discriminate whether $n<2$ or compute a prime factor 
of $n$, including the possibility that $n$ itself is prime. 
Then, to prove the prime induction principle, 
proceed by strong induction again.
\end{proof}

\setCoqFilename{H10.ArithLibs.lagrange}
\begin{lemma}[][lagrange_prelim]
\label{lemma:langrage_prime}
If\/ $p:\nat$ is prime then\/ $\exists n\,a\,b:\nat,\,n p = 1+a^2+b^2\lconj 0<n<p$.
\end{lemma}

\begin{proof}
Let us first rule out the case $p=2$ which has an obvious solution.
So let us write $p=2m+1$ because all the other primes are odd. 
In the field $\Zp p$, let us study the modular equation $\modeq{a^2}{-(1+b^2)}p$.
Because $\Zp p$ has no zero divisor, 
the map $a\mapsto a^2$ from $[0,m]\cfun \Zp p$
is injective. As a consequence, so is the map
$b\mapsto -(1+b^2)$ from $[0,m]\cfun \Zp p$.
Hence none of the two lists $\bigl[a^2:\Zp p\mid a\in[0,m]\bigr]$ and
$\bigl[-(1+b^2):\Zp p\mid b\in[0,m]\bigr]$ contain a duplicate.
Since they have a combined length of $2(m+1)>2m+1$ and their
concatenation is contained in a list of length $p=2m+1$
enumerating $\Zp p$,
by the pigeon hole principle, they must intersect and this gives $a,b\in[0,m]$
such that $\modeq{a^2}{-(1+b^2)}p$. This in turn
gives $n:\nat$ such that $np = 1+a^2+b^2$.
Given that $2a<p$ and $2b<p$, we deduce $0<n<p$.
\end{proof}

Lagrange's theorem gives a Diophantine characterisation of those
relative integers which are positive as the sum of four squares.

\begin{theorem}[Lagrange, 1770][lagrange_theorem_Z]
For any relative integer\/ $z:\Z$, $z$ is positive if and only if\/
there exists\/ $a,b,c,d:\Z$ such that\/ $z=a^2+b^2+c^2+d^2$. 
\end{theorem}

\begin{proof}
It is enough to show that any natural number $n:\nat$ is the sum of four squares,
i.e.\ there exists $a,b,c,d:\Z$ such that $n=a^2+b^2+c^2+d^2$. 
By Euler's four-square identity (see Proposition~\ref{prop:euler})
and the principle of prime induction (see Theorem~\ref{thm:primeind}), 
we only need to show the property for prime numbers.

We fix a prime number $p$ and define the predicate ``$mp$ is the sum of four squares'' as
$$P(m:\nat)\cdef \exists a\,b\,c\,d:\Z,\,m p =  a^2+b^2+c^2+d^2.$$
We want to show that $P\,1$ holds.
By \Cref{coq:lagrange_prelim}, we know that $P\,n$ holds
for some $0<n<p$. We are going to decrease this value of $n$ until
it reaches $1$, i.e.\ to establish $P\,1$, it is enough to give 
\coqlink[lagrange_prime_step]{a proof of}
$$\forall m:\nat,\,1<m<p\cfun P\,m\cfun\exists r:\nat,\, 1\leq r<m\lconj P\,r$$
and then finish the argument by strong induction (implementing ``infinite descent'' here). 

So let us assume $1<m<p$ and $m p=x_1^2+x_2^2+x_3^2+x_4^2$.
For each $i=1,\ldots,4$ we compute a ``small'' representative of $x_i$ in $\Zp m$,
i.e.\ $y_i:\Z$ such that $\modeq{x_i}{y_i}m$ and $4y_i^2\leq m^2$.
Then we have $\modeq{mp=x_1^2+x_2^2+x_3^2+x_4^2}{y_1^2+y_2^2+y_3^2+y_4^2}m$ hence 
$\modeq{y_1^2+y_2^2+y_3^2+y_4^2}0m$ and we get $r:\nat$ such that $y_1^2+y_2^2+y_3^2+y_4^2=r m$.
Because $y_1^2+y_2^2+y_3^2+y_4^2\leq m^2$, we deduce $r \leq m$.
Now we rule out the cases $r=0$ and $r=m$:
\begin{itemize}
\item if $r=0$ then $y_1=\cdots=y_4=0$ and thus $\modeq{x_i^2}0{m^2}$ for $i=1,\ldots,4$.
As a consequence, $\modeq{mp=x_1^2+x_2^2+x_3^2+x_4^2}0{m^2}$ and thus $m$ divides $p$,
contradicting the primality of $p$;
\item if $r=m$ then $y_1^2+y_2^2+y_3^2+y_4^2=m^2$ and since 
$4y_i^2\leq m^2$ holds for $i=1,\ldots,4$, we deduce $4y_1^2=\cdots=4y_4^2=m^2$
and thus $m=2q$ and $y_i=\pm q$ for $i=1,\ldots,4$.
As a consequence $\modeq{x_i^2}{q^2}{m^2}$ and
thus $\modeq{x_1^2+x_2^2+x_3^2+x_4^2}{4q^2}{m^2}$ 
hence $\modeq{mp}0{m^2}$, in contradiction with
the primality of $p$ again.
\end{itemize}

So we have $1\leq r<m$, $r m=y_1^2+y_2^2+y_3^2+y_4^2$
and $\modeq{x_i}{y_i}m$ for $i=1,\ldots,4$.
We also have $m p=x_1^2+x_2^2+x_3^2+x_4^2$ and using
Euler's four squares identity
we get $a,b,c,d:\Z$ such that $(m p)(r m) = a^2+b^2+c^2+d^2$.
Using $\modeq{x_i}{y_i}m$ and the values of
$a,b,c$ and $d$ as defined in Proposition~\ref{prop:euler}, 
we show that $\modeq a0m$, $\modeq b0m$, $\modeq c0m$
and $\modeq d0m$, i.e.\ $m$ divides $a$, $b$, $c$ and $d$. 
Hence $rp = (a/m)^2+(b/m)^2+(c/m)^2+(d/m)^2$
which establishes $P\,r$ for $1\leq r<m$, as required.
\end{proof}

\subsection{\texorpdfstring{\Hten reduces to \Htenz}{H10 reduces to H10Z}}

\newcommand{\dsquare}[1]{{(#1)}\dmult{(#1)}}
\newcommand{\dfoursquare}[4]{\dsquare{#1}\dplus\dsquare{#2}\dplus\dsquare{#3}\dplus\dsquare{#4}}
\newcommand{\dlagrange}[1]{\dfoursquare{4{#1}}{4{#1}+1}{4{#1}+2}{4{#1}+3}}

Let $f : \dpoly(\fin n, \fin 0) \to \dpolyz(\fin{4n}, \fin 0)$ be the function replacing every variable $i:\fin n$ in a polynomial by the 
polynomial expression 
\[ \dlagrange i. \]
The function $f$ can now be used to define the reduction from \Hten to \Htenz:

\setCoqFilename{H10.H10Z_undec}
\begin{corollary}[][H10_H10Z]
  $\Hten \reducesto \Htenz$.
\end{corollary}
\begin{proof}
  Given polynomials $p, q : \dpoly(\fin n,\fin 0)$ the reduction function returns the polynomial $f(p) \dplus (-1) \dmult f(q) : \dpolyz(\fin {4n}, \fin 0)$.
  
  If $\Hten(n,p,q)$ holds, i.e.\ if $\semb{p}{\cnil}{\vec w} = \semb{q}{\cnil}{\vec w}$ for $\vec w = [w_1; \ldots; w_n]$ we know via 
  Lagrange's \Cref{coq:lagrange_theorem_Z} that there exist $a_1, b_1, c_1, d_1, \dots , a_n, b_n, c_n, d_n$ s.t.\ $w_i = a_i^2 + b_i^2 + c_i^2 + d_i^2$ and thus that $\semb{f(p)}\cnil{\vec v} = \semb{p}\cnil{\vec w}$ for 
$\vec v: \Z^{4n} \defeq [a_1; b_1; c_1; d_1; \ldots; a_n; b_n; c_n; d_n]$ as well as $\semb{f(q)}\cnil{\vec v} = \semb{q}\cnil{\vec w}$.
  Thus, $\semb{f(p) \dplus (-1) \dmult f(q)}\cnil{\vec v} = 0$ and $\Htenz\bigl(4n, f(p) \dplus (-1) \dmult f(q)\bigr)$.

  For the other direction, let $\vec v : \Z^{4n}$ be given s.t.\ $\semb{f(p) \dplus (-1) \dmult f(q)}\cnil{\vec v} = 0$.
  We thus know that $\semb{f(p)}\cnil{\vec v} = \semb{f(q)}\cnil{\vec v}$.
  Because $\vec v:\Z^{4n}$, it has the form $\vec v = [a_1; b_1; c_1; d_1; \ldots; a_n; b_n; c_n; d_n]$.
  Then all elements of $\vec w \defeq [a_1^2 + b_1^2 + c_1^2 + d_1^2; \ldots; a_n^2 + b_n^2 + c_n^2 + d_n^2]$ are natural numbers and 
  the equations $\semb{f(p)}\cnil{\vec v} = \semb{p}\cnil{\vec w}$ and $\semb{f(q)}\cnil{\vec v} = \semb{q}\cnil{\vec w}$ hold.
  Thus $\semb{p}\cnil{\vec w} = \semb{q}\cnil{\vec w}$, i.e.\ $\Hten(n, p, q)$.
\end{proof}

\section{\texorpdfstring{$\mu$}{mu}-Recursive Algorithms}
\label{sec:murec}

In order to show that $\M{MM}$, $\M{FRACTRAN}$, and $\M{H10}$ are in the same many-one reduction class (i.e.~interreducible via many-one reductions), we introduce $\mu$-recursive algorithms as intermediate layer.
Programming in this well-known model of computation resembles functional programming (in a first-order language) and we will use it first for a reduction $\M{H10} \preceq \M{$\mu$-rec}$.
Afterwards, we explain a compiler to Minsky machines, yielding a reduction from $\M{$\mu$-rec}$ to $\M{MM}$.
The next section will then connect $\mu$-recursive algorithms to the weak call-by-value $\lambda$-calculus.

\subsection{\mathversion{bold}$\mu$-Recursive Recognisability}

\newcommand{\cst}[1]{\texttt{cst}_{#1}}
\newcommand{\zero}{\texttt{zero}}
\newcommand{\succN}{\texttt{succ}}
\newcommand{\prj}[1]{\texttt{prj}_{#1}}
\newcommand{\comp}[2]{\texttt{comp}\,#1\,#2}
\newcommand{\recA}[2]{\texttt{rec}\,#1\,#2}
\newcommand{\minA}{\mu}

\makeatletter
\def\bi@steplft{\mbox{$-\mskip -4mu\lbrack$}}
\def\bi@steprt{\mbox{$\rangle$}}
\def\bi@stepsym{\mathrel{\bi@steplft\mskip -2mu\bi@steprt}}
\def\bi@step#1#2#3{[{#1};{#2}]\bi@stepsym{#3}}
\def\bi@@stepsym#1{\mathrel{\bi@steplft\mskip 1mu{#1}\mskip 1mu\bi@steprt}}
\def\bi@@step[#1]#2#3#4{[{#2};{#3}]\bi@@stepsym{#1}{#4}}
\newcommand{\bigstep}{\@ifnextchar[\bi@@step\bi@step}
\makeatother

\newcommand{\pfun}{\mathbin{{-}\mskip -3mu{\rightharpoondown}}}

\newcommand{\semra}[1]{\llbracket{#1}\rrbracket}

\setCoqFilename{MuRec.recalg}
We define a type \coqlink[recalg]{$\recalg k$ of $\mu$-recursive algorithms} 
representing $\mu$-recursive partial functions in $\nat^k \pfun \nat$.
\begin{mathpar}
\def\myspace{\mskip 80mu}%
\begin{array}{c}
  \infer{ n : \nat }{ \cst n : \recalg 0 }\myspace
  \infer{~}{\zero : \recalg 1}\myspace
  \infer{~}{\succN : \recalg 1}\myspace
  \infer{p : \fin k}{\prj p : \recalg k}\\[3ex]
  \infer{f : \recalg k \and g : \recalg i^k}{\comp f g : \recalg i}\myspace
  \infer{f : \recalg k \and g : \recalg {2 + k}}{\recA f g : \recalg{1 + k}}\myspace
  \infer{f : \recalg{1 + k}}{\minA f : \recalg k}
\end{array}
\end{mathpar}
We represent these $\mu$-recursive partial functions using the 
 standard \coqlink[ra_rel]{relational semantics $\semra f:\nat^k\to\nat\to\Prop$} of the $\mu$-recursive algorithm $f:\recalg k$, 
 formalised in Figure~\ref{ra_rel_sem} as a fixpoint definition.
Intuitively, $\cst n$ represents the constant $n$ of arity $0$, $\zero$ is the constant $0$-function of arity $1$, $\succN$ the successor function of arity $1$, projection $\prj p$ returns the $p$-th argument, $\comp f g$ where $g$ is a $k$-vector of functions of arity $i$ first applies each element of $g$ to the $i$ inputs and then $f$ to the resulting $k$ numbers.
$\recA f g$ is performing primitive recursion on the first argument.
If the argument is $0$, $f$ is used.
If the argument is $1+n$, $g$ is applied to $n$, the recursive call and the rest of the arguments.
Finally, minimisation $\minA f$ performs unbounded search returning the smallest number $x$ s.t.\ $f$ on $x$ returns $0$
and $f$ terminates on a non-zero value for every $y<x$.

\begin{figure}
$$\displaystyle
  \begin{array} {r@{~\toot~}l@{\qquad}r@{~\toot~}l}
    \semra{\cst n}~\vec v~x   & n = x 
 &  \semra{\zero}~\vec v~x    & 0 = x \\
    \semra{\succN}~\vec v~x   & 1+\vec v_0 = x 
 &  \semra{\prj p}~\vec v~x  & \vec v_p = x\\[1ex]
    \semra{\comp f{\vec g}}~\vec v~x 
 & \multicolumn{3}{@{}l}{\exists \vec w,\, \semra f~\vec w~x \lconj \forall p, \semra{\vec g_p}~\vec v~\vec w_p}\\
    \semra{\recA f g}~(0\ccons\vec v)~x 
  & \multicolumn{3}{@{}l}{\semra f~\vec v~x} \\
    \semra{\recA f g}~(1+n\ccons\vec v)~x 
  & \multicolumn{3}{@{}l}{\exists y, \semra{\recA f g}~(n\ccons \vec v)~y \lconj \semra g~(n\ccons y\ccons \vec v)~ x} \\
    \semra{\minA f}~\vec v~x
  & \multicolumn{3}{@{}l}{\exists \vec w:\nat^x, \semra f~(x\ccons\vec v)~0\lconj 
         \forall y:\fin x,\,\semra f~(\overline y\ccons\vec v)~(1+\vec w_y)}
\end{array}$$
\caption{\label{ra_rel_sem}Relational semantics for $\mu$-recursive algorithms.}
\end{figure}

\begin{figure}
\begin{mathpar}
  \infer{~}{ \bigstep[1 + c]{\cst n}{\vec v}{n}  }

  \infer{~}{ \bigstep[1 + c]{\zero}{\vec v}{0}  }

  \infer{~}{ \bigstep[1 + c]{\succN}{x\ccons\vec v}{1 + x}  }

  \infer{~}{ \bigstep[1 + c]{\prj p}{\vec v}{\vec v_p} }

  \infer{\forall p:\fin k,\, \bigstep[c - p]{\vec g_p}{\vec v}{\vec w_p} \and \bigstep[1 + c]{f}{\vec w}{x}}{\bigstep[2 + c]{\comp f {\vec g}}{\vec v}{x}}

  \infer{\bigstep[c]{f}{\vec v}{x} }{\bigstep[1 + c]{\recA f g}{0 \ccons \vec v}{x}}

  \infer{\bigstep[c]{\recA f g}{n \ccons \vec v}{y} \and \bigstep[c]{g}{n \ccons y \ccons \vec v}{x} }{\bigstep[1 + c]{\recA f g}{(1 + n) \ccons \vec v}{x}}

  \infer{\bigstep[c - x]{f}{x \ccons \vec v}{0} \and \forall p : \fin x,\,\bigstep[c - p]{f}{p \ccons \vec v}{1 + \vec w_p}}{\bigstep[1 + c]{\minA f}{\vec v}{x} }
\end{mathpar}
\caption{\label{ra_ca_sem}Cost aware big-step semantics for $\mu$-recursive algorithms.}
\end{figure}

Following~\cite{DLW17}, in Fig.~\ref{ra_ca_sem} we also formalise a cost aware big-step evaluation 
predicate $\bigstep[c] f v x$, where $f : \recalg k$, $v : \nat^k$, $x : \nat$, and where $c : \nat$ 
denotes the cost of a computation, the cost being tailored toward the naive step-indexed
evaluator to be defined in the next section\rlap.\footnote{In the Coq code, the predicate has one additional auxiliary argument to ease the correctness proof for the step-indexed interpreter which is only internal to the proof and thus omitted here.}\enspace
We can directly relate the relational semantics with the big-step semantics:

\setCoqFilename{L.Reductions.MuRec}
\begin{lemma}[][ra_bs_c_correct]
  $\forall\,k\,(f:\recalg k)\, (\vec v:\nat^k)\,(x:\nat), \semra f~\vec v~x \toot \exists c,\bigstep[c] f {\vec v} x$.
\end{lemma}

\setCoqFilename{MuRec.recalg}

With these characterizations of the semantics of $\mu$-recursive algorithms, we can define $\mu$-recursive halting as follows:
\[ \coqlink[MUREC_HALTING]{\M{$\mu$-rec}} (k : \nat, f : \recalg k, \vec v : \nat^k ) := \exists x,\semra f~{\vec v}~x \]
A function $f : \recalg k$ is called \textit{total} if it terminates on any input, i.e.\ $\forall \vec v:\nat^k\, \exists x,\semra f~{\vec v}~x$.

\setCoqFilename{H10.DPRM}
\begin{definition}[][mu_recursive_n]
We call a relation $k$-ary relation $R$ over natural numbers \emph{$\mu$-recursively recognisable} if
there is a $\mu$-recursive algorithm\/ $f:\recalg k$ s.t.\/
for any\/ $\vec v : \nat^k$ the equivalence\/ $R~\vec v \toot \exists x,\sem f~ {\vec v}~ x$ holds.
\end{definition}

\subsection{\mathversion{bold}Compiling \texorpdfstring{$\mu$}{mu}-recursive algorithms into Minsky machines}

We describe how to compile a $\mu$-recursive algorithm 
$f:\recalg k$ into a Minsky machine. To avoid solving complicated
inequality constraints over the bounded numbers in $\fin m$,\footnote{Unlike $\nat$ or $\Z$, the inductive type $\fin m$
representing the numbers bounded by $m$ is not equipped with powerful linear
constraints solving tactics such as \ctactic{lia} or \ctactic{omega}. Working
with bounded numbers here would require from us to manipulate embeddings
of e.g.\ the disjoint sum $\fin i+\fin j$ into $\fin k$ (for some large enough $k$),
which would also impact types that depend on $i$, $j$ and $k$. For those who
experienced it, this should instantly trigger the bad memory of \emph{setoid hell}.
Using the above mentioned tools, it is much simpler to work with intervals of
$\nat$ instead.}
 we here work with Minsky machines where registers are indexed with $\nat$
instead of $\Fin m$. Hence the number of registers is not bounded a priori but of course, for a given
Minsky machine, the number of register that actually occur in the code is bounded. And indeed,
in the end of the process in \Cref{coq:ra_mm_simulator},
we project to Minsky machines with registers in $\Fin m$ (for some $m$) after the
$\mu$-recursive algorithm has been fully compiled. 
We will not enter into the details on how particular $\mu$-recursive operators are implemented
but instead focus on the global invariant of the compiler. 

For a $\mu$-recursive algorithm
$f:\recalg k$ with $k$ inputs and one output, we produce a Minsky machine $P:\instr_\nat$ with the following
structure for registers, described by the four extra parameters $i$, $p$, $o$ and $m$, all of type $\nat$: 
\begin{description}
\item[\mathversion{bold}$i$] is the \PC-index of the first instruction of $P$;
\item[\mathversion{bold}$p$] the $k$ inputs for $f$ are to be read in the $k$ registers $\{p,\ldots,p+k-1\}$ of $P$;
\item[\mathversion{bold}$o$] the output of $f$ is to be written in the register $o$ of $P$;
\item[\mathversion{bold}$m$] $P$ can use spare registers above $m$ and assume their initial value is $0$.
\end{description} 
Moreover, we require that all registers except the output register $o$ are returned to their initial
value when the computation is terminated. In particular, spare registers return to their initial value $0$. 

\newcommand{\mmenv}{\rho}
\newcommand{\mmsubst}[3]{{#1}\{{#2}\leftarrow{#3}\}}
\newcommand{\exteq}{\mathrel{\simeq}}

\setCoqFilename{MuRec.ra_mm_env}

\begin{definition}[][ra_compiled]
We say that\/ \emph{$P$ properly compiles $f$} under the constraints\/ $i,p,o,m$ and
we write\/ \coqlink[ra_compiled]{$\cstfont{ra\_compiled}~k~(f:\recalg k)~(i~p~o~m:\nat)~(P: \instr_\nat)$} if
for any\/ $\vec v:\nat^k$ and\/ $\mmenv:\nat\to \nat$ such that 
(1)\/~$\forall u\geq m,\, \mmenv_u = 0$ and (2)\/~$\forall q:\fin k,\,\mmenv_{\overline q+p} =\vec v_q$,
we have
\begin{description}
\item[soundness] $\forall x:\nat,\,\semra f~\vec v~x
    \to \mmeval iPi\mmenv{\clength P+i}{\mmsubst\mmenv o x}$;
\item[completeness] $\mmterm{(i,P)}{(i,\mmenv)} \to \exists x,\, \sem f~\vec v~x$.
\end{description}
\end{definition}

In these $\nat$-indexed register machines, the state of a machine cannot be described
by a finite vector $\vec v:\nat^k$ but is instead represented by an environment $\mmenv:\nat\to\nat$
mapping register  indices (in $\nat$) to the values they contain (in $\nat$ also, but with a different meaning).
The two conditions~(1) and~(2) state that $\mmenv$ is null above register $m$
and $\mmenv$ contains $\vec v$ starting at register $p$. And under these conditions, $(i,P)$ should 
simulate a terminating computation of $f$ on $\vec v$ outputting $x$ (soundness), while conversely, 
whenever $(i,P)$ terminates starting from $(i,\mmenv)$, this entails that $f$ terminates
on $\vec v$ (completeness).

We can now construct a certified compiler 
from $\mu$-recursive algorithms to $\nat$-indexed Minsky machines.
Literally, \Cref{coq:ra_compiler} below states that whenever the output
register $o$ does not belong to the input registers $\{p,\ldots,p+k-1\}$ or the spare 
registers $\{m,m+1,\ldots\}$ and that spare registers have indices above input registers, 
then one can properly compile $f$ at \PC\ value $i$.

\begin{theorem}[][ra_compiler]
Given a $\mu$-recursive algorithm\/ $f:\recalg k$, we can build a term:
$$\textstyle\forall\, i\, p\, o\, m,\, o < m
                 \to \lnot(p \leq o < k+p)
                 \to k+p \leq m
                 \to \sum P,\, \cstfont{ra\_compiled}~k~f~i~p~o~m~P.$$
\end{theorem}

The proof proceeds by structural induction on $f:\recalg k$. The details are quite involved
and not exposed in here: we just wanted to state the above invariant, which is closed
under each $\mu$-recursive constructor. The sub-Minsky machines are
built and composed using the compositional techniques already presented 
in~\cite{FLW19}.
Notice that the statement of $\cstfont{ra\_compiled}$ takes the output value $x$
in $\sem f~\vec v~x$ into account while in the below \Cref{coq:ra_mm_simulator}, we only care about termination.
Because the termination of e.g.\ the $\mu$-recursive composition $\comp f{\vec g}$ depends 
not just on the termination of $\vec g_p$ (for all $p:\fin k$) but
on the actual output values $\vec w_p$ of 
$\semra{\vec g_p}~\vec v~{\vec w_p}$, 
it is necessary to be more precise in
the stated invariant used in the inductive construction of the compiler.

\setCoqFilename{MuRec.ra_simul}
\begin{theorem}[][ra_mm_simulator]
Given a $\mu$-recursive algorithm\/ $f:\recalg k$, one can compute\/ $n:\nat$
and a list of Minsky machine instructions\/ $P:\instr_{\fin{k+1+n}}$ such that
for any\/ $\vec v:\nat^k$, 
$$(\exists x,\,\sem f~\vec v~x)\toot \mmterm{(1,P)}{(1,\vec v\capp\vec 0)}.$$
\end{theorem}

\begin{proof}
By \Cref{coq:ra_compiler}, we compile $f$ into $P':\instr_\nat$ under the constraints
$i\defeq 1$, $p\defeq 0$, $o\defeq k$ and $m\defeq 1+k$. Then we compute $n$
such that $k+1+n$ is a strict upper bound of all the register indices that occur 
in $P'$. Now we can map $P':\instr_\nat$ to $P:\instr_{\fin{k+1+n}}$ while
preserving its semantics. Notice that contrary to \Cref{coq:ra_compiler}, 
the resulting output value of $f$ or $(1,P)$ can now be disregarded since 
we only care about termination.
\end{proof}

\setCoqFilename{MinskyMachines.Reductions.MUREC_MM}
\begin{corollary}[][MUREC_MM_HALTING]
  $\M{$\mu$-rec} \reducesto \M{MM}$.
\end{corollary}

\subsection{\mathversion{bold}Diophantine relations are \texorpdfstring{$\mu$}{mu}-recursively recognisable}

We encode $n$-ary  Diophantine relations using $\mu$-recursive algorithms.
Following the DPRM \Cref{coq:DPRM_n}, we consider a $n$-ary relation $R:\nat^n\to\Prop$ to be Diophantine 
if there is a single Diophantine equation\/ $p\dequal q:\dsingle(\fin m,\fin n)$ 
with\/ $n$ parameters and\/ $m$ (existential) variables such that\/ 
\[ \forall \vec v : \nat^n,\, R~\vec v \toot \exists \vec w : \nat^m,\, \semb p{\vec v}{\vec w} = \semb q{\vec v}{\vec w}.\]
Again we abusively confuse vectors with maps using the equivalence $X^n \simeq\fin n\to X$. 
We will prove that every Diophantine relation in this sense is $\mu$-recursively recognisable (see
\Cref{coq:mu_recursive_n}).
The proof idea is relatively straightforward.

\newcommand{\mueval}{\cstfont{eval}}
\newcommand{\mutest}{\cstfont{test}}

First, we implement $\mueval : \dpoly(\fin m, \fin n) \to \recalg{m + n}$ evaluating any polynomial:

\setCoqFilename{MuRec.ra_dio_poly}
\begin{lemma}[][ra_dio_poly_val]
  Given any Diophantine polynomial\/ $p : \dpoly(\fin m,\fin n)$, one can compute a\/ $\mu$-recursive algorithm\/ $\mueval_p : \recalg {m + n}$ 
  s.t.\ $\semra{\mueval_p}~(\vec w \app \vec v)~\semb p {\vec v}{\vec w}$.
\end{lemma}

\begin{proof}
  The implementation relies on implementations of addition and multiplication on natural numbers, which are relatively straightforward for $\mu$-recursive 
  algorithms.
\end{proof}

Secondly, we implement a bijection between $\nat$ and $\nat^m$ for any $m$:

\newcommand{\mupr}{\cstfont{pr}}
\newcommand{\muinj}{\cstfont{inj}}
\newcommand{\muproject}{\cstfont{project}}
\newcommand{\mufind}{\cstfont{find}}

\setCoqFilename{MuRec.ra_recomp}
\begin{lemma}[][ra_project_val]
  There are functions\/ $\mupr_m : \nat \to \nat^m$ and\/ $\muinj_m : \nat^m \to \nat$ s.t.\ $\mupr_m(\muinj_m\, \vec v) = \vec v$.
  Furthermore, given\/ $i : \fin m$, one can compute a\/ $\mu$-recursive algorithm $\muproject_i : \recalg 1$ s.t.\ 
  $\semra{\muproject_i}~(x\ccons\cnil)~{(\mupr_m\,x)_i}$, i.e.\ all the\/ $m$ components of\/ $\mupr_m$ are\/ $\mu$-recursive.
\end{lemma}

Thirdly, given an equation $p \dequal q$, we compute a test algorithm $\mutest_{p,q}:\recalg {1+n}$, always terminating and 
returning $0$ iff $p \dequal q$ is satisfied when decoding the first argument
as a vector $\vec w$ using $\mupr_m$:

\setCoqFilename{MuRec.ra_dio_poly}
\begin{lemma}[][ra_dio_poly_test_val]
  Given any $p \dequal q : \dsingle(\fin m, \fin n)$, one can compute a\/ $\mu$-recursive algorithm 
  $\mutest_{p,q} : \recalg {1+n}$ s.t.\/
  $\semra{\mutest_{p,q}}~(x \ccons \vec v)~0 \toot \semb p {\vec v}{\vec w} = \semb q {\vec v}{\vec w}$
  where $\vec w \defeq \mupr_m(x)$.
  Furthermore,\/ $\mutest_{p,q}$ is 
  \coqlink[ra_dio_poly_test_prim]{primitive recursive} hence
  \coqlink[ra_dio_poly_test_total]{total}.
\end{lemma}

To finish, given a vector $\vec v : \nat^n$ describing the values of parameters, 
we use minimisation $\minA$ to search for a number $x$ such that $\vec w\defeq\mupr_m(x)$ 
is a solution for $p \dequal q$:

\begin{theorem}[][ra_dio_poly_find_spec]
  Given\/ $p \dequal q : \dsingle(\fin m, \fin n)$, one can compute a\/ $\mu$-recursive algorithm $\mufind_{p,q} : \recalg{n}$ 
  s.t.\ for any\/ $\vec v : \nat^n$, we have the equivalence 
  $$(\exists  x:\nat,\,\semra{\mufind_{p,q}}~{\vec v}~{x})~ \toot~ \exists \vec w : \nat^m,\, \semb p{\vec v}{\vec w} = \semb q{\vec v}{\vec w}.$$
\end{theorem}

Notice that while $\mufind_{p,q}$ cannot be primitive recursive (it does not terminate when the equation $p\dequal q$ is unsolvable), 
it is however implemented as a single unbounded minimization applied to an otherwise primitive recursive algorithm.

\setCoqFilename{MuRec.Reductions.H10_to_MUREC_HALTING}
\begin{corollary}[][H10_MUREC_HALTING]
  $\M{H10}\reducesto \M{$\mu$-rec}$.
\end{corollary}

\begin{proof}
Direct application of \Cref{coq:ra_dio_poly_find_spec} with no parameters, i.e.\ $n\defeq 0$.
\end{proof}

\section{The Weak Call-by-Value \texorpdfstring{$\lambda$}{lambda}-Calculus \texorpdfstring{$\M{L}$}{L}}
\label{sec:wcbv}

Until now, we have shown that the problems $\M{MM}$, $\M{FRACTRAN}$, $\M{H10}$, and $\M{$\mu$-rec}$ are all interreducible w.r.t.\ many-one reductions.
We contribute our proofs to the Coq library of undecidability proofs~\cite{PSLSyntCT}, which contains several other well-known reduction proofs.
Amongst them is a chain of reductions published in related work and discussed in Section~\ref{sec:rel_work}, 
establishing that the halting problem of Turing machines reduces to $\M{MM}$ and thus to all other problems we consider.

It is also possible to prove that all considered problems are in fact interreducible to the halting problem of Turing machines.
We demonstrate one technique to do so, based on the weak call-by-value $\lambda$-calculus $\M{L}$,
which is already shown interreducible with the halting problem of Turing machines in~\cite{forster2021simulation}.

We briefly introduce $\M L$ in this section and then reduce $\M{$\mu$-rec}$ to halting in $\M{L}$.
Syntactically, $\M L$ is just the untyped $\lambda$-calculus, De Bruijn style:
\[ s, t, u : \M{L} \bnfdef n ~|~ s t ~|~ \lambda s \qquad \text{where $n : \nat$} \]
We define a weak call-by-value evaluation predicate $s \triangleright t$, coinciding with other definitions of reduction for closed terms,
where $s[t/0]$ means ``replace De Bruijn variable 0 in $s$ by $t$''.
\begin{mathpar}
  \infer{~}{\lambda s \triangleright \lambda s}

  \infer{s \triangleright \lambda s' \and t \triangleright t' \and s'[t'\!/0] \triangleright u}{st \triangleright u}
\end{mathpar}
The halting problem for $\M L$ can then be defined as $\M{WCBV} (s : \M L) := \exists t,\, s \triangleright t$.

It is possible to encode natural numbers, vectors over natural numbers, and other data types into $\M L$ using Scott's encoding.
We will denote with $\overline{\vphantom{t} \,\cdot\,}$ such encoding functions.
For further details on encodings and $\M{L}$ we refer to~\cite{FS}, since they do not actually matter to understand the reduction.

In general, a predicate reduces to the $\M{L}$-halting problem $\M{WCBV}$ if it is $\M L$-recognisable:

\setCoqFilename{L.Computability.Synthetic}

\begin{definition}[][L_recognisable]
  We say that a function $f : {X_1 \to \dots \to X_n \to Y}$ is \emph{$\M L$-computable} if there exists a closed term $t$ s.t.\ $\forall x_1 : X_1.\dots\forall x_n : X_n.\;t \; \overline {x_1} \dots \overline {x_n} \triangleright \overline{f x_1 \dots x_n}$.

  Given a type $X$ encodable in $\M L$ and a predicate $P$ over $X$ we say that $P$ is \emph{$\M{L}$-recognisable} if there is an $\M{L}$-computable function $f : X \to \nat \to \bool$ s.t. $P\,x \toot \exists n, f\,x\,n=\mathsf{true}$.
\end{definition}

\begin{theorem}[][L_recognisable_halt]\label{thm:Lrec}
  Let $X$ be encodable in $\M L$.
  For a predicate $P$ over $X$ we have that
  $P \preceq \M{WCBV}$ if and only if $P$ is $\M{L}$-recognisable.
\end{theorem}

To instantiate the theorem to \M{$\mu$-rec} one would like to give a step-indexed evaluation function for $\mu$-recursive algorithms, i.e.\ for $f : \recalg k$ a function $\sem{f}{}{} : \nat \to \nat ^k  \to \option \nat$ s.t.\ $\bigstep[c] f{\vec v} x \toot \sem f {}{} c~\vec v = \some x$ and then show that it is computable in $\M L$.
The second step, i.e.~proving that this function is $\M L$-computable, is in principle fully automatic using the certifying extraction framework from Forster and Kunze~\cite{forster_et_al:LIPIcs:2019:11072}, which is implemented using tools from the MetaCoq projet~\cite{metacoq}.
The framework supports extracting non-dependent, non-mutual, non-nested, simply-typed Coq functions to $\M L$ automatically, and also generates a proof of correctness of the extract.
However, recall that the type $\recalg k$ of $\mu$-recursive functions is both heavily dependent and used nested constructions:
It has a type index in $k$, makes use of finite types $\Fin {k}$, and has a nested use of the dependent vector type, i.e.\ contains subterms of type $(\recalg k)^i$.
Thus, a direct step-indexed interpreter will not be extractable.
Instead, we use a general technique and implement a step-indexed interpreter working on the \textit{syntactic skeleton} of $\mu$-recursive functions.

A syntactic skeleton for a type $I$ mirrors the constructors of $I$, but without any dependent types.
If $I$ furthermore has nested applications of types $N_1, \dots, N_n$, the syntactic skeleton also contains the constructors of $N_1, \dots, N_n$.
That is, if $I$ has $i$ constructors and $N_1, \dots, N_n$ have $m_1, \dots, m_n$ constructors respectively,
the syntactic skeleton of $I$ has $i + m_1 + \dots + m_n$ constructors.
The syntactic skeleton $\recalg{}'$ of the type $\recalg{}$ is defined as follows:
\begin{mathpar}
  \def\myspace{\mskip 50mu}%
  \begin{array}{c@{\mskip 40mu}c}
    \infer{ n : \nat }{ \cst n : \recalg{}' }\myspace
    \infer{~}{\zero : \recalg{}'}\myspace
    \infer{~}{\succN : \recalg{}'}\myspace
    \infer{j : \nat}{\prj j : \recalg{}'}
&   \infer{~}{\cstfont{nil} : \recalg{}'}\\[3ex]
    \infer{f : \recalg{}' \and g : \recalg{}'}{\comp f g : \recalg{}'}\myspace
    \infer{f : \recalg{}' \and g : \recalg {}'}{\recA f g : \recalg{}'}\myspace
    \infer{f : \recalg{}'}{\minA f : \recalg{}'}
&   \infer{f : \recalg{}' \and g : \recalg{}'}{\cstfont{cons}\,f\,g : \recalg{}'}\\[2ex]
  \end{array}
\end{mathpar}

Note that we re-use constructor names.
In the syntactic skeleton we have $j : \nat$ instead of $j : \Fin {k}$ for $\prj j$,
$g : \recalg{}'$ instead of $G :{(\recalg k)^i}$ for $\comp fg$,
and the constructors $\cstfont{cons}$ and $\cstfont{nil}$ as the constructors of the vector type are added.

\newcommand{\erase}{\cstfont{erase}}
It is straightforward to implement mutually recursive functions $\erase:{\recalg k \to \recalg{}'}$ and $\erase' : {(\recalg k)^i \to \recalg{}'}$ which are essentially the identity.
We call a skeleton $s : \recalg{}'$ \emph{valid} if it corresponds to a $\mu$-recursive function or a vector, i.e.\ if $\exists k : \nat,\;(\exists f :{\recalg k},\;s = \erase\,f) \lor (\exists i : \nat,\exists l :{(\recalg k)^i},\;s = \erase'\,l)$.

For $\recalg{}'$ we can now define a step-indexed evaluation function \[ \semb{\cdot}{}{}{} :{\recalg{}'\to\nat\to\nat\to\List\nat\to\option(\nat+\List\nat)}. \]
A call $\semb{f}{c}{m}~l$ uses $c$ as step-index, $m$ as auxiliary counter to implement unbounded search, and $l$ as input.
If $\semb{f}{c}{m}~l = \some {\inl v}$, then $v$ is the value of the evaluation.
If $\semb{f}{c}{m}~l = \some {\inr l'}$, then $f$ encoded a list of functions via $\cstfont{cons}$ and $\cstfont{nil}$ which pointwise evaluated to $l' :{\List\nat}$.
If $\semb{f}{c}{m}~l = \emptyset$ either if usual the step-index $c$ was not big enough, the function $f$ does not terminate on input $l$, or $f$ is not a valid skeleton.

\bigskip

\centerline{\small$\begin{array}{r@{\,\defeq\,}l@{~~}l}
  \semb{\cst n}{1 + c}{m}~l              & \some{\inl n} \\
  \semb{\zero}{1 + c}{m}~l               & \some{\inl 0} \\
  \semb{\succN}{1 + c}{m}~(x\ccons l)    & \some{\inl (1 + x)} \\
  \semb{\prj p}{1 + c}{m}~l              & \some{\inl x} 
                                         & \text{for $l_p = \some x$} \\
  \semb{\comp f g}{1 + c}{m}~l           & \some{\inl x} 
                                         & \text{for $\semb{g}{c}{m}~l=\some{\inr l'}$ and $\semb{f}{c}{m}~l' = \some{\inl x}$} \\
  \semb{\recA f g}{1 + c}{m}~(0\ccons l) & \some{\inl x} 
                                         & \text{for $\semb{f}{c}{m}~l = \some{\inl x}$} \\
  \semb{\recA f g}{1 + c}{m}~\bigl((1 + n)\ccons l\bigr) 
                                         & \some{\inl x} 
                                         & \text{for $\semb{\recA f g}{c}{m}~(n\ccons l) = \some{\inl y}$ and $\semb{g}{c}{m}~(n\ccons y\ccons l) = \some{\inl x}$} \\
  \semb{\cstfont{nil}}{1 + c}{m}~l       & \some{\inr \cnil} \\
  \semb{\cstfont{cons}~f~g}{1 + c}{m}~l  & \some{\inr(x \ccons l')} 
                                         & \text{for $\semb{f}{c}{m}~l=\some{\inl x}$ and $\semb{g}{c}{m}~l=\some{\inr l'}$} \\
  \semb{\minA f}{1 + c}{m}~l             & \some{\inl m} 
                                         & \text{for $\semb{f}{c}0~(m \ccons l) = \some{\inl 0}$} \\
  \semb{\minA f}{1 + c}{m}~l             & \some{\inl x} 
                                         & \text{for $\semb{f}{c}0~(m \ccons l) = \some{\inl (1 + y)}$ and $\semb{\minA f}{c}{1 + m}~l = \some{\inl x}$} \\[0.5ex] 
  \semb{f}{c}{m}~l                       & \emptyset
                                         & \text{in all other cases}
\end{array}$}

\medskip

\setCoqFilename{L.Reductions.MuRec}
\begin{lemma}[][erase_correct]
  For all $f :{\recalg k}$, $\vec v :{\nat^k}$, $c : \nat$, and $n : \nat$ we have
  \[\bigstep[c]{f}{\vec v}{n} \leftrightarrow \semb{f}{c}{0} (\cstfont{v2l}~{\vec v}) = \some{\inl n}\]
\end{lemma}

Note that on the right hand side of this equivalence, the vector $\vec v$ is just mapped to its underlying list $\cstfont{v2l}~{\vec v}$.

\begin{theorem}[][MUREC_WCBV]
  $\M{$\mu$-rec} \preceq \M{WCBV}$.
\end{theorem}
\begin{proof}
  We use \Cref{thm:Lrec} and have to provide a function $F$ s.t.\ $\M{$\mu$-rec} (k : \nat, f : \recalg k, \vec v : \nat^k) \leftrightarrow \exists c,\,F\,c\,(k,f,\vec v) = \mathsf{true}$.
  
  We define $F\,(k : \nat, f : \recalg k, \vec v : \nat^k)\,c \defeq \mathsf{true}$ if $\semb{f}{c}{0} (\cstfont{v2l}~{\vec v}) = \some {\inl n}$ for some $n$ and $F\, (k,f,v)\,c \defeq \mathsf{false}$ otherwise.
  
  Now $F$ is part of the Coq fragment supported by the extraction framework in~\cite{forster_et_al:LIPIcs:2019:11072},
  since the step-indexed evaluator $\sem{\cdot}{}{}$ does not use any dependent, nested, mutual, or non-simple types.
  Thus $F$ is $\M L$-computable.

  Furthermore we have as wanted that
  \begin{align*}
    \M{$\mu$-rec}(k,f,\vec v) ~~&\toot~~ \exists n,\,\semra{f}~{\vec v}~{n} \\
                              &\toot~~ \exists n\,c,\,\bigstep[c]{f}{\vec v}{n} \\
                              &\toot~~ \exists n\,c,\,\semb{f}{c}{0}\,(\cstfont{v2l}~{\vec v}) = \some{\inl n} \\
                              &\toot~~ \exists c,\,F\,(k, f, \vec v)\,c = \mathsf{true} \qedhere
  \end{align*}
\end{proof}

\section{Related Work}

\label{sec:rel_work}

Regarding formalisations of Hilbert's tenth problem, there are various unfinished and preliminary results in different proof assistants: Carneiro~\cite{carneiro_lean_2018} formalises Matiyasevich's theorem (Diophantineness of exponentiation) in Lean, 
but does not consider computational models or the DPRM theorem.
P{a}k formalises results regarding Pell's equation~\cite{pak2017matiyasevich} and proves that Diophantine sets are closed under union and intersection~\cite{pkak2018diophantine}, both as parts of the Mizar Mathematical Library. 
Stock et al.~\cite{stock_hilbert_2018,bayer_et_al:LIPIcs:2019:11088} report on an unfinished formalisation of the DPRM theorem in Isabelle based on~\cite{Matiyasevich2000}.
They cover some parts of the proof, but acknowledge for important missing
results like Lucas's or ``Kummer's theorem'' and a ``formalisation of a register machine\rlap.''\enspace Moreover, none of the cited reports considers the computability of the reductions involved or the verification of a universal machine in the chosen model of computation yet, one of them being a necessary proof goal for an actual undecidability result in the classical meta-theories of Isabelle/HOL and Mizar.

Regarding undecidability proofs in type theory, Forster, Heiter, and Smolka~\cite{FHS18} reduce the halting problem of Turing machines to \M{PCP}.
Forster and Larchey-Wendling~\cite{FLW19} reduce \M{PCP} to provability in linear logic via the halting problem of Minsky machines, which we build on.
Forster, Kirst and Smolka develop the notion of synthetic undecidability in Coq and prove the undecidability of various notions in first-order logic~\cite{0002KS19}.
Spies and Forster mechanise the undecidability proof of second-order unification by reduction from $\M{H10}$~\cite{spies2020undecidability} originally shown by Goldfarb~\cite{goldfarb1981undecidability}.
Forster, Kunze, and Wuttke reduce the halting problem of multi-tape Turing machines to single-tape Turing machines~\cite{forster2020verified}.
Dudenhefner and Rehof~\cite{dudenhefner2018simpler} mechanise a recently simplified undecidability proof for System F inhabitation.

Regarding formalisations of $\mu$-recursive functions, Larchey-Wendling~\cite{DLW17} shows that every total $\mu$-recursive function can directly be computed in Coq and Carneiro~\cite{carneiro:LIPIcs:2019:11067} mechanises standard computability theory based on $\mu$-recursive functions.
Xu, Zhang, and Urban~\cite{xu2013} mechanise $\mu$-recursive functions and Turing machines in Isabelle and prove their computational equivalence.
Their proof uses Abacus machines as intermediate layer, which are similar to our Minsky machines.

\section*{Acknowledgements}

We would like to thank Gert Smolka, Dominik Kirst, and Simon Spies for helpful discussion regarding the presentation.
The first author was partially supported by the \href{https://ticamore.logic.at/}{TICAMORE} project (\href{http://www.agence-nationale-recherche.fr/?Projet=ANR-16-CE91-0002}{ANR grant 16-CE91-0002}).

\bibliographystyle{alphaurl}
\bibliography{biblio}

\newcommand{\etalchar}[1]{$^{#1}$}
\begin{thebibliography}{FLWD{\etalchar{+}}21}

\bibitem[BDP{\etalchar{+}}19]{bayer_et_al:LIPIcs:2019:11088}
Jonas Bayer, Marco David, Abhik Pal, Benedikt Stock, and Dierk Schleicher.
\newblock {The DPRM Theorem in Isabelle (Short Paper)}.
\newblock In {\em 10th International Conference on Interactive Theorem Proving
  (ITP 2019)}, volume 141 of {\em Leibniz International Proceedings in
  Informatics (LIPIcs)}, pages 33:1--33:7, Dagstuhl, Germany, 2019. Schloss
  Dagstuhl--Leibniz-Zentrum fuer Informatik.
\newblock \href {https://doi.org/10.4230/LIPIcs.ITP.2019.33}
  {\path{doi:10.4230/LIPIcs.ITP.2019.33}}.

\bibitem[Car18]{carneiro_lean_2018}
Mario Carneiro.
\newblock {A Lean formalization of Matiyasevi\v c's theorem}, 2018.
\newblock \href {http://arxiv.org/abs/1802.01795} {\path{arXiv:1802.01795}}.

\bibitem[Car19]{carneiro:LIPIcs:2019:11067}
Mario Carneiro.
\newblock {Formalizing Computability Theory via Partial Recursive Functions}.
\newblock In {\em 10th International Conference on Interactive Theorem Proving
  (ITP 2019)}, volume 141 of {\em Leibniz International Proceedings in
  Informatics (LIPIcs)}, pages 12:1--12:17, Dagstuhl, Germany, 2019. Schloss
  Dagstuhl--Leibniz-Zentrum fuer Informatik.
\newblock \href {https://doi.org/10.4230/LIPIcs.ITP.2019.12}
  {\path{doi:10.4230/LIPIcs.ITP.2019.12}}.

\bibitem[Con87]{Conway1987}
John~H. Conway.
\newblock {\em {FRACTRAN: A Simple Universal Programming Language for
  Arithmetic}}, pages 4--26.
\newblock Springer New York, New York, NY, 1987.
\newblock \href {https://doi.org/10.1007/978-1-4612-4808-8_2}
  {\path{doi:10.1007/978-1-4612-4808-8_2}}.

\bibitem[Dav53]{davis1953arithmetical}
Martin Davis.
\newblock Arithmetical problems and recursively enumerable predicates 1.
\newblock {\em The Journal of Symbolic Logic}, 18(1):33--41, 1953.
\newblock \href {https://doi.org/10.2307/2266325} {\path{doi:10.2307/2266325}}.

\bibitem[Dav73]{davis73}
Martin Davis.
\newblock {Hilbert's Tenth Problem is Unsolvable}.
\newblock {\em The American Mathematical Monthly}, 80(3):233--269, 1973.
\newblock URL: \url{10.1080/00029890.1973.11993265}.

\bibitem[DP59]{davis1959computational}
Martin Davis and Hilary Putnam.
\newblock {\em {A computational proof procedure; Axioms for number theory;
  Research on Hilbert's Tenth Problem}}.
\newblock Air Force Office of Scientific Research, Air Research and
  Development, 1959.

\bibitem[DPR61]{davis1961decision}
Martin Davis, Hilary Putnam, and Julia Robinson.
\newblock {The decision problem for exponential Diophantine equations}.
\newblock {\em Annals of Mathematics}, pages 425--436, 1961.
\newblock \href {https://doi.org/10.2307/1970289} {\path{doi:10.2307/1970289}}.

\bibitem[DR19]{dudenhefner2018simpler}
Andrej Dudenhefner and Jakob Rehof.
\newblock {A Simpler Undecidability Proof for System F Inhabitation}.
\newblock In Peter Dybjer, Jos{\'e}~Esp{\'\i}rito Santo, and Lu{\'\i}s Pinto,
  editors, {\em 24th International Conference on Types for Proofs and Programs
  (TYPES 2018)}, volume 130 of {\em Leibniz International Proceedings in
  Informatics (LIPIcs)}, pages 2:1--2:11, Dagstuhl, Germany, 2019. Schloss
  Dagstuhl--Leibniz-Zentrum fuer Informatik.
\newblock \href {https://doi.org/10.4230/LIPIcs.TYPES.2018.2}
  {\path{doi:10.4230/LIPIcs.TYPES.2018.2}}.

\bibitem[FHS18]{FHS18}
Yannick Forster, Edith Heiter, and Gert Smolka.
\newblock {Verification of PCP-Related Computational Reductions in Coq}.
\newblock In {\em ITP 2018}, pages 253--269. Springer, 2018.
\newblock \href {https://doi.org/10.1007/978-3-319-94821-8_15}
  {\path{doi:10.1007/978-3-319-94821-8_15}}.

\bibitem[FK19]{forster_et_al:LIPIcs:2019:11072}
Yannick Forster and Fabian Kunze.
\newblock {A Certifying Extraction with Time Bounds from Coq to Call-By-Value
  Lambda Calculus}.
\newblock In {\em 10th International Conference on Interactive Theorem Proving
  (ITP 2019)}, volume 141 of {\em Leibniz International Proceedings in
  Informatics (LIPIcs)}, pages 17:1--17:19, Dagstuhl, Germany, 2019. Schloss
  Dagstuhl--Leibniz-Zentrum fuer Informatik.
\newblock \href {https://doi.org/10.4230/LIPIcs.ITP.2019.17}
  {\path{doi:10.4230/LIPIcs.ITP.2019.17}}.

\bibitem[FKS19]{0002KS19}
Yannick Forster, Dominik Kirst, and Gert Smolka.
\newblock On synthetic undecidability in coq, with an application to the
  entscheidungsproblem.
\newblock In {\em Proceedings of the 8th ACM SIGPLAN International Conference
  on Certified Programs and Proofs}, CPP 2019, page 38–51, New York, NY, USA,
  2019. Association for Computing Machinery.
\newblock \href {https://doi.org/10.1145/3293880.3294091}
  {\path{doi:10.1145/3293880.3294091}}.

\bibitem[FKSW21]{forster2021simulation}
Yannick Forster, Fabian Kunze, Gert Smolka, and Maximilian Wuttke.
\newblock {A Mechanised Proof of the Time Invariance Thesis for the Weak
  Call-By-Value $\lambda$-Calculus}.
\newblock 193:19:1--19:20, 2021.
\newblock \href {https://doi.org/10.4230/LIPIcs.ITP.2021.19}
  {\path{doi:10.4230/LIPIcs.ITP.2021.19}}.

\bibitem[FKW20]{forster2020verified}
Yannick Forster, Fabian Kunze, and Maximilian Wuttke.
\newblock {Verified Programming of Turing Machines in Coq}.
\newblock In {\em Proceedings of the 9th ACM SIGPLAN International Conference
  on Certified Programs and Proofs. ACM}, 2020.
\newblock \href {https://doi.org/10.1145/3372885.3373816}
  {\path{doi:10.1145/3372885.3373816}}.

\bibitem[FL19]{FLW19}
Yannick Forster and Dominique Larchey{-}Wendling.
\newblock {Certified Undecidability of Intuitionistic Linear Logic via Binary
  Stack Machines and Minsky Machines}.
\newblock In {\em CPP 2019}, pages 104--117. {ACM}, 2019.
\newblock \href {https://doi.org/10.1145/3293880.3294096}
  {\path{doi:10.1145/3293880.3294096}}.

\bibitem[FLW18]{forstertowards}
Yannick Forster and Dominique Larchey-Wendling.
\newblock {Towards a library of formalised undecidable problems in Coq: The
  undecidability of intuitionistic linear logic}.
\newblock {\em Workshop on Syntax and Semantics of Low-level Languages,
  Oxford}, 2018.

\bibitem[FLWD{\etalchar{+}}20]{PSLSyntCT}
Yannick Forster, Dominique Larchey-Wendling, Andrej Dudenhefner, Edith Heiter,
  Dominik Kirst, Fabian Kunze, Gert Smolka, Simon Spies, Dominik Wehr, and
  Maximilian Wuttke.
\newblock A {Coq} library of undecidable problems.
\newblock In {\em The Sixth International Workshop on Coq for Programming
  Languages (CoqPL 2020)}, 2020.
\newblock URL: \url{https://github.com/uds-psl/coq-library-undecidability}.

\bibitem[FLWD{\etalchar{+}}21]{zenodolibrary}
Yannick Forster, Dominique Larchey-Wendling, Andrej Dudenhefner, Lennard
  Gäher, Edith Heiter, Marc Hermes, Dominik Kirst, Fabian Kunze, Maximilian
  Wuttke, Gert Smolka, Simon Spies, and Dominik Wehr.
\newblock {The Coq Library of Undecidability Proofs: Hilbert's Tenth problem in
  Coq (LMCS) v1.1}, May 2021.
\newblock \href {https://doi.org/10.5281/zenodo.4827453}
  {\path{doi:10.5281/zenodo.4827453}}.

\bibitem[FS17]{FS}
Yannick Forster and Gert Smolka.
\newblock Weak call-by-value lambda calculus as a model of computation in coq.
\newblock In {\em Interactive Theorem Proving}, pages 189--206, Cham, 2017.
  Springer International Publishing.
\newblock \href {https://doi.org/10.1007/978-3-319-66107-0_13}
  {\path{doi:10.1007/978-3-319-66107-0_13}}.

\bibitem[G{\"o}d31]{godel1931formal}
Kurt G{\"o}del.
\newblock {\"U}ber formal unentscheidbare {S\"a}tze der {Principia Mathematica
  und verwandter Systeme I}.
\newblock {\em Monatshefte f{\"u}r mathematik und physik}, 38(1):173--198,
  1931.
\newblock \href {https://doi.org/10.1063/1.3051400}
  {\path{doi:10.1063/1.3051400}}.

\bibitem[Gol81]{goldfarb1981undecidability}
Warren~D. Goldfarb.
\newblock The undecidability of the secondorder unification problem.
\newblock {\em Theoretical Computer Science}, 13:225--230, 1981.
\newblock \href {https://doi.org/10.1016/0304-3975(81)90040-2}
  {\path{doi:10.1016/0304-3975(81)90040-2}}.

\bibitem[Hil02]{hilbert1902mathematical}
David Hilbert.
\newblock Mathematical problems.
\newblock {\em Bulletin of the American Mathematical Society}, 8(10):437--479,
  1902.
\newblock \href {https://doi.org/10.1090/S0002-9904-1902-00923-3}
  {\path{doi:10.1090/S0002-9904-1902-00923-3}}.

\bibitem[JM84]{JonesM84}
J.~P. Jones and Y.~V. Matijasevi{\v c}.
\newblock {Register Machine Proof of the Theorem on Exponential Diophantine
  Representation of Enumerable Sets}.
\newblock {\em J. Symb. Log.}, 49(3):818--829, 1984.
\newblock \href {https://doi.org/10.2307/2274135} {\path{doi:10.2307/2274135}}.

\bibitem[Lar17]{DLW17}
Dominique Larchey{-}Wendling.
\newblock {Typing Total Recursive Functions in Coq}.
\newblock In {\em ITP 2017}, pages 371--388. Springer, 2017.
\newblock \href {https://doi.org/10.1007/978-3-319-66107-0_24}
  {\path{doi:10.1007/978-3-319-66107-0_24}}.

\bibitem[Luc78]{lucas1878}
Edouard Lucas.
\newblock {Th\'eorie des Fonctions Num\'eriques Simplement P\'eriodiques.
  [Continued]}.
\newblock {\em American Journal of Mathematics}, 1(3):197--240, 1878.
\newblock \href {https://doi.org/10.2307/2369308} {\path{doi:10.2307/2369308}}.

\bibitem[LWF19]{fscdversion}
Dominique Larchey-Wendling and Yannick Forster.
\newblock {Hilbert's Tenth Problem in Coq}.
\newblock In {\em 4th International Conference on Formal Structures for
  Computation and Deduction (FSCD 2019)}, volume 131 of {\em Leibniz
  International Proceedings in Informatics (LIPIcs)}, pages 27:1--27:20,
  Dagstuhl, Germany, 2019. Schloss Dagstuhl--Leibniz-Zentrum fuer Informatik.
\newblock \href {https://doi.org/10.4230/LIPIcs.FSCD.2019.27}
  {\path{doi:10.4230/LIPIcs.FSCD.2019.27}}.

\bibitem[Mat70]{matijasevic1970enumerable}
Yuri~V. Matijasevi{\v c}.
\newblock {Enumerable sets are Diophantine}.
\newblock In {\em Soviet Mathematics: Doklady}, volume~11, pages 354--357,
  1970.

\bibitem[Mat97]{Matiyasevich1997}
Yuri~V. Matiyasevich.
\newblock {A new technique for obtaining Diophantine representations via
  elimination of bounded universal quantifiers}.
\newblock {\em J. Math. Sci.}, 87(1):3228--3233, 1997.

\bibitem[Mat00]{Matiyasevich2000}
Yuri~V. Matiyasevich.
\newblock {On Hilbert's Tenth Problem}.
\newblock Expository Lectures~1, Pacific Institute for the Mathematical
  Sciences, University of Calgary, February 2000.
\newblock URL:
  \url{http://www.mathtube.org/sites/default/files/lecture-notes/Matiyasevich.pdf}.

\bibitem[Mat16]{matiyasevich2016martin}
Yuri~V. Matiyasevich.
\newblock Martin davis and hilbert's tenth problem.
\newblock In {\em Martin Davis on Computability, Computational Logic, and
  Mathematical Foundations}, volume~10 of {\em Outstanding Contributions to
  Logic}, pages 35--54. Springer, 2016.
\newblock \href {https://doi.org/10.1007/978-3-319-41842-1\_2}
  {\path{doi:10.1007/978-3-319-41842-1\_2}}.

\bibitem[Min67]{Minsky}
Marvin~L. Minsky.
\newblock {\em Computation: finite and infinite machines}.
\newblock Prentice-Hall, Inc., 1967.

\bibitem[P{{a}}k17]{pak2017matiyasevich}
Karol P{{a}}k.
\newblock {The Matiyasevich Theorem. Preliminaries}.
\newblock {\em Formalized Mathematics}, 25(4):315--322, 2017.
\newblock \href {https://doi.org/10.1515/forma-2017-0029}
  {\path{doi:10.1515/forma-2017-0029}}.

\bibitem[P{a}k18]{pkak2018diophantine}
Karol P{a}k.
\newblock {Diophantine sets. Preliminaries}.
\newblock {\em Formalized Mathematics}, 26(1):81--90, 2018.
\newblock \href {https://doi.org/10.2478/forma-2018-0007}
  {\path{doi:10.2478/forma-2018-0007}}.

\bibitem[Pos44]{post1944recursively}
Emil~L. Post.
\newblock Recursively enumerable sets of positive integers and their decision
  problems.
\newblock {\em bulletin of the American Mathematical Society}, 50(5):284--316,
  1944.
\newblock \href {https://doi.org/10.1090/S0002-9904-1944-08111-1}
  {\path{doi:10.1090/S0002-9904-1944-08111-1}}.

\bibitem[Rob52]{robinson1952}
Julia Robinson.
\newblock Existential definability in arithmetic.
\newblock {\em Transactions of the American Mathematical Society},
  72(3):437--449, 1952.
\newblock \href {https://doi.org/10.1090/S0002-9947-1952-0048374-2}
  {\path{doi:10.1090/S0002-9947-1952-0048374-2}}.

\bibitem[S{\etalchar{+}}18]{stock_hilbert_2018}
Benedikt Stock et~al.
\newblock Hilbert meets {I}sabelle: Formalisation of the {DPRM} theorem in
  {I}sabelle.
\newblock {\em Isabelle Workshop 2018}, 2018.
\newblock \href {https://doi.org/10.29007/3q4s} {\path{doi:10.29007/3q4s}}.

\bibitem[SAB{\etalchar{+}}20]{metacoq}
Matthieu Sozeau, Abhishek Anand, Simon Boulier, Cyril Cohen, Yannick Forster,
  Fabian Kunze, Gregory Malecha, Nicolas Tabareau, and Th{\'{e}}o Winterhalter.
\newblock The {MetaCoq} project.
\newblock {\em J. Autom. Reason.}, 64(5):947--999, 2020.
\newblock \href {https://doi.org/10.1007/s10817-019-09540-0}
  {\path{doi:10.1007/s10817-019-09540-0}}.

\bibitem[SF20]{spies2020undecidability}
Simon Spies and Yannick Forster.
\newblock {Undecidability of higher-order unification formalised in Coq}.
\newblock In {\em Proceedings of the 9th ACM SIGPLAN International Conference
  on Certified Programs and Proofs}, pages 143--157, 2020.
\newblock \href {https://doi.org/10.1145/3372885.3373832}
  {\path{doi:10.1145/3372885.3373832}}.

\bibitem[SM19]{equations}
Matthieu Sozeau and Cyprien Mangin.
\newblock Equations reloaded: High-level dependently-typed functional
  programming and proving in coq.
\newblock {\em Proc. ACM Program. Lang.}, 3(ICFP), July 2019.
\newblock \href {https://doi.org/10.1145/3341690} {\path{doi:10.1145/3341690}}.

\bibitem[tea]{mathcomp}
The Mathematical~Components team.
\newblock {Mathematical Components}.
\newblock URL: \url{https://math-comp.github.io}.

\bibitem[Tea21]{Coq}
The Coq~Development Team.
\newblock {The Coq Proof Assistant}, January 2021.
\newblock \href {https://doi.org/10.5281/zenodo.1003420}
  {\path{doi:10.5281/zenodo.1003420}}.

\bibitem[XZU13]{xu2013}
Jian Xu, Xingyuan Zhang, and Christian Urban.
\newblock Mechanising turing machines and computability theory in isabelle/hol.
\newblock In {\em Interactive Theorem Proving}, pages 147--162, Berlin,
  Heidelberg, 2013. Springer Berlin Heidelberg.
\newblock \href {https://doi.org/10.1007/978-3-642-39634-2_13}
  {\path{doi:10.1007/978-3-642-39634-2_13}}.

\end{thebibliography}

\appendix

\section{Some Remarks about the Coq Code Contents}

\label{code_comments_appendix} 

\newcommand{\undecroot}[1]{https://github.com/uds-psl/coq-library-undecidability/#1/H10-LMCS-v1.1}
\newcommand{\undecref}[3]{{\color{ACMDarkBlue}\expandafter\href\expandafter{\undecroot{#1}/theories/#2}{#3}}}
\newcommand{\dirref}[1]{\undecref{tree}{#1}{\filefont{#1}}}
\newcommand{\fileref}[2]{\undecref{blob}{#1/#2}{\filefont{#2}}}

The file names below are hyper-linked to the corresponding files
in the following specific release of the Coq library of undecidability proofs:
$$\text{\color{ACMDarkBlue}\expandafter\url\expandafter{\undecroot{tree}}}$$
To help at understanding the \Hten\ code from a high-level point of view, 
we provide two additional Coq source code files:
\begin{itemize}
\item the file \fileref{H10}{summary.v} which gives specific pointers
      to the main problems contained in the reduction chain described in this paper;
\item the file \fileref{H10}{standalone.v} which gives a (combined)
      reduction from the binary \M{BPCP} to \Hten\ but where the statements of 
      those two problems have been rewritten to minimize dependencies.
\end{itemize}

Additionally, we here give a detailed overview of the structure of the code corresponding 
to the results presented in this paper, and which was contributed to our
\href{https://github.com/uds-psl/coq-library-undecidability}{Coq library of undecidability proofs}.
The following \emph{lines of code (loc)} measurements combine both definitions
and proof scripts but do not account for comments. Notice that there are more
files in the whole library than those needed to actually cover \Hten, but here, we 
only present the latter.
In total, we contribute 16k loc to our undecidability project, 4k being additions 
to its shared libraries as extensions of the Coq standard library.

Concerning the multi-purpose shared libraries in \dirref{Shared/Libs/DLW/Utils}:
\begin{itemize}
\item we implemented finitary sums/products (over monoids) up to the binomial
      theorem (Newton) over non-commutative rings in \fileref{Shared/Libs/DLW/Utils}{sums.v} and
      \fileref{Shared/Libs/DLW/Utils}{binomial.v} for a total of 550 loc;  
\item we implemented bitwise operations over $\nat$, both a lists of bits 
      in \fileref{Shared/Libs/DLW/Utils}{bool\_list.v} and Peano \cstfont{nat}
      in \fileref{Shared/Libs/DLW/Utils}{bool\_nat.v} for a total of 1700 loc;
\item we implemented many results about Euclidean division and Bézout's identity
      in \fileref{Shared/Libs/DLW/Utils}{gcd.v}, prime numbers
      and their unboundedness in  \fileref{Shared/Libs/DLW/Utils}{prime.v}, 
      and base $p$ representations in \fileref{Shared/Libs/DLW/Utils}{power\_decomp.v}
      for a total of 1200 loc;
\item we implemented miscellaneous libraries for the
      reification of 
      \fileref{Shared/Libs/DLW/Utils}{bounded\_quantification.v}
      (120 loc), the Pigeon Hole Principle in \fileref{Shared/Libs/DLW/Utils}{php.v}
      (350 loc) and iterations of binary relations in \fileref{Shared/Libs/DLW/Utils}{rel\_iter.v}
      (230 loc).
\end{itemize}
Concerning the libraries for Minsky machines and \FRACTRAN programs:
\begin{itemize}
\item by a slight update to the existing code~\cite{FLW19}, we proved in \fileref{MinskyMachines/MM}{mm\_comp.v}
      that \M{MM}-termination (on any state) is undecidable (10 loc). Both the
      pre-existing result (undecidability of \M{MM}-termination on the zero state) and the new result
      derive from the correctness of the compiler of binary stack machines into Minsky machines;
\item we implemented the removal of self loops in Minsky machines in
      \fileref{MinskyMachines/MM}{mm\_no\_self.v} (340 loc);
\item we construct two infinite sequences of primes $\pprime i$ and $\qprime i$
      in \fileref{FRACTRAN/Utils}{prime\_seq.v} (240 loc);
\item \FRACTRAN definitions and basic results occur in \fileref{FRACTRAN/FRACTRAN}{fractran\_utils.v}
      (310 loc) and the verified compiler from Minsky machines to \FRACTRAN
      occurs in  \fileref{FRACTRAN/Reductions}{MM\_FRACTRAN.v} (300 loc);
\end{itemize}
Concerning the libraries for proving Matiyasevich's theorems:
\begin{itemize}
\item we implemented a library for modular arithmetic ($\mathbb Z/p\mathbb Z$)
      in \fileref{H10/ArithLibs}{Zp.v} (920 loc);
\item we implemented a library for $2\times 2$-matrix computation including 
      exponentiation  and determinants 
      in \fileref{H10/ArithLibs}{matrix.v} (210 loc);
\item we implemented an elementary proof of Lucas's theorem in \fileref{H10/ArithLibs}{luca.v}
      (330 loc);
\item we implemented the ``classical proof'' of Lagrange's theorem in \fileref{H10/ArithLibs}{lagrange.v}
      (520 loc);
\item the solution $\alpha_b(n)$ of Pell's equation and its (modular) arithmetic properties
      up to a proof of its Diophantineness are in \fileref{H10/Matija}{alpha.v} (1150 loc);
\item from $\alpha_b(n)$, we implement the meta-level Diophantine encoding of the exponential 
      in \fileref{H10/Matija}{expo\_diophantine.v} (150 loc);
\item we implement the sparse ciphers used in the Diophantine elimination of bounded universal quantification 
      in \fileref{H10/Matija}{cipher.v} (1450 loc).
\end{itemize}
Concerning the object-level Diophantine libraries:
\begin{itemize}
\item the definition of Diophantine logic and basic results is in 
      \fileref{H10/Dio}{dio\_logic.v} (540 loc);
\item the definition of elementary Diophantine constraints and the reduction
      from Diophantine logic is in \fileref{H10/Dio}{dio\_elem.v} (440 loc);
\item the definition of single Diophantine equations and the reduction from
      elementary Diophantine constraints is in \fileref{H10/Dio}{dio\_single.v} (350 loc);
\item we implement the object-level Diophantine encoding of the exponential relation 
      in \fileref{H10/Dio}{dio\_expo.v} (60 loc); but all the work is done
      in the previously mentioned libraries;
\item the object-level Diophantine encoding of bounded universal quantification
      spans over \fileref{H10/Dio}{dio\_binary.v},  \fileref{H10/Dio}{dio\_cipher.v}
      and \fileref{H10/Dio}{dio\_bounded.v} (460 loc);
\item we derive the object-level Diophantine encoding of the reflexive-transitive closure 
      in \fileref{H10/Dio}{dio\_rt\_closure.v} (40 loc);
\item we implement the object-level Diophantine encoding of the \FRACTRAN termination predicate
      in \fileref{H10/Fractran}{fractran\_dio.v} (80 loc).
\end{itemize}
Concerning $\mu$-recursive algorithms, reducing \Hten\ and reduced to Minsky machines:  
\begin{itemize}
\item  building on the pre-existing developments corresponding to~\cite{DLW17} (1000 loc);
\item the $\mu$-recursive algorithm that searches for a solution to a
      single Diophantine equation in \fileref{MuRec}{ra\_utils.v}, 
      \fileref{MuRec}{ra\_dio\_poly.v}, \fileref{MuRec}{recomp.v} and
      \fileref{MuRec}{ra\_recomp.v} (1370 loc);
\item extensions of the Minsky machines library for $\nat$-indexed registers in \fileref{MinskyMachines/MMenv}{env.v},
      \fileref{MinskyMachines/MMenv}{mme\_defs.v} and \fileref{MinskyMachines/MMenv}{mme\_utils.v} (600 loc);
\item the certified compiler from $\mu$-recursive algorithms to Minsky machines 
      in \fileref{MuRec}{ra\_mm.v} and \fileref{MuRec}{ra\_mm\_env.v} (1250 loc).
\end{itemize}
Concerning the reduction from $\mu$-recursive algorithms to $\M{L}$:
\begin{itemize}
\item the step-indexed evaluator and the reduction are in \fileref{L/Reductions}{MuRec.v} (450 loc)
\item the framework used to extract the step-indexed evaluator is in the directory \texttt{L} and consists of about 2500 lines of code.
  More details can be found in~\cite{forster_et_al:LIPIcs:2019:11072}.
\end{itemize}
To finish, the main undecidability results and the DPRM:
\begin{itemize}
\item the undecidability of Minsky machines is in \fileref{MinskyMachines}{MM\_undec.v} (20 loc);
\item the reduction from \M{MM} to \FRACTRAN is in \fileref{FRACTRAN/Reductions}{MM\_FRACTRAN.v} (90 loc);
\item the Diophantine encoding of \FRACTRAN termination is in 
      \fileref{H10}{FRACTRAN\_DIO.v} (70 loc);
\item the whole reduction chain leading to the undecidability 
      of \Hten\ is in \fileref{H10}{H10\_undec.v} (60 loc);
\item the reduction from \Hten to \Htenz is in \fileref{H10}{H10Z.v} (200 loc)
\item and the DPRM theorem is in \fileref{H10}{DPRM.v} (170 loc).
\end{itemize}

\section{Lucas's theorem}

\label{append:lucas} 

\newcommand{\Binomial}[2]{\left(\begin{array}{@{\,}c@{\,}}{#1}\\{#2}\end{array}\right)}

Lucas's theorem allows for the computation of the binomial coefficient $\binomial m n$ modulo 
a prime number $p$ using the base $p$ expansions of $m$ and $n$.
There are various proofs of this theorem but most of them involve high-level concepts 
like \emph{generating functions} or \emph{group action} and we choose instead to implement
a low-level \href{https://math.stackexchange.com/questions/1463758/proof-of-lucas-theorem-without-the-polynomial-hint}{combinatorial proof} 
of the theorem. Such a proof provides specific combinatorial insights into Lucas's results.

While a high-level proof could rightfully be considered as more beautiful,
it would also likely require a lot of additional library code. 
We could of course rely on an external library such as MathComp~\cite{mathcomp} but this would generate
a new external dependency, additional to the already involved Equations~\cite{equations} 
and MetaCoq~\cite{metacoq} libraries, increasing the likelihood of synchronization issues between 
all these developments. An additional benefit of a low-level proof is that it could reasonably
be imported in alternate developments (e.g.~\cite{bayer_et_al:LIPIcs:2019:11088}) in other proof 
assistants (e.g.\ Isabelle/HOL) without the need to convert a substantial helper library.

\smallskip

Before we enter this low-level proof of Lucas's theorem, 
we must give a light-weight, working and formal definition of binomial coefficients.
For this, we use Pascal's identity as a ground for a fixpoint definition:
$$
\binomial m 0 \defeq 1
\qquad\text{and}\qquad
\binomial 0{1+n} \defeq 0
\qquad\text{and}\qquad
\binomial {1+m}{1+n} \defeq \binomial mn+\binomial m{1+n}
$$
where we use the compact notation $\binomial m n$ for a more compact
typesetting of the upcoming equations.
Starting from Pascal's definition,
one can derive the following identities
$$n!\,(m-n)!\,\binomial m n = m! \quad \text{for $n\leq m$}
\qquad\text{and}\qquad
\binomial m n = 0 \quad \text{for $m<n$}$$
and we call the leftmost one as \emph{the binomial identity}\rlap.\footnote{We do not
enter the
details  of these inductive proofs which are standard exercises in basic arithmetic.}

\setCoqFilename{H10.ArithLibs.luca}
\begin{lemma}[][lucas_lemma]
\label{lem:luca}
Let\/ $p$ be a prime number and let us consider\/ $M=mp+m_0$ and\/ $N =np+n_0$ with\/ $m_0,n_0 < p$.
Then the identity\/ $\modeq{\binomial MN}{\binomial{m}{n}\,\binomial{m_0}{n_0}}p$ holds.
\end{lemma}

\newcommand{\mphi}[2]{\phi^{#1}_{#2}}
\newcommand{\mpsi}[1]{\psi_{#1}}
\newcommand{\mphipsi}[2]{{#1}!\,p^{#1}\,\mphi {#1}{#2}\,\mpsi {#1}}
\newcommand{\mphipsip}[2]{{(#1)}!\,p^{#1}\,\mphi {#1}{#2}\,\mpsi {#1}}
\newcommand{\msimple}[2]{\mphi {#1}{#2}\,\mpsi {#1}}

\begin{proof}
Let us first notice that the identity  is trivial (in $\nat$ already) when $n > m$ or $n = m\lconj n_0 > m_0$
because both sides evaluate to $0$.
So below, we only consider the cases  
$n \leq m\lconj n_0 \leq m_0$
or
$n < m \lconj m_0 < n_0$ which
cover the remaining part of the domain. 

Before we split those two cases, 
let us define  $\phi:\nat\to\nat\to\nat$ and $\psi:\nat\to\nat$ by
$$\mphi i r\defeq (ip+1)\cdots(ip+r)
\qquad\text{and}\qquad
\mpsi i\defeq\mphi 0{p-1}\cdots\mphi {i-1}{p-1}.$$
Notice that $\mphi i r$ generalize the factorial function
as $r! = \mphi 0 r$ but we will always use $\mphi i r$ with $r < p$.
Projected on $\Zp p$, both $\phi$ and $\psi$ simplify to factorial
as it is easy to show\footnote{Notice that Wilson's theorem
establishes $\modeq{(p-1)!}{-1}p$ but this equation is not needed in this proof.}\enspace
\begin{equation}
\label{eq:luca:2}
\modeq{\mphi i r}{r!}p
\qquad\text{and}\qquad
\modeq{\mpsi i}{(p-1)!^i}p.
\end{equation}
Moreover, $\mphi ir$ is invertible in $\Zp p$ for any $r<p$, and
in particular $(p-1)!=\mphi 0{p-1}$ and $\mpsi i$ are
also invertible in $\Zp p$. Using $\phi$, $\psi$ and
the semiring structure of $\nat$, 
we establish the identity
\begin{equation}
\label{eq:luca:1}
(ip+r)!=\mphi 0{p-1}\cdot(1p+0)\,\cdots\,\mphi {i-1}{p-1}\cdot(ip+0)\cdot\mphi ir
=\mphipsi ir 
\end{equation}

\medskip

In the case $n \leq m$ and $n_0 \leq m_0$, we have $M-N = (m-n)p+(m_0-n_0)$ 
and we rewrite the binomial equation $N!\,(M-N)!\,\binomial M N = M!$ in $\nat$ as
$$  (np+n_0)!\,\bigl((m-n)p+(m_0-n_0)\bigr)!\,\binomial M N = (mp+m_0)!$$ 
and then rewrite again with Eq.~\eqref{eq:luca:1} and the binomial equation
$m! = n!\,(m-n)!\,\binomial m n$ as
$$      
   \mphipsi n{n_0}\cdot\mphipsip {m-n}{m_0-n_0}\cdot\binomial M N 
=  n!\,(m-n)!\,\binomial m n\cdot p^m\,\mphi m{m_0}\,\mpsi m$$
In $\nat$, we simplify by $p^n p^{m-n}=p^m$ and 
$n!(m-n)!$ and get
$$      \msimple n{n_0}\,\msimple {m-n}{m_0-n_0}\,\binomial M N 
      = \binomial m n\,\msimple m{m_0}$$
Now switching to $\Zp p$ using Eqs.~\eqref{eq:luca:2},
we derive:
$$ \modeq{n_0!\,(p-1)!^n\,(m_0-n_0)!\,(p-1)!^{m-n}\,\binomial M N} 
         {\binomial m n\,m_0!\,(p-1)!^m}p$$
Because $(p-1)!$ is invertible in $\Zp p$, we can simplify by 
$(p-1)!^n(p-1)!^{m-n}=(p-1)!^m$. Then we
rewrite using the binomial equation 
$m_0! = n_0!\,(m_0-n_0)!\,\binomial {m_0}{n_0}$
and deduce
$$ \modeq{n_0!\,(m_0-n_0)!\,\binomial M N} 
        {\binomial m n\,m_0!\equiv \binomial m n\,n_0!\,(m_0-n_0)!\,\binomial {m_0}{n_0}}p$$
Finally, as both $n_0<p$ and $m_0-n_0<p$ hold, then $n_0!\,(m_0-n_0)!$ is invertible
in $\Zp p$ and, simplifying, we get 
$\modeq{\binomial M N}{\binomial m n\binomial {m_0}{n_0}}p$ as required.

\medskip

Then we consider the alternative case where $n < m$ and $m_0 < n_0 < r$. In this case
we have $M-N = \bigl(m-(n+1)\bigr)p+\bigl(p-(n_0-m_0)\bigr)$ and again we develop
the binomial equation $N!\,(M-N)!\,\binomial M N = M!$ in $\nat$ as
$$  (np+n_0)!\,\bigl((m-(n+1))p+(p-(n_0-m_0))\bigr)!\,\binomial M N = (mp+m_0)!$$
and thus, using Eq.~\eqref{eq:luca:1} and the binomial equation
$m! = n!\,(m-n)!\,\binomial m n$, rewriting it into 
$$      
   \mphipsi n{n_0}\mphipsip {m-(n+1)}{p-(n_0-m_0)}\binomial M N 
=  n!\,(m-n)!\,\binomial m n\,p^m\,\mphi m{m_0}\,\mpsi m$$
that we then simplify by $n!\,(m-(n+1))!$ and $p^n\,p^{m-(n+1)} = p^{m-1}$ to get
$$      \msimple n{n_0}\,\msimple {m-(n+1)}{p-(n_0-m_0)}\,\binomial M N 
      = (m-n)\,\binomial m n\,p\,\msimple m{m_0}$$
Notice that $p$ appears (at least once) on the right-hand side, thus 
switching to $\Zp p$, we get:
$$ \modeq{n_0!\,(p-1)!^n\,(p-(n_0-m_0))!\,(p-1)!^{m-(n+1)}\,\binomial M N} 
         {0}p$$
Now, as $n_0<p$ and $p-(n_0-m_0)<p$,
the factorials $n_0!$ and $(p-(n_0-m_0))!$ are both invertible in $\Zp p$.
This is also the case of any power of $(p-1)!$. We deduce that $\modeq{\binomial M N}{0}p$
which leads us to conclude $\modeq{\binomial M N}{\binomial m n\binomial {m_0}{n_0}}p$
as $\binomial {m_0}{n_0}=0$ trivially follows from $m_0<n_0$.
\end{proof}

\begin{theorem}[Lucas~\cite{lucas1878}][lucas_theorem]
\label{thm:luca}
Whenever\/ $u_0,\ldots u_n,v_1,\ldots,v_n < p$, the following identity holds:
$$\modeq{\Binomial {u_np^n+\cdots+u_0}{v_np^n+\cdots+v_0}}{\Binomial{u_n}{v_n}\cdots\Binomial{u_0}{v_0}}p.$$
\end{theorem}

\begin{proof}
By induction on the number $n$ of coefficients using Lemma~\ref{lem:luca}.
\end{proof}

\setCoqFilename{Shared.Libs.DLW.Utils.bool_nat}
\begin{definition}[][binary_le]
Let\/ $a=a_n2^n+\cdots+a_0$ and\/ $b=b_n2^n+\cdots+b_0$ representations
in base\/ $2$. We denote\/ $a\bwleq b$ if\/ $a_i\leq b_i$ holds for
any\/ $i\leq n$. 
\end{definition}

Notice that the definition of $a\bwleq b$ is irrelevant to which
base $2$ representations of $a$ and $b$ are picked up, i.e.\ it is
not influenced by trailing zeros that might appear in front of (the
representations of) $a$ or $b$.

\setCoqFilename{H10.Dio.dio_binary}
\begin{corollary}[][binary_le_binomial]
For any $a,b:\nat$ we have\/ $a\bwleq b\toot \modeq{\binomial b a}12$. 
\end{corollary}

\section{Avoiding Overflows\texorpdfstring{ in the Proof of~\Cref{coq:dio_rel_fall_lt}}{}}

\label{overflow_appendix}

The section explains why we slightly modified the original proof of the
elimination of bounded universal quantification~\cite{Matiyasevich1997}
to avoid overflows when dealing with parallel multiplications, in search for
a Diophantine encoding of the simultaneous equations
$a_1=b_1c_1,\ldots, a_n=b_nc_n$, which can also be viewed
as the identity between the vectors $\vec a = [a_1;\ldots;a_n]$
and the scalar/dot product $\vec b\cdot\vec c = [b_1c_1;\ldots;b_nc_n]$.

Following Matiyasevich~\cite{Matiyasevich1997}, the component of those vectors
are encoded sparsely within the digits in base $r$ of the corresponding so called \emph{ciphers} (see below for examples). 
Using the natural product of ciphers, he can somehow recover the identity between
the components of $\vec a$ and those of $\vec b\cdot\vec c$. The basis $r$ for
the encoding is chosen large enough w.r.t.\ the component of $\vec a$, $\vec b$ and $\vec c$ 
to avoid digit overflow during the natural product of ciphers. However, this product 
generates artefacts on in-between digits that need 
to be masked out and, with Matiyasevich's choice for $r$ as $r\cdef 2^{2q}$, 
some overflow could occur within these artefacts, a subtlety which we
speculate he might have missed. This overflow pointlessly complicates the 
correctness proof of masking. In our Coq code, we propose to 
avoid any digit overflow by increasing $r$ to $r\cdef 2^{4q}$, but, as we realized
later on, any power of $2$ greater than $2^{2q+1}$ would work as well.

\newcommand{\Ir}{D_q}

Let us switch to the technical details explaining how the overflow occurs
and how it can be avoided, with little impact on the rest of the argument.
For the moment, we describe what Matiyasevich
does in~\cite{Matiyasevich1997} so we stick with his definition of $r$ as $r\cdef 2^{2q}$
and consider the encoding of vectors of $\Ir^n$
as ciphers, written in base $r$, where $q\geq 1$ and 
$\Ir\cdef \{0,\ldots,2^q-1\}$ is the set of allowed digits for vectors.
He considers five vectors and their respective ciphers:
\[\begin{array}{r@{\,=\,}lcr@{\,=\,}l}
\multicolumn{2}{c}{\text{vector/length}} &\qquad & \multicolumn{2}{c}{\text{cipher}}\\[1ex]
\vec u  & [1;\ldots;1]\in \Ir^n        & & u  & 1.r^{2^1}+\cdots +1.r^{2^n} \\
{\vec u}\,' & [0;1;\ldots;1]\in \Ir^{n+1}  & & u' & 0.r^{2^1}+1.r^{2^2}+\cdots +1.r^{2^{n+1}} \\
\vec a  & [a_1;\ldots;a_n]\in \Ir^n    & & a  & a_1r^{2^1}+\cdots +a_nr^{2^n}\\
\vec b  & [b_1;\ldots;b_n]\in \Ir^n    & & b  & b_1r^{2^1}+\cdots +b_nr^{2^n}\\
\vec c  & [c_1;\ldots;c_n]\in \Ir^n    & & c  & c_1r^{2^1}+\cdots +c_nr^{2^n}.
\end{array}\]
Each summand is an actual digit of the cipher in base $r$. Notice 
that $q$ is to be chosen large enough so that the digits in $\Ir$ cover the 
components of the vectors $\vec a$, $\vec b$ and $\vec c$.  
Also, the ciphers have a sparse representation in base $r$, i.e.\ only the
digits which correspond to powers $r^{2^i}$ are possibly non-zero.
Finally, most of the digits of base $r$ are unused, i.e.\ only $2^q$ of them out of $r=2^{2q}$ many,
to avoid overflows when adding or multiplying digits. But remark that
which overflows are avoided (or not) is the purpose of the whole discussion here.

The problem solved in~\cite{Matiyasevich1997} is to find, depending on $q/r$ and $n$, 
Diophantine representations for (the ciphers of) $\vec u$ and $\vec u\,'$, and most importantly, for 
the identities $\vec a = \vec b+\vec c$ (the sum) and $\vec a = \vec b\cdot\vec c$ (the scalar product),
implementing some kind of parallel addition and multiplication. We skip over the 
Diophantine encodings of $\vec u$ and $\vec u\,'$, of which the description spans about one fourth 
of~\cite{Matiyasevich1997}, to directly consider the parallel sum and multiplication.
For $\vec a = \vec b+\vec c$, the encoding is as straightforward as the addition of ciphers $a=b+c$,
because the basis $r$ is large enough to avoid overflows of digit additions, i.e.\ 
$2(2^q-1) < r=2^{2q}$, and no artefacts appear at powers different from $r^{2^i}$.

The case of parallel multiplication is more complicated.
To recover the identity between $\vec a$ and the scalar product $\vec b\cdot\vec c =[b_1c_1;\ldots;b_nc_n]$,
Matiyasevich also uses the natural product of the ciphers but it generates artefacts on in-between digits 
that need to be masked out. Indeed, considering the identities
\begin{equation}
\label{appC_eq_au}
au=\sum_{i=1}^na_ir^{2^i}\times\sum_{i=1}^nr^{2^i} =
\sum_{i=1}^{n}a_ir^{2^{i+1}}+\!\!\!\!
 \sum_{1\leq i < j\leq n}\!\!\!\!\!(a_i+a_j)r^{2^i+2^j}
\end{equation}
and
\begin{equation}
\label{appC_eq_bc}
bc=\sum_{i=1}^nb_ir^{2^i}\times\sum_{i=1}^nc_ir^{2^i} =
\sum_{i=1}^{n}b_ic_ir^{2^{i+1}}+\!\!\!\!
 \sum_{1\leq i < j\leq n}\!\!\!\!\!(b_ic_j+b_jc_i)r^{2^i+2^j}
\end{equation}
by filtering out the two parts $\sum_{1\leq i<j\leq n}\ldots{}$
on the right hand side of the $+$ sign,
and then identifying the two values,
we may recover the identity $\vec a= \vec b\cdot\vec c$.

For this, he uses $(r-1)u' = (r-1)r^{2^2}+\cdots +(r-1)r^{2^{n+1}}$
as a mask.
This is the meaning of Equation~(40) of page~3232 in~\cite{Matiyasevich1997},
(slightly) reformulated here as
\begin{equation}
\label{appC_eq_40}
au \mathbin{\&} (r-1)u' = bc \mathbin{\&} (r-1)u'
\end{equation}
where $\&$ is the digit by digit A\!N\!D operator\rlap.\footnote{The Diophantine encoding of the A\!N\!D operator is itself derived from  Lucas's theorem.}\enspace
However, to show the equivalence between 
Equation~\eqref{appC_eq_40} and the identity $\vec a= \vec b\cdot\vec c$,
we have to make sure that the masking operator $x\mapsto x\mathbin{\&} (r-1)u'$
 effectively erases the two rightmost sums $\sum_{1\leq i<j\leq n}\ldots{}$
in Equations~\eqref{appC_eq_au} and~\eqref{appC_eq_bc}, while keeping the left sums 
intact. Proving this property is not even considered worthy of an explanation 
in~\cite{Matiyasevich1997} but it turned out to be not that obvious.

We first remark that \emph{the digit by digit A\!N\!D operator $\mathbin{\&}$ does not commute with $+$} (in general), hence
computing the mask can be tricky. However, $\mathbin{\&}$ commutes with the
digit by digit O\!R operator. Moreover O\!R and $+$ coincide whenever the binary digits
of the added numbers do not overlap (there are never two $1$ on the same binary digit power
in each number). As $r$ is a power of $2$, this generalises from base $2$ to base $r$, and 
provided the below summations are base $r$ representations, one can commute the sum/$+$ with $\mathbin{\&}$ and apply
masking component by component, e.g.\ one can show the identity
\[ \left({\,\sum_{i\in I} x_ir^{\varphi(i)}}\right)  \mathbin{\&} \left({\,\sum_{i\in I} m_ir^{\varphi(i)}}\right) = \sum_{i\in I} (x_i \mathbin{\&} m_i)r^{\varphi(i)}\]
as soon as $x_i<r$ and $m_i<r$ holds for any $i\in I$, and $\varphi:I\to\nat$ \emph{injectively} maps the finite type $I$.
But this implies that we have to check that Equations ~\eqref{appC_eq_au} and~\eqref{appC_eq_bc}
contain such genuine base $r$ representations.

In that spirit, if $i< j$ holds then the powers $r^{2^k}$ and $r^{2^i+2^j}$ are never 
the same power of $r$. So these represent two distinct digit positions in base $r$
and hence, in Equations~\eqref{appC_eq_au} and~\eqref{appC_eq_bc}, those digits from 
the left of $+$, and those from the right of $+$ do not overlap/interfere.
This also motivates the use of the mask  $(r-1)u' = (r-1)r^{2^2}+\cdots +(r-1)r^{2^{n+1}}$ 
which indeed filters out digits at any $r^{2^i+2^j}$ position while leaving those
at $r^{2^i}$ intact.

The way Equations~\eqref{appC_eq_au} and~\eqref{appC_eq_bc} are written
unfortunately suggests that these are base $r$ representations of both $au$
and $bc$. However, formally, we have to verify that the two rightmost sums 
$\sum_{1\leq i<j\leq n}\ldots{}$ are proper base $r$ representations. 
This is not the case for~\eqref{appC_eq_bc} and we here give a counter example.
On the one hand, it is true that $2^i+2^j=2^{i'}+2^{j'}$ implies $(i,j)=(i',j')$
when $i<j$ and $i'<j'$, ensuring that digits at those powers of $r$ do not accidentally overlap.
On the other hand, it possible that the value  $b_ic_j+b_jc_i$ overflows the digit range 
of base $r$, even though considerable room was initially reserved to avoid that situation. 
For instance, when $b_i=c_j=b_j=c_i=2^q-1$ (the maximum allowed
digit in $\Ir$), then we get $b_ic_j+b_jc_i=2^{2q+1}-2^{q+2}+2\geq r+2$ 
as soon as $q>1$. So this component at digit position $r^{2^i+2^j}$ overflows
and we would have to consider the spill out at digit position $r^{2^i+2^j+1}$
and show that it does not overlap with the other digits or other spill
outs. This is a property which we think holds true but the proof of it
would add significant formal complexity.

On the other hand, with our proposed alternative choice of $r\cdef 2^{4q}$, 
keeping the same set $\Ir=\{0,\ldots,2^q-1\}$ for
the allowed digits in vectors, the overflow does not occur any more, i.e.\ no spill out,
and Equations~\eqref{appC_eq_au} and~\eqref{appC_eq_bc} become
genuine base $r$ representations, allowing smooth component by component masking.

Because there is an obvious way out of a tricky overflow management
problem, we think that it is possible that Matiyasevich simply 
did not notice the eventuality of an overflow in the to be masked parts
of Equation~\eqref{appC_eq_bc}. We speculate this
because this overflow does not occur in the parts which are left intact by the mask, 
at the left of the $+$ sign.
With this remark, we do not imply that Equation~(40) of~\cite{Matiyasevich1997} 
is improper in any way. However, the formal proof of its equivalence with
the identity $\vec a= \vec b\cdot\vec c$ is really more complicated 
when overflows occur and we did not try to prove it in the case where $r\cdef 2^{2q}$, 
this situation 
being straightforward to avoid with a change (of otherwise low impact) in the value of $r$.

\section{Proof of Proposition~\ref{convexity2}}

\label{app_convexity2}

Let\/ $(p_1,q_1),\ldots,(p_n,q_n)$ be a sequence of pairs in\/ $\nat\times\nat$. We establish
the equivalence:
\[\hfil \sum_{i=1}^n 2p_iq_i = \sum_{i=1}^n p_i^2+q_i^2~\toot~p_1=q_1\lconj\cdots\lconj p_n=q_n.\]

\begin{proof}
We give an elementary arithmetic justification of the result, 
proof which involves none of the high-level tools of mathematical analysis.
We first show the two following statements
\begin{equation}
\label{convexity_eq}
2ab \leq a^2+b^2\quad\text{and}\quad 2ab = a^2+b^2\toot a=b
\qquad\text{for any\/ $a,b:\nat$}
\end{equation}
Assuming w.l.o.g.\ that $a\leq b$, we can write $b=a+\delta$ with $\delta\in\nat$ and
then, for ${\grel}\in\{{\leq},{=}\}$ we have 
$2ab \grel a^2+b^2\toot 2a^2+2a\delta \grel a^2+a^2+2a\delta+\delta^2\toot 0\grel \delta^2$
hence the desired result.

From the left inequality~\eqref{convexity_eq}, we easily generalize by induction on $n$
and obtain the following inequality:
\begin{equation}
\textstyle
\label{convexity_eq1}
\sum_{i=1}^n 2p_iq_i \leq \sum_{i=1}^n p_i^2+q_i^2
\end{equation}

Now we can proceed with the proof of the main stated equivalence.
The \emph{if case} is obvious so we only describe the \emph{only if case}. 
Hence we show that
$\sum_{i=1}^n 2p_iq_i = \sum_{i=1}^n p_i^2+q_i^2$ entails $p_i=q_i$ for all $i\in[1,n]$.
We proceed by induction on $n$ again.
The base case $n=0$ is trivial. For the inductive step $1+n$, let us assume
\begin{equation}
\textstyle
\label{convexity_eq2}
\sum_{i=1}^n 2p_iq_i +2p_{n+1}q_{n+1}=\sum_{i=1}^n (p_i^2+q_i^2) + p_{n+1}^2+q_{n+1}^2.
\end{equation}
By the left inequality of~\eqref{convexity_eq} and inequality~\eqref{convexity_eq1},
we have both $2p_{n+1}q_{n+1}\leq p_{n+1}^2+q_{n+1}^2$
and $\sum_{i=1}^n 2p_iq_i \leq \sum_{i=1}^n p_i^2+q_i^2$. 
The only possibility for the identity in hypothesis~\eqref{convexity_eq2} to hold 
is that both inequalities are in fact identities,
hence both $2p_{n+1}q_{n+1}= p_{n+1}^2+q_{n+1}^2$ and
$\sum_{i=1}^n 2p_iq_i = \sum_{i=1}^n p_i^2+q_i^2$ hold.
From this we derive $p_{n+1} = q_{n+1}$ by the  
equivalence on the right of~\eqref{convexity_eq} and $p_1=q_1, \ldots, p_n=q_n$
by the induction hypothesis.
\end{proof}

\end{document}